\documentclass[runningheads]{llncs}
\usepackage[paperheight=235mm,paperwidth=155mm,textwidth=12.2cm,textheight=19.3cm]{geometry}
\usepackage{bbding}
\usepackage{graphicx}
\makeatletter
\RequirePackage[bookmarks,unicode,colorlinks=true]{hyperref}%
   \def\@citecolor{blue}%
   \def\@urlcolor{blue}%
   \def\@linkcolor{blue}%

\def\orcidID#1{\smash{\href{http://orcid.org/#1}{\protect\raisebox{-1.25pt}{\protect\includegraphics{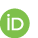}}}}}
\makeatother
\usepackage{adjustbox}
\usepackage{amsmath,amssymb,amsfonts}
\usepackage{mathtools}
\usepackage{enumerate}
\usepackage{xcolor}
\usepackage{xfrac}
\usepackage{xspace}
\usepackage[textwidth=29mm,disable]{todonotes}
\usepackage{tensor}
\usepackage{stmaryrd}
\usepackage{cleveref}
\usepackage{multirow}
\usepackage{csquotes}

\usepackage{mdframed}

\newcommand{\cmcomment}[1]{\todo[color=green,size=\scriptsize,fancyline,author=Christoph]{#1}\xspace}

\newenvironment{changed}{}{}%\begin{mdframed}[hidealllines=true,backgroundcolor=red!20]}{\end{mdframed}}
\newcommand{\changedinline}[1]{#1}%\textcolor{red}{#1}}

\newcommand{\cm}[1]{#1}%\textcolor{blue}{#1}}
\newcommand{\cmout}[1]{}%\textcolor{blue!10!white}{#1}}

\newcommand{\slemp}{\sfsymbol{\textbf{emp}}}

\newcommand{\tuplenotation}[1]{(#1)}
\newcommand{\singleton}[2]{\iverson{#1 \mapsto #2}}
\newcommand{\slsingleton}[2]{#1 \mapsto #2}

\newcommand{\validpointer}[1]{\iverson{#1 \mapsto \,{-}\,}}

\newcommand{\sepcon}{\mathbin{{\star}}}

\newcommand{\sepimp}{\mathbin{\text{\raisebox{-0.1ex}{$\boldsymbol{{-}\hspace{-.55ex}{-}}$}}\hspace{-1ex}\text{\raisebox{0.13ex}{\rotatebox{-17}{$\star$}}}}}

\newcommand{\Sup}{\reflectbox{\textnormal{\textsf{\fontfamily{phv}\selectfont S}}}\hspace{.2ex}}
\newcommand{\Inf}{\raisebox{.6\depth}{\rotatebox{-30}{\textnormal{\textsf{\fontfamily{phv}\selectfont \reflectbox{J}}}}\hspace{-.1ex}}}

\newcommand{\emax}[2]{#1 \max #2}
\newcommand{\emin}[2]{#1 \min #2}

\newcommand{\entails}{\models}

\DeclareRobustCommand*{\mirrormodels}{%
	\Relbar\joinrel\mathrel{|}%
}

\newcommand{\mirrorentails}{\mirrormodels}

\newcommand{\QSL}{\sfsymbol{QSL}\xspace}
\newcommand{\SL}{\sfsymbol{SL}\xspace}
\newcommand{\QSLfrag}{\sfsymbol{Q}\xspace}
\newcommand{\SLfrag}{\sfsymbol{S}\xspace}

\newcommand{\QSLpar}[1]{\QSL\left[ #1 \right]\xspace}

\newcommand{\SLpar}[1]{\SL\left[ #1\right]\xspace}

\newcommand{\QSLA}{\QSLpar{\Predset}}
\newcommand{\SLA}{\SLpar{\Predset}}

\newcommand{\QSH}{\sfsymbol{QSH}\xspace}
\newcommand{\eQSH}{\sfsymbol{eQSH}\xspace}

\newcommand{\atleastsymbol}{\preceq}
\newcommand{\atleast}[2]{\lceil #1 \atleastsymbol #2 \rceil}

\newcommand{\evaluationSetSymbol}{\sfsymbol{Val}}
\newcommand{\evaluationSet}[1]{\ensuremath{\evaluationSetSymbol\left[#1\right]}}

\newcommand{\NonNegRats}{\Rats_{\geq 0}}
\newcommand{\ra}{\ensuremath{\alpha}}
\newcommand{\rb}{\ensuremath{\beta}}
\newcommand{\rc}{\ensuremath{\gamma}}
\newcommand{\rd}{\ensuremath{\delta}}

\newcommand{\SLRuntime}{\text{SL-Time}}

\newenvironment{restateTheorem}[1]{%
\begingroup
\def\thetheorem{\ref{#1}}
\begin{theorem}}{%
\end{theorem}
\addtocounter{theorem}{-1}
\endgroup}

\newenvironment{restateLemma}[1]{%
\begingroup
\def\thelemma{\ref{#1}}
\begin{lemma}}{%
\end{lemma}
\addtocounter{lemma}{-1}
\endgroup}

\newenvironment{restateDefinition}[1]{%
\begingroup
\def\thedefinition{\ref{#1}}
\begin{definition}}{%
\end{definition}
\addtocounter{definition}{-1}
\endgroup}

\newcommand{\sfsymbol}[1]{\textsf{\upshape {#1}}}
\newcommand{\ttsymbol}[1]{\texttt{\upshape {#1}}}
\newcommand{\wlpsymbol}{\sfsymbol{wlp}}
\newcommand{\boldwlpsymbol}{\textbf{\sfsymbol{wlp}}}
\newcommand{\wlp}[2]{\wlpsymbol\llbracket#1\rrbracket\left(#2\right)}
\newcommand{\wlpC}[1]{\wlpsymbol\llbracket#1\rrbracket}
\newcommand{\conditionalPair}[2]{{\let\oldarraystretch\arraystretch}\renewcommand{\arraystretch}{1}~\holter{~\raisebox{.5ex}{${#1}$}~}{~\raisebox{.125ex}{${#2}$}~}~\renewcommand{\arraystretch}{\oldarraystretch}}

\newcommand{\cc}{\ensuremath{C}} % programs
\newcommand{\guard}{\ensuremath{B}} % conditial guard
\newcommand{\ee}{\ensuremath{E}} % expressions

\newcommand{\preda}{\ensuremath{\Phi}} % predicates

\newcommand{\sla}{\ensuremath{\varphi}} % SL formulae
\newcommand{\slb}{\ensuremath{\psi}} 
\newcommand{\slc}{\ensuremath{\vartheta}}

\newcommand{\hh}{\ensuremath{h}} % heap
\newcommand{\sk}{\ensuremath{s}} % stack 

\newcommand{\pp}{\ensuremath{p}} % probabilities
\newcommand{\qq}{\ensuremath{q}}

\newcommand{\FF}{\ensuremath{X}} % expectations

\newcommand{\INV}{\ensuremath{I}} % invariants

\newcommand{\ff}{\ensuremath{f}} % QSL formulae
\newcommand{\fg}{\ensuremath{g}}
\newcommand{\fh}{\ensuremath{u}}

\newcommand{\xx}{\ensuremath{x}} % program variables
\newcommand{\xy}{\ensuremath{y}} 
\newcommand{\xz}{\ensuremath{z}}

\newcommand{\BB}{\ensuremath{B}} % Pure predicates

\newcommand{\bb}{\ensuremath{\pi}} % Boolean expressions/pure formulae

\newcommand{\recnum}{\ensuremath{\Bbbk}}

\newcommand{\loca}{\ell}%locations
\newcommand{\locseta}{L}

\newcommand{\terma}{\ensuremath{t}} %terms

\newcommand{\SKIP}{\ttsymbol{skip}}

\newcommand{\AssignSymbol}{\mathrel{\textnormal{\texttt{:=}}}}
\newcommand{\ASSIGN}[2]{\ensuremath{#1 \AssignSymbol #2}}

\newcommand{\ALLOC}[2]{\ensuremath{{#1} \AssignSymbol \mathtt{new}\left( #2 \right)}}

\newcommand{\AVAILLOC}[1]{\PosNats}
\usepackage{actuarialangle}

\newcommand{\dereference}[1]{\texttt{<}\,#1\,\texttt{>}}

\newcommand{\HASSIGN}[2]{\ensuremath{\dereference{#1} \AssignSymbol #2}}
\newcommand{\ASSIGNH}[2]{\ensuremath{#1 \AssignSymbol \dereference{#2}}}
\newcommand{\FREE}[1]{\ensuremath{\mathtt{free}(#1)}}

\newcommand{\SEMI}{\ensuremath{\,;\,}}
\newcommand{\COMPOSE}[2]{\ensuremath{{#1}{\,;}~ {#2}}}
\newcommand{\PCHOICE}[3]{\ensuremath{\left\{\, {#1} \,\right\}\mathrel{\left[\,#2\,\right]}\left\{\, {#3} \,\right\}}}

\newcommand{\IFSYMBOL}{\ensuremath{\textnormal{\texttt{if}}}}

\newcommand{\ELSESYMBOL}{\ensuremath{\textnormal{\texttt{else}}}}

\newcommand{\ITE}[3]{\ensuremath{\IFSYMBOL\,\left(\, {#1} \,\right)\,\left\{\, {#2} \,\right\}\,\ELSESYMBOL\,\left\{\, {#3} \,\right\}}}

\newcommand{\WHILESYMBOL}{\ensuremath{\textnormal{\texttt{while}}}}
\newcommand{\WHILE}[1]{\ensuremath{\WHILESYMBOL \left(\, {#1} \,\right)\left\{\right.}}

\newcommand{\WHILEDO}[2]{\ensuremath{\WHILESYMBOL \left(\, {#1} \,\right)\left\{\, {#2} \,\right\}}}

\newcommand{\hpgcl}{\textnormal{\sfsymbol{hpGCL}}\xspace}   % Programs
\newcommand{\Vars}{\ensuremath{\mathsf{Vars}}\xspace}   % Variables
\newcommand{\Vals}{\ensuremath{\mathsf{Vals}}\xspace}    % Values
\newcommand{\Locs}{\ensuremath{\mathsf{Locs}}\xspace}    % Values

\newcommand{\Nats}{\ensuremath{\mathbb{N}}\xspace}
\newcommand{\PosNats}{\ensuremath{\mathbb{N}_{>0}}\xspace}
\newcommand{\Ints}{\ensuremath{\mathbb{Z}}\xspace}
\newcommand{\Rats}{\ensuremath{\mathbb{Q}}\xspace}

\newcommand{\Probs}{\mathbb{P}}

\newcommand{\Preds}{\pot{\States}}

\newcommand{\Eone}{\mathbb{E}_{\leq 1}}

\newcommand{\Predset}{\mathfrak{A}}

\newcommand{\dom}[1]{\sfsymbol{dom}\left({#1}\right)}
\newcommand{\iverson}[1]{\left[ {#1} \right]}

\newcommand{\subst}[2]{\left[ {#1} \texttt{:=} {#2}\right]}
\newcommand{\statesubst}[2]{\left[ {#1} \texttt{:=} {#2}\right]}

\newcommand{\eval}[1]{\ensuremath{\sfsymbol{Eval} \left( #1\right)}}

\newcommand{\pot}[1]{\mathcal{P}\left({#1}\right)}
\newcommand{\sizeof}[1]{|#1|}
\newcommand{\probsizeof}[1]{|#1|_p}

\newcommand{\sem}[1]{\ensuremath{\left\llbracket {#1} \right\rrbracket}}
\newcommand{\semapp}[2]{\ensuremath{\sem{#1} #2}}

\newcommand{\bigo}{\ensuremath{\mathcal{O}}}

\newcommand{\Stacks}{\sfsymbol{Stacks}\xspace}
\newcommand{\Heaps}{\sfsymbol{Heaps}_\recnum\xspace}
\newcommand{\emptyheap}{h_\emptyset}
\newcommand{\disjoint}{\mathrel{\bot}}
\newcommand{\States}{\sfsymbol{States}\xspace}
\newcommand{\To}{\rightarrow}
\newcommand{\true}{\mathsf{true}}
\newcommand{\false}{\mathsf{false}}
\newcommand{\mydot}{\text{{\Large\textbf{.}}~}}

\newcommand{\ie}{i.e.,\xspace}
\newcommand{\eg}{e.g.,\xspace}
\newcommand{\cf}{cf.,\xspace}

\newcommand{\qand}{\quad\textnormal{and}\quad}

\newcommand{\qimplies}{\quad\textnormal{implies}\quad}

\newcommand{\eeq}{~{}={}~}

\newcommand{\lleq}{~{}\leq{}~}

\newcommand{\mmodels}{~{}\models{}~}

\newcommand{\eentails}{~{}\entails{}~}
\newcommand{\iin}{~{}\in{}~}

\newcommand{\setcomp}[2]{\left\{\, {#1} ~\middle|~ {#2} \,\right\}}
\newcounter{computationarrowsone}
\newcounter{computationarrowstwo}

\newcounter{sarrow}

\newcommand{\Update}[2]{\ensuremath{\subst{#1}{#2}}} %[#1 \mapsto #2]}}

\newcommand{\Lssymbol}{\ensuremath{\mathsf{ls}}}
\newcommand{\slLs}[2]{\ensuremath{\Lssymbol\left(#1,#2\right)}}

\newcommand{\cswap}{\cc_{\text{swap}}}
\newcommand{\cpop}{\cc_{\text{populate}}}
\newcommand{\ls}[2]{\textsf{ls}(#1,#2)}
\newcommand{\ils}[2]{\iverson{\ls{#1}{#2}}}

\begin{document}
\title{Foundations for Entailment Checking\\in Quantitative Separation Logic \\ (extended version)\thanks{This work is partially supported by the ERC AdG project 787914 FRAPPANT.}
}
\titlerunning{Foundations for Entailment Checking\\in Quantitative Separation Logic}
\author{%
	Kevin Batz\inst{1}\orcidID{0000-0001-8705-2564} 
	\and 
	Ira Fesefeldt\inst{1}\orcidID{0000-0001-7837-2611} 
	\and  
	Marvin Jansen\inst{1}  
	\and 
	Joost-Pieter Katoen\inst{1}\orcidID{0000-0002-6143-1926} 
	\and \\
	Florian Ke{\ss}ler\inst{1} 
	\and 
	Christoph Matheja\inst{2,3}\orcidID{0000-0001-9151-0441}
	\and 
	Thomas Noll\inst{1}\orcidID{0000-0002-1865-1798}
}
%\author{Anonymous Authors}
	%\and
%Second Author\inst{2,3}\orcidID{1111-2222-3333-4444} \and
%Third Author\inst{3}\orcidID{2222--3333-4444-5555}}
%
\authorrunning{Batz \emph{et al.}}
%\authorrunning{Anon.}
\institute{Software Modeling and Verification Group, RWTH Aachen University, Germany\\
	\email{\{kevin.batz,fesefeldt,katoen,noll\}@cs.rwth-aachen.de}
	\and
	Programming Methodology Group, ETH Zürich, Switzerland\\
	\and
	%Section for Formal Methods, 
	Technical University of Denmark, 
	\email{chmat@dtu.dk}
	}
%	, USA \and
%Springer Heidelberg, Tiergartenstr. 17, 69121 Heidelberg, Germany
\maketitle              % typeset the header of the contribution
\begin{abstract}
Quantitative separation logic ($\QSL$) is an extension of separation logic ($\SL$) for the verification of probabilistic pointer programs. In $\QSL$, formulae evaluate to real numbers instead of truth values, e.g., the probability of memory-safe termination in a given symbolic heap. As with $\SL$, one of the key problems when reasoning with $\QSL$ is \emph{entailment}: does a formula $\ff$ entail another formula $\fg$?

 We give a generic reduction from entailment checking in $\QSL$ to entailment checking in $\SL$. This allows to leverage the large body of $\SL$ research for the automated verification of probabilistic pointer programs. We analyze the complexity of our approach and demonstrate its applicability. In particular, we obtain the first decidability results for the verification of such programs by applying our reduction to a quantitative extension of the well-known symbolic-heap fragment \mbox{of separation logic.}
	%

%\keywords{QSL  \and is \and awesome}
\end{abstract}
\section{Introduction}
% SL
\emph{Separation logic}~\cite{Ishtiaq2001BI} (\SL) is a popular formalism for Hoare-style verification of
imperative, heap-manipulating and, possibly, concurrent programs.
Its assertion language extends first-order logic with two connectives---the separating
conjunction $\sepcon$ and the magic wand $\sepimp$---that enable concise specifications
of how program memory, or other resources, can be split-up and combined.
\SL builds upon these connectives to champion \emph{local reasoning} about the resources
employed by programs. Consequently, program parts can be verified by considering only those
resources they actually access---a crucial property for building scalable 
tools including automated verifiers~\cite{Piskac2013Automating,BerdineCO05,ChinDNQ12,0001SS17,JacobsSPVPP11}, 
static analyzers~\cite{BerdineCCDOWY07,GotsmanBCS07,CalcagnoDOY11}, 
%model checkers~\cite{ArndtJKMN18}, 
and interactive theorem provers~\cite{JungKJBBD18}.
At the foundation of almost any automated approach based on \SL, lies the \emph{entailment problem}
$\varphi \entails \psi$: are all models of \SL formula $\varphi$ also models of \SL formula $\psi$?
For example, Hoare-style verifiers need to solve entailments whenever they invoke the
rule of consequence, and static analyzers ultimately solve entailments to perform abstraction.
While undecidable in general~\cite{Antonopoulos2014Foundations}, the wide adoption of \SL and the central
role of the entailment problem have triggered a massive research effort to
identify \SL fragments with a decidable entailment problem~\cite{Berdine2004Decidable,Cook2011Tractable,Echenim2021Decidable,Echenim2021Unifying,Iosif2013Tree,Iosif2014Deciding,Katelaan2019Effective,Matheja2020Complete,Reynolds2016Decision,Demri2018,Echenim2020Bernays},
 and to build practical entailment solvers~\cite{Piskac2013Automating,BerdineCO05,ChinDNQ12,TaLKC18}.

% ProbProg
\emph{Probabilistic programs}, that is, programs with the ability to sample from probability
distributions, are an increasingly popular formalism for, amongst others,
designing efficient randomized algorithms~\cite{0012859} and
describing uncertainty in systems~\cite{Gordon2014Probprogs,CarbinMR16}.
%modeling data science problems~\cite{probprogstuff}.
While formal reasoning techniques for probabilistic programs exist since the 80s
(\cf~\cite{Kozen1997Semantics,Kozen1983Probabilistic,Saheb1978}), they are rarely automated and typically target only
simplistic programming languages.
For example, verification techniques that support reasoning about both randomization
and data structures are, with notable exceptions~\cite{Tassarotti2019Separation,Batz2019Quantitative}, rare---a surprising situation given that randomized algorithms typically rely on dynamic data structures.

% QSL
\emph{Quantitative separation logic} (\QSL)
is a weakest-precondition-style verification technique that targets randomized algorithms manipulating complex data structures; it marries \SL and weakest preexpectations~\cite{Morgan1996}---a well-established calculus for reasoning about probabilistic programs.
In contrast to classical \SL, \QSL's assertion language does not consist of predicates, which evaluate to Boolean values, 
but \emph{expectations} (or: random variables), which evaluate to real numbers.
\QSL has been successfully applied to the verification of randomized algorithms,
and \QSL expectations have been formalized in Isabelle/HOL~\cite{haslbeck_diss}.
However, reasoning is far from automated---mainly due to the lack of decision procedures or
solvers for entailments between expectations in \QSL.

This paper presents, to the best of our knowledge, the first technique for
automatically deciding \QSL entailments.
More precisely, we reduce \QSL quantitative entailments to classical entailments between SL formulas.
Hence, we can leverage two decades of separation logic research
to advance \QSL entailment checking, and thus also automated reasoning about probabilistic programs.

\paragraph{Contributions.} We make the following technical contributions:
\begin{itemize}
  \item We present a generic construction that reduces the entailment problem for
        quantitative separation logic
        to solving multiple entailments in fragments of \SL; 
        if we reduce to an \SL fragment where entailment is decidable, 
        our construction yields a \QSL fragment with a decidable entailment problem.
  \item We provide simple criteria for whether one can leverage
        a decision procedure or a practical entailment solver for \SL to build an entailment solver for \QSL.
  \item We analyze the complexity of our approach parameterized in the complexity of solving entailments in a given \SL fragment;
        whenever we identify a decidable \QSL fragment, it is thus accompanied by upper complexity bounds.
   \item We use our construction to derive the \QSL fragment of quantitative symbolic heaps
         for which entailment is decidable via a reduction to the Bernays-Sch\"onfinkel-Ramsey
         fragment of \SL~\cite{Echenim2020Bernays}.
\end{itemize}

\paragraph{Outline.} 
\Cref{sec:intensional} introduces (quantitative) separation logic. \Cref{sec:motivation} motivates our approach by providing the foundations for probabilistic pointer program verification with $\QSL$ together with several examples. We present the key ideas and our main contribution of reducing $\QSL$ entailment checking to $\SL$ entailment checking in \Cref{sec:entailment_checking}. We analyse the complexity of our approach in \Cref{sec:complexity}. In \Cref{sec:applications}, we apply our approach to obtain the first decidability results for probabilistic pointer verification. Finally, \Cref{sec:related} discusses related work and \Cref{sec:conclusion} concludes.

%\cm{Detailed proofs and examples are found in this paper's extended version~\cite{arxiv}.}
\cmcomment{Remove for arxiv, include for CR}

\begin{table}[h]
	\centering
	\caption{Metavariables used throughout this paper.}%
	\label{tab:metavariables}%
	\renewcommand{\arraystretch}{1.25}%
	\begin{tabular}{l@{\qquad}l@{\qquad}l@{\qquad}l}
		\hline\hline
		\textbf{Entities}			& \textbf{Metavariables}  				& \textbf{Domain}				\\
		\hline
		Natural numbers 		& $n,\, i,\, j,\, k$ 				                        & $\Nats$				             					              \\
		Rational probabilities 		& $\pp,\qq,\ra, \rb, \rc, \rd$ 				                           & $\Probs$									              \\
		%
		%
		%Non-negative rational numbers & $\ra, \rb, \rc, \rd, \re$											& %$\NonNegRats$                          &                           \\
		%
		%
		\\
		Programs                    & $\cc$                                                           &$\hpgcl$                                                     \\
		Stacks          & $\sk$                                       &$\Stacks$                                                     \\%
		Heaps          & $\hh$                                       &$\Heaps$                                                     \\
		\\
		Variables 		                & $\xx,\xy,\xz$ 				                           & $\Vars$			         						              \\
		Values 		                & $v,w$ 				                           & $\Vals$			         						              \\
		Locations 		                & $\loca$ 				                           & $\Locs$			         						              \\
		\\
	     Predicates                   & $\preda$            & $\Preds$                                                  \\
	     %
	    %  Pure predicates           & $\BB$   &  &                                 \\
	     %
	     one-bounded expectations               & $\FF$                              & $\Eone$                                                                  \\       %
         \\
         \SL formulae                         &     $\sla,\slb,\slc$                               &             $\SLpar{\cdot}$ \\
         Pure formulae           & $\bb$   &  &                                 \\
         \QSL formulae                         &     $\ff,\fg,\fh,\INV$                               &             $\QSLpar{\cdot}$ \\
	     
%
%%%%%		Interpretations 			& $\interpret$ 				& 		 			& \Cref{sec:semantics}		\\
		\hline\hline
	\end{tabular}%
	\renewcommand{\arraystretch}{1}%
	%\vspace*{.5em}%
\end{table}

\section{(Quantitative) Separation Logic}
\label{sec:intensional}
\subsection{Program States}\label{sec:qsl:program-states}
%A \emph{program state} $(s, h)$ consists of a \emph{stack} $s$, i.e.\ a valuation of variables by integers, and a \emph{heap} $h$ modeling dynamically allocated memory.
%Formally, the set of \emph{stacks} is given by
Let $\Vals$ be a countably infinite set of \emph{values}, and let $\Vars$ be a countably infinite set of variables with domain $\Vals$. The set of \emph{stacks} is given by
$$%\begin{align*}
	{\Stacks} = \setcomp{\sk}{\sk \colon \Vars \To \Vals}~.
$$ %\end{align*}
%
%Like in a standard RAM model, a heap consists of memory addresses that each store a value and is thus a \emph{finite} mapping from addresses (i.e.\ natural numbers) to values (which may themselves be allocated addresses in the heap).
%Formally, the set of \emph{heaps} is given by
%
Let $\Locs \subset \Vals$ be an infinite set of \emph{locations}. We denote locations by $\loca$ and variations thereof.
We fix a natural number $\recnum \geq 1$ and a heap model where finite sets of locations are mapped to fixed-size records over $\Vals$ of size $\recnum$. Put more formally, the set of \emph{heaps} is given by
\begin{align*}
%	{\Heaps} \eeq \setcomp{\hh}{\hh \colon N \To \Ints,~ N \subseteq \Nats \setminus \{0\},~ |N| < \infty}.
	{\Heaps} \eeq \setcomp{\hh}{\hh \colon \locseta \To \Vals^\recnum,~ \locseta \subseteq \Locs,~ |\locseta| < \infty}.
\end{align*}
%
%
%where $\Nats_{>0} = \Nats \setminus \{0\}$.
The set of \emph{program states} is then given by
$$%\begin{align*}
	{\States} \eeq \setcomp{(\sk, \hh)}{\sk \in \Stacks,~ \hh \in \Heaps}~.
$$ %\end{align*}
%
%Notice that expressions $\ee$ and guards $\guard$ 
%may depend on variables only (i.e.\ they may \emph{not} depend upon the heap) and thus their evaluation never causes any side effects.
%Side effects such as dereferencing unallocated memory can only occur \emph{after} evaluating an expression and trying to access the memory at \mbox{the evaluated address}. 
Given a program state $(\sk, \hh)$ and an expression $\terma$ over $\Vars$, we denote by \textbf{$\terma(\sk)$} the evaluation of expression $\terma$ in $\sk$, \ie the value that is obtained by evaluating $\terma$ after replacing any occurrence of any variable $\xx \in \Vars$ in $\terma$ by the value $\sk(\xx)$.
%Analogously, we denote by \textbf{$s(\guard)$} the evaluation of $\guard$ in stack $s$ to either $\true$ or $\false$.
We write $\sk\subst{x}{v}$ to indicate that we set variable $x$ to value $v \in \Vals$ in $\sk$, i.e.\footnote{We use $\lambda$-expressions to denote functions: Function $\lambda X \mydot f$ applied to an argument $v$ evaluates to $f$ in which every occurrence of $X$ is replaced by $v$.},	
\begin{align*}
	\sk\statesubst{x}{v} \eeq \lambda\, y\mydot \begin{cases}
		v, & \textnormal{if } y = x\\
		\sk(y), & \textnormal{if } y \neq x.
	\end{cases}
\end{align*}
For heap $\hh$, $\hh\Update{\loca}{(v_1,\ldots,v_\recnum)}$ is defined analogously. % sets the content of a heap $\hh$ at address $u$ to value $v$, i.e.
%%
%\begin{align*}
%	\hh\Update{u}{v} \eeq \lambda\, w\mydot \begin{cases}
%            v, & \textnormal{if } w = u \\
%            \hh(w), & \textnormal{if } w \neq u.
%	\end{cases}
%\end{align*}
%%
For a given heap $\hh \colon \locseta \To \Vals^\recnum$, we denote by $\dom{\hh}$ its \emph{domain} $\locseta$. %, i.e.\
%
%
%\begin{align*}
%	\dom{h} \eeq N.
%\end{align*}
%
Two heaps $\hh_1$, $\hh_2$ are \emph{disjoint}, denoted $\hh_1 \disjoint \hh_2$, if their domains do not overlap, \ie $\dom{\hh_1} \cap \dom{\hh_2} = \emptyset$.
%For any two heaps $h_1$, $h_2$, we denote \emph{disjointness} of $h_1$ and $h_2$, i.e.\ the fact \mbox{that $\dom{h_1} \cap \dom{h_2} = \emptyset$, by}
%
%\begin{align*}
%	h_1 \disjoint h_2.
%\end{align*}
%
%For two heaps $h_1 \colon N_1 \To \Rats$ and $h_2 \colon N_2 \To \Rats$ with disjoint domains, i.e. $\dom{h_1} \cap \dom{h_2} = \emptyset$, we denote their \emph{disjoint union} by $h_1 \sepcon h_2$, i.e.\  $h_1 \sepcon h_2\colon \dom{h_1} \mathrel{\dot{\cup}} \dom{h_2} \To \Rats$, with
The \emph{disjoint union} of two disjoint heaps $\hh_1 \colon \locseta_1 \To \Vals^\recnum$ and $\hh_2 \colon \locseta_2 \To \Vals^\recnum$ is
\begin{align*}
  \hh_1 \sepcon \hh_2\colon \dom{\hh_1} \mathrel{\dot{\cup}} \dom{\hh_2} \To \Vals^\recnum, ~ %with
	\bigl(\hh_1 \sepcon \hh_2 \bigr) (\loca) \eeq \begin{cases}\hh_1(\loca), & \textnormal{if } \loca \in \dom{\hh_1} \\ \hh_2(\loca), & \textnormal{if } \loca \in \dom{\hh_2}. \end{cases}
\end{align*}
%
%We denote \emph{heap inclusion} by $h_1 \subseteq h_2$ and define it formally as 
%
%
%Likewise, for an allocated address $n$, i.e.\ $n \in \dom{h}$, and a value $v \in \Rats$ we write \textbf{$h\subst{n}{v}$} to indicate that the value stored at address $n$ is set to value $v$ in heap $h$, i.e. if $n \in \dom{h}$, then
%%
%\begin{align*}
%	h\subst{n}{v} \eeq \lambda\, m\mydot \begin{cases}
%		v, & \textnormal{if } m = n\\
%		h(m), & \textnormal{if } m \neq n.
%	\end{cases}
%\end{align*}
%%
%

\subsection{Separation Logic}
\label{sec:sl}
\begin{table}[t]
	\centering
	\caption{Semantics of $\SLA$ formulae.}%
	\label{tab:semantics_sl}%
	\renewcommand{\arraystretch}{1.4}%
	\begin{tabular}{l@{\qquad}l@{\qquad}l@{\qquad}l}
		\hline\hline
		$\sla$			&          $(\sk,\hh)\models \sla$ iff			\\
		\hline
		$\slc$		& $(\sk,\hh)\in \sem{\slc}$		              \\
		$\neg \slb$		& $(\sk,\hh)\not\models \slb$		              \\
		$\slb \wedge \slc$		& $(\sk,\hh)\models \slb$ and $(\sk,\hh)\models \slc$ 		              \\
		$\slb \vee \slc$		& $(\sk,\hh)\models \slb$ or $(\sk,\hh)\models \slc$ 		              \\
		$\exists \xx \colon \slb$  & $(\sk\statesubst{\xx}{v},\hh)\models \slb$ for some  $v \in \Vals$ \\
		$\forall \xx \colon \slb$  & $(\sk\statesubst{\xx}{v},\hh)\models \slb$ for all $v \in \Vals$\\
		$\slb \sepcon \slc$ &  $(\sk,\hh_1) \models \slb$ and $(\sk,\hh_2)\models \slc$ for some $\hh_1 \sepcon \hh_2 = \hh$ \\
		$\slb \sepimp \slc$ &  $(\sk,\hh\sepcon\hh')\models \slc$ for all $\hh' \disjoint \hh$ with $(\sk,\hh') \models \slb$ \\
		\hline\hline
	\end{tabular}%
	\renewcommand{\arraystretch}{1}%
	%
	%\vspace*{.5em}%
\end{table}
 A \emph {predicate} $\preda \in \Preds$ is a set of states. A predicate $\preda$ is called \emph{pure} if it does not depend on the heap, i.e, for every stack $\sk$ and heaps $\hh,\hh'$, we have $(\sk,\hh) \in \preda$ iff $(\sk,\hh')\in \preda$. 
 %We call a set $\PurePredset$ of predicates pure if every predicate in $\PurePredset$ is pure. 
 %A formula $\slb$ is called \emph{qualitative} if its semantics $\sem{\slb}$ is a predicate. A qualitative formula $\slb$ is called \emph{pure} if $\sem{\slb}$ is pure.
 %
 %We call a formula \emph{qualitative} 
 %
 
%
%
We consider a separation logic $\SLA$ with standard semantics \cite{Reynolds2002Separation}.
 A distinguishing aspect is that $\SLA$ is parametrized by a set $\Predset$ of predicate symbols $\slb$ with \emph{given} semantics $\sem{\slb} \in \Preds$. We often identify predicate symbols $\slb$ with their predicates $\sem{\slb}$. Elements of $\Predset$ build the atoms of $\SLA$.
 Our reduction from quantitative entailments to qualitative entailments does not depend on the choice of these predicate symbols. We therefore take a generic approach that allows for user-defined atoms, e.g., list or tree predicates. 
 %The distinction between predicates and \emph{pure} predicates becomes important in the quantitative setting. \\
%
%$\SLA$ is the qualitative counterpart of $\QSLA$.
%
\begin{definition}
	Let $\Predset$ be a countable set of predicate symbols. Formulae in \emph{separation logic $\SLA$} with atoms in $\Predset$ adhere to the grammar
	\begin{align*}
	\sla \quad \rightarrow \quad 
	&\slc 
	~\mid~ \neg \sla 
	 ~\mid~ \sla \wedge \sla  
	~\mid~ \sla \vee \sla  
	~\mid~ \exists \xx \colon \sla 
	~\mid~ \forall \xx \colon \sla 
	~\mid~ \sla \sepcon \sla 
	~\mid~ \sla \sepimp \sla ~, 
	\end{align*} 
	where $\slc \in \Predset$,  and where $\xx \in \Vars$. \hfill $\triangle$
\end{definition}
The Boolean connectives $\neg$, $\wedge$, and $\vee$ as well as the quantifiers $\exists$ and $\forall$ are standard. $\sepcon$ is the \emph{separation conjunction} and $\sepimp$ is the \emph{magic wand}.

The semantics $\sem{\sla} \in \Preds$ of a formula $\sla \in \SLA$ is defined by induction on the structure of $\sla$ as shown in \Cref{tab:semantics_sl}. Recall that we assume the semantics $\sem{\slb}$ of predicate symbols $\slb\in \Predset$ to be given. We often write $(\sk,\hh)\models \sla$ instead of $(\sk,\hh) \in \sem{\sla}$. For $\sla,\slb \in \SLA$, we say that \emph{$\sla$ entails $\slb$}, denoted $\sla \entails \slb$, if whenever $(\sk,\hh) \in \States$ such that $(\sk,\hh)\models \sla$, also $(\sk,\hh)\models \slb$.
%
%\tncommentinline{Conservativity property, relating $\wedge$ to $\cdot$/$\emin{}{}$, $\vee$ to $\emax{}{}$, $\exists$ to $\sup$, and $\forall$ to $\inf$?}
%
%\kbcommentinline{Example instantiation? E.g. example form Sec 4?}
%
%
\begin{example}\label{ex:sl}
	\begin{changed}
	Let $\Vals = \Ints$, $\Locs=\PosNats$, and $\recnum = 1$. A \emph{term} $\terma$ is either a variable $\xx\in\Vars$ or the constant $0 \in \Vals$. 
	The set $\Predset$ of \mbox{predicate symbols is} %given by
	\begin{align*}
		\Predset \eeq \{\: \true, \slemp, \slsingleton{\xx}{\terma},
	\terma=\terma', \terma\neq\terma', \slLs{\terma}{\terma'} ~\mid~ \xx \in \Vars, \terma, \terma'~\text{terms}\:\}
	\end{align*}
	%
	%
%	\begin{align*}
%	\Predset \quad \rightarrow \quad  
%	\terma=\terma' ~\mid~ \terma\neq\terma' ~\mid~ \true
%	~\mid~ \slemp ~\mid~ \slsingleton{\xx}{\terma}
%	 ~\mid~ \slLs{\terma}{\terma'}~.
%	\end{align*}
	%
	%
	Here, apart from standard predicates for $\true$, equalities,
	 and disequalities,
	\begin{enumerate}
		\item $\slemp$ is %(identified with) 
		      the \emph{empty-heap predicate}, i.e., 
		\[
		(\sk,\hh)\models\slemp \quad \text{iff}\quad \dom{\hh}= \emptyset~,
		\]
		\item $\slsingleton{\xx}{\terma}$ is the \emph{points-to predicate}, i.e., 
		\[
		(\sk,\hh)\models \slsingleton{\xx}{\terma}
		\quad \text{iff}\quad 
		\dom{\hh} = \{\sk(\xx)\}~\text{and}~ \hh(\sk(\xx)) = \terma(\sk)~,
		\]
		\item the \emph{list predicate}  $\slLs{\terma}{\terma'}$ asserts that the heap models a singly-linked list segment from $\terma$ to $\terma'$:
		\begin{align*}
		 &(\sk,\hh)\models\slLs{\terma}{\terma'} \\
		 \text{iff} \quad&
		 \text{$\dom{\hh} = \emptyset$ and $\terma(\sk) = \terma'(\sk)$ or} \\
		 &\text{there exist $n \geq 1$ and terms $\terma_1,\ldots,\terma_n$ with $\terma_n = \terma'$ such that} \\
		  &\text{$(\sk,\hh)\models \slsingleton{\terma}{\terma_1}\sepcon\ldots\sepcon\slsingleton{\terma_{n-1}}{\terma_n}$~.}
		 %&\text{there is $n\geq 1$, $\loca_1,\ldots,\loca_n \in \Locs$, $v_1,\ldots,v_n\in \Vals$ such that} \\
		 %& \loca_1 = \terma(\sk), v_n = \terma'(\sk)~\text{and} \\
		 %&\text{$(\sk,\hh)\models \slsingleton{\loca_1}{v_1}\sepcon\ldots\sepcon\slsingleton{\loca_n}{v_n}$ , and $v_i = \loca_{i+1}$ for $1\leq i < n$~.}
		 %
		\end{align*}
	\end{enumerate}
	\end{changed}
	In this setting, $\SLA$ contains, e.g., the well-known \emph{symbolic heap fragment} of separation logic with lists. 
	For instance, the $\SLA$ formula
	\[
	\exists y\colon ~ \exists z\colon ~  \slsingleton{x}{y} \sepcon \slsingleton{y}{z} \sepcon \slLs{z}{0}~.
	\] 
	asserts that the heap consists of a list with head $\xx$ of length at least $2$. \hfill $\triangle$
\end{example}

\subsection{Quantitative Separation Logic}
\begin{table}[t]
	\centering
	\caption{Semantics of $\QSLA$ formulae.}%
	\label{tab:semantics_qsl}%
	\renewcommand{\arraystretch}{1.4}%
	\begin{tabular}{l@{\qquad}l@{\qquad}l@{\qquad}l}
		\hline\hline
		$\ff$			&          $\semapp{\ff}{(\sk,\hh)}$			\\
		\hline
		$\iverson{\slb}$		& $\iverson{\slb}(\sk,\hh)$	              \\
		$\iverson{\bb} \cdot \fg + \iverson{\neg\bb} \cdot \fh $ & $\iverson{\bb}(\sk,\hh) \cdot \semapp{\fg}{(\sk,\hh)} + \iverson{\neg\bb}(\sk,\hh) \cdot \semapp{\fh}{(\sk,\hh)}$ \\
		$\qq \cdot \fg + (1-\qq) \cdot \fh $ & $\qq \cdot \semapp{\fg}{(\sk,\hh)} + (1-\qq) \cdot \semapp{\fh}{(\sk,\hh)}$ \\
		$\fg \cdot \fh $  & $\semapp{\fg}{(\sk,\hh)} \cdot \semapp{\fh}{(\sk,\hh)}$ \\
		$1-\fg$ & $1-\semapp{\fg}{(\sk,\hh)}$  \\
		$\emax{\fg}{\fh}$ & $\max \{\semapp{\fg}{(\sk,\hh)}, \semapp{\fh}{(\sk,\hh)} \} $\\
		$\emin{\fg}{\fh}$ & $\min \{\semapp{\fg}{(\sk,\hh)}, \semapp{\fh}{(\sk,\hh)} \} $\\
		$\Sup \xx \colon \fg $ & $\max \big\{\semapp{\fg}{(\sk\statesubst{\xx}{v},\hh)} ~\mid~ v \in \Vals \big\} $\\
		$\Inf \xx \colon \fg $ & $\min \big\{\semapp{\fg}{(\sk\statesubst{\xx}{v},\hh)} ~\mid~ v \in \Vals \big\} $\\
		$\fg \sepcon \fh $ & $\max \left\{\semapp{\fg}{(\sk,\hh_1)} \cdot \semapp{\fh}{(\sk,\hh_2)} ~\mid~ \hh = \hh_1 \sepcon \hh_2 \right\}$ \\
		$\iverson{\slb} \sepimp \fg$ & $\inf \left\{ \semapp{\fg}{(\sk,\hh\sepcon\hh')}  ~\mid~ \hh' \disjoint \hh ~\text{and}~ \iverson{\slb}(\sk,\hh) = 1 \right\}$ \\
		\hline\hline
	\end{tabular}%
	\renewcommand{\arraystretch}{1}%
	%
	%\vspace*{.5em}%
\end{table}
In quantitative separation logic \cite{Batz2019Quantitative,Matheja2020Automated}, formulae evaluate to non-negative real numbers or infinity instead of truth values. By conservatively extending the weakest preexpectation calculus by McIver \& Morgan \cite{McIver2005Abstraction}, this enables the compositional verification of probabilistic pointer programs by reasoning about expected list-sizes, probabilities of terminating with an empty heap, and alike.

We consider here a fragment of quantitative separation logic suitable for reasoning about the likelihood of events in probabilistic pointer programs such as, e.g., the probability of terminating in a given symbolic heap. The formulae we consider evaluate to \emph{rational probabilities} rather than arbitrary reals or infinity. We denote the set $[0,1]\cap \NonNegRats$ of rational probabilities by $\Probs$. Like $\SLA$, quantitative separation logic is parameterized by a set $\Predset$ of predicate symbols $\slb$ with \emph{given} semantics $\sem{\slb}\in \Preds$, building the atoms of $\QSLA$.

\begin{definition}
Let $\Predset$ be a countable set of predicate symbols. Formulae in \emph{quantitative separation logic $\QSLA$} with atoms in $\Predset$ \mbox{adhere to the grammar}
	\begin{align*}
		\ff \quad \rightarrow \quad 
		&\iverson{\slb} 
		~\mid~ \iverson{\bb} \cdot \ff + \iverson{\neg\bb} \cdot \ff 
		~\mid ~ \qq \cdot \ff + (1-\qq) \cdot \ff  
		~\mid~ \ff \cdot \ff  \\
	    &\mid~ 1 - \ff  
		~\mid~ \emax{\ff}{\ff} 
		~\mid~ \emin{\ff}{\ff} 
		~\mid~ \Sup \xx \colon \ff 
		~\mid~ \Inf \xx \colon \ff \\
		&\mid~ \ff \sepcon \ff 
		~\mid~ \iverson{\slb} \sepimp \ff ~, 
	\end{align*} 
    where $\slb,\bb \in \Predset$ with $\bb$ pure,  $\qq \in \Probs$, and where $\xx \in \Vars$. \hfill $\triangle$

\end{definition}
The semantics of a formula $\ff \in \QSLA$ is a \emph{(one-bounded) expectation}. The set $\Eone$ of one-bounded expectations is defined as
\[
\Eone \eeq \left\{ \FF ~\mid~\FF\colon \States \to [0,1] \right\}~.
\]
We use the \emph{Iverson bracket}~\cite{Iverson1962} notation $\iverson{\preda}$ to associate with predicate $\preda$ its indicator function.
Formally, %the Iverson bracket \textbf{$\iverson{\varphi}$} of predicate $\varphi$ is defined \mbox{as the function}
\begin{align*}
\iverson{\preda} \colon\quad \States \To \{0,1\},\quad \iverson{\preda}(\sk, \hh) \eeq \begin{cases}
1, & \textnormal{if $(\sk, \hh) \in \preda$}\\
0, & \textnormal{if $(\sk, \hh) \not\in \preda$}~.
\end{cases}
\end{align*} 
Given a predicate symbol $\slb$, we often write $\iverson{\slb}$ instead of $\iverson{\sem{\slb}}$.
The semantics $\sem{\ff} \in \Eone$ of $\ff \in \QSLA$ is defined by induction on the structure of $\ff$ in \Cref{tab:semantics_qsl}. We write $\ff \equiv\fg$ if $\ff$ and $\fg$ are \emph{equivalent}, i.e. if $\sem{\ff} = \sem{\fg}$. Infima and suprema are taken over the complete lattice $([0,1], \leq)$. In particular, $\inf \emptyset =1$ and $\sup \emptyset = 0$.
\begin{theorem} \label{thm:qsl_well_defined}
	The semantics of $\QSLA$ formulae is well-defined, i.e., for all $\ff\in \QSLA$, we have
	%
	%\[
	$\sem{\ff} \in\Eone.$
	%\]
\end{theorem}
\begin{proof}
	By induction on the structure of $\ff$. 
	For details see \Cref{app:qsl_well_defined}.\cmcomment{TR / App.}
\end{proof}
Let us go over the individual constructs. Formulae of the form $\iverson{\slb}$ are the atomic formulae.  $\iverson{\bb} \cdot \fg + \iverson{\neg\bb} \cdot \fh$ is a \emph{Boolean} choice between $\fg$ and $\fh$ that does not depend upon the heap since $\sem{\bb}$ is pure.  $\qq \cdot \fg + (1-\qq) \cdot \fh$ is a convex combination of $\fg $ and $\fh$. $\fg \cdot \fh $ is the pointwise multiplication of $\fg$ and $\fh$. $1-\fg$ is the quantitative (or probabilistic) negation of $\fg$. $\emax{\fg}{\fh}$ and $\emin{\fg}{\fh}$ is the pointwise maximum and minimum of $\fg$ and $\fh$, respectively.

 $\Sup \xx \colon \fg$ is the \emph{supremum quantification} that, given a state $(\sk,\hh)$, evaluates to the supremum of the set obtained from evaluating $\fg$ in $(\sk\statesubst{\xx}{v}, \hh)$ for every value $v \in \Vals$. In our setting, this supremum is actually a maximum. Dually, $\Inf\xx \colon \fg$ is the \emph{infimum quantification}. 
 
 $\sepcon$ and $\sepimp$ are the quantitative analogous of the separating conjunction and the magic wand from separation logic as defined in \cite{Batz2019Quantitative}. $\fg \sepcon \fh$ is the \emph{quantitative separating conjunction} of $\fg$ and $\fh$. Intuitively speaking, whereas the \emph{qualitative} separating conjunction maximizes a \emph{truth value} under all appropriate partitionings of the heap, the \emph{quantitative} separating conjunction maximizes a \emph{probability}. $\iverson{\slb} \sepimp \fh$ is the \emph{quantitative magic wand}. Whereas the \emph{qualitative} magic wand minimizes a \emph{truth value} under all appropriate extensions of the heap, the \emph{quantitative} magic wand minimizes a \emph{probability}. For an in-depth treatment of these connectives, we refer to \cite{Batz2019Quantitative}.

\begin{example}
	Let $\Vals$, $\Locs$, $\recnum$, and $\Predset$ be as in \Cref{ex:sl}. Then $\QSLA$ contains, e.g., a quantitative extension of the symbolic heap fragment of separation logic with lists. For instance, the $\QSLA$ formula 
	\[
	0.7 \cdot(\Sup y\colon ~ \Sup z\colon ~  \singleton{x}{y} \sepcon \singleton{y}{z} \sepcon \ils{z}{0}) + 0.3 \cdot \iverson{\slemp}
	\] 
	expresses that with probability $0.7$ the heap consists of a list with head $\xx$ of length at least $2$ and that with probability $0.3$ the heap is empty. \hfill $\triangle$
\end{example}

Finally, given $\ff,\fg \in \QSLA$, we say that $\ff$ \emph{entails} $\fg$, denoted $\ff \entails \fg$, if
\[
\text{for all}~ (\sk,\hh) \in \States\colon \quad \semapp{\ff}{(\sk,\hh)} \lleq \semapp{\fg}{(\sk,\hh)}~.
\]
\cm{Quantitative entailments $\ff \entails \fg$ generalize classical entailments in the sense that $\ff$ (pointwise) lower-bounds the quantity $\fg$. 
For example, if $g$ assigns to each state the probability that some program $\cc$ terminates without a memory error, then the entailment $\iverson{\true} \entails \fg$ means that $\cc$ terminates almost-surely, i.e., with probability one.
}
Our problem statement now reads as follows: Reduce entailment checking in $\QSLA$ to checking finitely many entailments in $\SLA$.

%
%

%
%%
%
%
%%%
%
%
%
%
%%
%
%
%%%%
%
%
%
%
%
%
%\subsection{Problem Statement} 
%

\section{Entailments in Probabilistic Program Verification}
\label{sec:motivation}
Our primary motivation for studying the entailment problem for quantitative
separation logic is to provide foundations for the automated verification
of probabilistic pointer programs.
In this section, we consider examples of such programs
written in \hpgcl---an extension of McIver \& Morgan's probabilistic guarded command
language (\cf \cite{McIver2005Abstraction}) by heap-manipulating instructions--- and the entailments that
arise from their verification.
We briefly formalize reasoning about \hpgcl programs with weakest liberal preexpectations;
for a thorough introduction of \hpgcl programs and techniques for their verification,
we refer to~\cite{Batz2019Quantitative,Matheja2020Automated}.

\subsection{Heap-manipulating pGCL}
Recall from \Cref{sec:qsl:program-states} that heaps map memory locations to fixed-size records (or tuples) of length $\recnum \geq 1$.
The set of programs in \emph{heap-manipulating probabilistic guarded command language} for $\recnum = 1$, \changedinline{$\Vals=\Ints$ and $\Locs=\PosNats$}, denoted $\hpgcl$, is given by the grammar
\begin{align*}
	%\begin{array}{lll}
    \begin{aligned}
    \cc  ~~\longrightarrow~~ &\SKIP & \text{(effectless program)} \\
	%& \quad |~~ \DIVERGE \tag{crash} \\
    & |~~ \ASSIGN{x}{\ee} & \text{(assignment)} \\%[0.5em]
    & |~~ \PCHOICE{\cc}{\pp}{\cc} & \text{(prob. choice)}\\
    & |~~ \COMPOSE{\cc}{\cc} & \text{(seq. composition)} \\
    & |~~ \ITE{\guard}{\cc}{\cc} & \text{(conditional choice)} \\
	%& \quad |~~ \NDCHOICE{C}{C} \tag{nondeterministic choice}\\
    & |~~ \WHILEDO{\guard}{\cc} & \text{(loop)} \\
    & |~~ \ALLOC{x}{\ee} & \text{(allocation)} \\
    %& |~~ \ALLOC{x}{\ee_1,\, \ldots,\, \ee_\recnum} & \text{(allocation)} \\
    & |~~ \FREE{\ee}, & \text{(disposal)} \\
    & |~~ \ASSIGNH{x}{\ee} & \text{(lookup)} \\
    & |~~ \HASSIGN{\ee}{\ee'} & \text{(mutation)} \\
    \end{aligned}
	%\end{array}
\end{align*}
where $x\in\Vars$, $\pp\in\Probs$, 
$\ee,\ee'$ are arithmetic expressions and
$\guard$ is a Boolean expression.
We assume that expressions do not depend on the heap.
For now, we do not fix a specific syntax for expressions but assume evaluation mappings 
\begin{align*}
  \ee\colon \Stacks \to \Ints
  \qand
  \guard\colon \Stacks \to \{\true,\false\}~.
\end{align*}
In addition to the usual control flow structures for sequential composition, conditionals, and loops, 
$\SKIP$ does nothing, 
$\ASSIGN{x}{\ee}$ assigns the value $\ee(\sk)$ obtained from evaluating expression $\ee$ in the current program state $(\sk,\hh)$ to $x$, and the probabilistic choice
$\PCHOICE{\cc_1}{\pp}{\cc_2}$ flips a coin with bias $p$---it executes $\cc_1$ if the coin flip yields heads, and $\cc_2$ otherwise.
The allocation $\ALLOC{x}{\ee}$ nondeterministically selects a fresh location, stores it in $x$, and puts a record with value $\ee$ on the heap at that location.
Since we assume an infinite address space, allocation never fails.
Conversely, $\FREE{\ee}$ disposes the record at location $\ee$ from the heap; it fails if no such location exists.
The mutation $\HASSIGN{\ee}{\ee'}$ and the lookup $\ASSIGNH{x}{\ee}$ update to $\ee'$ resp.\ assign to $x$ the value stored at location $\ee$; both statements fail if the heap contains no such location.

\subsection{Weakest Liberal Preexpectations}

We formalize reasoning about $\hpgcl$ programs in terms of the 
weakest liberal preexpectation transformer
$\wlpsymbol\colon \hpgcl \to (\QSLA \to \QSLA)$, where
$\Predset$ at least contains formulae of the form $\singleton{\ee}{\ee'}$;
\Cref{tab:wlp} summarizes the rules for computing $\wlpsymbol$
of \emph{loop-free} programs on the program structure.

\begin{table}[t]
\centering
\caption{Rules for compositionally computing weakest liberal preexpectations.
Here, $\ff$ is a $\QSLA$ formula representing the postexpectation.
$\ff\subst{x}{\ee}$ denotes the substitution of every free occurrence
of $x$ by $\ee$ in $\ff$.
$\validpointer{\ee}$ desugars to $\protect\Sup z\colon\singleton{\ee}{z}$.
}
\label{tab:wlp}
\renewcommand{\arraystretch}{1.5}
\begin{tabular}{@{\hspace{1em}}l@{\hspace{2em}}l}
\hline
$\boldsymbol{\cc}$
&
$\boldsymbol{\boldwlpsymbol\llbracket\cc\rrbracket\left(\ff\right)}$
\\
\hline
$\SKIP$
&
$\ff$
\\
$\ASSIGN{x}{\ee}$
&
$\ff\subst{x}{\ee}$
\\
$\PCHOICE{\cc_1}{\pp}{\cc_2}$
&
$\pp \cdot \wlp{\cc_1}{\ff} + (1- \pp) \cdot \wlp{\cc_2}{\ff}$ 
\\
$\COMPOSE{\cc_1}{\cc_2}$
&
$\wlp{\cc_1}{\vphantom{\big(}\wlp{\cc_2}{\ff}}$ 
\\
$\ITE{\guard}{\cc_1}{\cc_2}$
&
$\iverson{\guard} \cdot \wlp{\cc_1}{\ff} + \iverson{\neg \guard} \cdot \wlp{\cc_2}{\ff}$ 
\\
%$\WHILEDO{\guard}{\cc'}$
%&
%$\fg\qquad\text{if}~\fg \entails \iverson{\neg \guard} \cdot \ff + \iverson{\guard} \cdot \wlp{\cc'}{\fg}$
%%$\gfp \fg\mydot \iverson{\neg \guard} \cdot \ff + \iverson{\guard} \cdot \wlp{\cc'}{\fg}$ 
%\\
$\ALLOC{x}{\ee}$
%$\ALLOC{x}{\ee_1,\ldots,\ee_\recnum}$
&
%$\displaystyle\Inf_{v \in \Locs} \singleton{v}{\ee_1,\ldots,\ee_\recnum} \sepimp \ff\subst{x}{v}$
$\displaystyle\Inf y\colon \singleton{y}{\ee} \sepimp \ff\subst{x}{y}$
\\
$\FREE{\ee}$
&
$\validpointer{\ee} \sepcon \ff$ 
\\
%$\ASSIGNHF{x}{\ee}{i}$
$\ASSIGNH{x}{\ee}$
&
$\displaystyle\Sup y\colon \singleton{\ee}{y} \sepcon \displaystyle\bigl( \singleton{\ee}{y} \sepimp \ff\subst{x}{y} \bigr)$
%$\displaystyle\sup_{v_1,\ldots,v_\recnum \in \Vals} \singleton{\ee}{v_1,\ldots,v_\recnum} \sepcon$ \\
%& \qquad $\displaystyle\bigl( \singleton{\ee}{v_1,\ldots,v_\recnum} \sepimp \ff\subst{x}{v_i} \bigr)$
\\
$\HASSIGN{\ee}{\ee'}$
&
$\displaystyle\validpointer{\ee} \sepcon
\bigl( \singleton{\ee}{\ee'} \sepimp \ff \bigr)$
%$\HFASSIGN{\ee}{i}{\ee'}$
%&
%$\displaystyle\sup_{v_1,\ldots,v_\recnum \in \Vals} \singleton{\ee}{v_1,\ldots,v_\recnum} \sepcon$ \\
%& \qquad $\displaystyle\bigl( \singleton{\ee}{v_1,\ldots,v_{i-1},\ee',v_{i+1},\ldots,v_\recnum} \sepimp \ff \bigr)$
\\
\hline
\end{tabular}
\end{table}

Conceptually, the \emph{weakest liberal preexpectation} $\semapp{\wlp{\cc}{\ff}}{(\sk,\hh)}$
of program $\cc$ with respect to 
\emph{postexpectation} $\ff \in \QSLA$ on $(\sk,\hh)$ is the least expected value of $\sem{\ff}$
(measured in the final states) after successful\footnote{\ie without encountering a memory error.} termination of $\cc$ on initial state $(\sk, \hh)$,
plus the probability that $\cc$ does not terminate on $(\sk, \hh)$.
Adding the non-termination probability can be thought of as a partial correctness view: we include the non-termination probability of $\cc$ on state $(\sk,\hh)$ in the $\wlpsymbol$ of $\cc$ just as we include the state $(\sk,\hh)$ in the weakest liberal pre\emph{condition} of $\cc$ in case $\cc$ does not terminate on $(\sk,\hh)$.

A reader familiar with separation logic will realize the close similarity between
the rules in \Cref{tab:wlp} and the weakest preconditions for SL by
Ishtiaq and O'Hearn~\cite{Ishtiaq2001BI}.
The main differences are (1) the use of the quantitative connectives $\sepcon$, $\sepimp$, and $\cdot$, and $+$, and (2) the additional rule for probabilistic choice,
$\wlp{\PCHOICE{\cc_1}{\pp}{\cc_2}}{\ff}$, which is a convex sum 
that weights $\wlp{\cc_1}{\ff}$ and $\wlp{\cc_2}{\ff}$ by $\pp$ and $(1-\pp)$, respectively.

The transformer $\wlpsymbol$ is well-defined in the sense that, for every loop-free $\hpgcl$-program and every $\QSLA$ formula, we obtain---under mild conditions---again a $\QSLA$ formula:
\begin{theorem}\label{thm:wp_closed}
Let $\cc \in \hpgcl$ be loop-free and $\Predset$ be a set of \mbox{predicate symbols. If}
	\begin{enumerate}
		\item  $\Predset$ contains the points-to predicate for all variables and all expressions occurring in allocation, disposal, lookup and mutation in $\cc$,
		\item $\Predset$ contains all guards and their negations occurring in $\cc$, and
		\item all predicates in $\Predset$ are closed under substitution of variables by variables and arithmetic expressions occurring on right-hand sides of assignments in $\cc$, 
	\end{enumerate}
	then, for every $\QSLA$ formula $\ff$, $\wlp{\cc}{\ff} \in \QSLA$.
\end{theorem}
\begin{proof}
	By induction on loop-free $\cc$. 
	For details see \Cref{app:wp_closed}.\cmcomment{TR / App.}
\end{proof}
For loops, $\wlp{\WHILEDO{\guard}{\cc}}{\ff}$ is typically characterized as the greatest fixed point of loop unrollings.
However, we fixed an explicit syntax of formulae instead of allowing arbitrary
expectations; the above fixed point is in general not expressible in our syntax.\footnote{It is noteworthy that a sufficiently expressive syntax for weakest preexpectation reasoning without heaps has been developed only recently~\cite{Batz2021Relatively}.}
To deal with loops, we thus require a user-supplied invariant $\INV$ and apply the following proof rule 
(\cf~\cite{Kaminski2019Advanced}) to approximate $\wlpsymbol$:
\[
\INV \entails \iverson{\neg \guard} \cdot \ff + \iverson{\guard} \cdot \wlp{\cc'}{\INV}
\qimplies \INV \entails \wlp{\WHILEDO{\guard}{\cc'}}{\ff}
\]
Notice that verifying that $\INV$ is indeed an invariant via the above rule requires proving an entailment between $\QSLA$ formulae.

\subsection{Interfered Swap}\label{sec:motivation:swap}
Our first example concerns a program $\cswap$, implemented in $\hpgcl$ below, 
that attempts to swap the contents
of two memory locations $x$ and $y$.
However, since variable $x$ is shared with a concurrently running process, writing to $x$ can be unreliable, 
that is, instead of the intended value, the concurrently running process may write a corrupted value $\textsf{err}$ into
memory with some \mbox{probability, say $0.001$}. A similar situation occurs, e.g., when using the protocol described in \cite{PWCS2013}.
\begin{align*}
		\cswap\colon \qquad
		& \ASSIGNH{\textsf{tmp1}}{x}\SEMI \\
		& \ASSIGNH{\textsf{tmp2}}{y}\SEMI \\
		& \PCHOICE{\HASSIGN{x}{\textsf{tmp2}}}{0.999}{\HASSIGN{x}{\textsf{err}}}\SEMI \\
		& \HASSIGN{y}{\textsf{tmp1}}~.
\end{align*}
We can use $\wlpsymbol$ to verify an upper bound on the
probability that an erroneous write operation happened by solving the \QSL entailment
    \begin{align*}
       &\wlp{\cswap }{\singleton{\xx}{\xz_2} \sepcon\singleton{\xy}{\xz_1} }    \\
       \models{}~& \iverson{\xz_2 = \textsf{err}}\cdot(\singleton{\xx}{\xz_1} \sepcon\singleton{\xy}{\xz_2} )
                         +\iverson{\xz_2 \neq \textsf{err}}\cdot(0.999 \cdot (\singleton{\xx}{\xz_1} \sepcon\singleton{\xy}{\xz_2} ))~.                      
    \end{align*}
That is, the probability that $\cswap$ successfully swaps the contents of $\xx$ and $\xy$
is at most $0.999$ if $\xy$ does initially not point to the corrupt value $\textsf{err}$. 

As we will see in \Cref{sec:qsh}, our approach for solving \QSL entailments is capable
of deciding the above entailment, where
$\wlp{\cswap }{\singleton{\xx}{\xz_2} \sepcon\singleton{\xy}{\xz_1} }$ is computed 
according to the rules in \Cref{tab:wlp}.

\subsection{Avoiding Magic Wands}\label{sec:no-wands}
Recall from \Cref{tab:wlp} that computing $\wlpsymbol$ introduces a magic wand ($\sepimp$)
for almost every statement that accesses the heap.
This is unfortunate because many decidable separation logic fragments as well as 
practical entailment solvers do not support magic wands.

In particular, in \Cref{sec:qsh} we present a $\QSL$ fragment with a decidable 
entailment problem that supports magic wands only on the left-hand side of entailments.
Hence, proving a \emph{lower} bound on the probability that the program $\cswap$ from
above successfully swapped the contents of two memory cells, \eg
\begin{align*}
0.98 \cdot (\singleton{\xx}{\xz_2} \sepcon \singleton{\xy}{\xz_1}) \entails \wlp{\cswap }{\singleton{\xx}{\xz_1} \sepcon\singleton{\xy}{\xz_2} }~,  
\tag{$\dag$}
\end{align*}
might still be possible with our technique but requires a different
separation logic fragment to reduce to.
%requires a different separation logic  fragment underlying our reduction.
%
%\tncommentinline{What does the former sentence mean?}

Fortunately, we can often avoid introducing magic
wands by employing local reasoning and rules for computing $\wlpsymbol$ for specific
pre- and postexpectations.
In particular, the $\wlpsymbol$ calculus features 
(1) the \emph{frame rule} from separation logic, \ie
if no free variable in $\fg$ is modified by $\cc$, then
$\wlp{\cc}{\ff} \sepcon \fg \entails \wlp{\cc}{\ff \sepcon \fg}$, (2) \cm{ \emph{super-distributivity} for convex combinations and maximum, i.e., $\qq \cdot \wlp{\cc}{\ff} + (1-\qq) \cdot \wlp{\cc}{\fg} \entails \wlp{\cc}{\qq \cdot \ff + (1-\qq) \cdot \fg}$ and $\emax{\wlp{\cc}{\ff}}{\wlp{\cc}{\fg}} \entails \wlp{\cc}{\emax{\ff}{\fg}}$}, and (3) \emph{monotonicity}, i.e., $\ff\entails\fg$ implies $\wlp{\cc}{\ff}\entails\wlp{\cc}{\fg}$.
Moreover, we give four examples of specialized rules that avoid magic wands but require
specific postexpectations: if $x$ is \emph{not} a free variable of $\ee$ or $\ff$, and $\xx$ and $\xy$ are distinct variables, then
\begin{enumerate}[(i)]
  \item 
        $\wlp{\ASSIGNH{\xx}{\ee}}{(\singleton{\ee}{\xy} \cdot \iverson{\xx=\xy}) \sepcon \ff} 
         = \singleton{\ee}{\xy} \sepcon \ff\subst{\xx}{\xy}$~;
  \item 
    $\wlp{\HASSIGN{\ee}{\ee'}}{
      \singleton{\ee}{\ee'} \sepcon \ff} = \validpointer{\ee} \sepcon \ff$~;
  \item 
  $\wlp{\ALLOC{x}{x}}{\Sup \xy \colon \singleton{\xx}{\xy} \sepcon f} =
     \ff\subst{\xy}{\xx}$~; and
  \item 
  $\wlp{\ALLOC{\xx}{\xy}}{\singleton{\xx}{\xy} \sepcon f} = \ff$~.
\end{enumerate}
\cm{Similar rules have been used successfully for symbolic execution with separation logic in non-probabilistic settings~\cite{Berdine2005Symbolic}.}
Combining the above rules with framing, distributivity, and monotonicity often allows avoiding 
magic wands. In such cases, we have a richer set of decidable SL fragments upon which to build
solvers for $\QSL$ entailments at our disposal.
Coming back to the entailment ($\dag$) from above and writing $\cswap=\cc_1;\cc_2;\cc_3;\cc_4$, we calculate %
\begin{changed}
\begin{align*}
                      & \wlp{\cswap}{\singleton{x}{\xz_1} \sepcon \singleton{y}{\xz_2}} \\
\mirrorentails ~ & \wlp{\cswap}{\singleton{y}{\textsf{tmp1}} \sepcon \singleton{x}{\textsf{tmp2}} \cdot \iverson{\textsf{tmp1}=\xz_2} \cdot \iverson{\textsf{tmp2} =\xz_1}} \tag{monotonicity} \\
\mirrorentails ~& \wlpC{\cc_1 \SEMI \cc_2 \SEMI \cc_3}(\wlp{\cc_4}{\singleton{y}{\textsf{tmp1}}} \tag{framing} \\
&\qquad \qquad \qquad \sepcon (\singleton{x}{\textsf{tmp2}}  %\\
%& 
 \cdot (\iverson{\textsf{tmp1} =\xz_2} \cdot \iverson{\textsf{tmp2} =\xz_1})))  \\
\mirrorentails ~& \wlp{\cc_1 \SEMI \cc_2 \SEMI \cc_3}{\validpointer{y} \sepcon (\singleton{x}{\textsf{tmp2}} \cdot (\iverson{\textsf{tmp1} =\xz_2} \cdot \iverson{\textsf{tmp2} =\xz_1})) } \tag{Rule (ii)} \\
\vdots \quad& \\
\phantom{a}\\
\mirrorentails ~ & \wlp{\cc_1}{0.999 \cdot (\singleton{y}{\xz_1} \sepcon (\iverson{\textsf{tmp1} = \xz_2} \cdot \validpointer{x})) + 0.001 \cdot \iverson{\false}} \tag{Rule (i)} \\
\mirrorentails ~ & 0.999 \cdot \wlp{\cc_1}{(\singleton{x}{\xz_2} \cdot \iverson{\textsf{tmp1} = \xz_2}) \sepcon \singleton{y}{\xz_1}}
 + 0.001 \cdot \iverson{\false} \tag{super-distributivity, monotonicity and commutativity} \\
\mirrorentails ~& 0.999 \cdot (\singleton{x}{\xz_2} \sepcon \singleton{y}{\xz_1}) + 0.001 \cdot \iverson{\false}  \tag{Rule (i)}\\
\end{align*}
\end{changed}
which yields a preexpectation without magic wand.
%Detailed calculations are found in \Cref{app:swap_example}.
%Hence, we obtain an entailment that can
%be discharged again with our technique, \eg by reducing to the separation logic fragment discussed in 
%\Cref{sec:qsh}.
Hence, we obtain a magic wand-free entailment in ($\dag$). 
\changedinline{We have used our technique to transform this quantitative entailment into 
several qualitative entailments and checked them successfully using the separation logic extension of CVC4 \cite{Reynolds2016Decision}. 
Detailed calculuations, the resulting qualitative entailments, and the input for CVC4 in SMT-LIB 2 format are found in \Cref{app:swap_example}.}\cmcomment{TR / App.}

\subsection{Randomized List Population}
Our second example populates a singly-linked list by flipping coins
and adding a list element until the coin flip yields heads, \ie
we consider the program
\begin{align*}
  \cpop\colon\qquad & \WHILE{c \neq 0} \\ 
  & \qquad \PCHOICE{\ASSIGN{c}{0}}{0.5}{\ALLOC{x}{x}} \\
  & \}~,
\end{align*}
where $x$ is the head of a linked list.
Assume we would like to determine a lower bound on the probability that the above program does not crash and produces a list of length at least two\footnote{plus the probability of nontermination, which is 0.}. 
For that, recall from \Cref{ex:sl} the 
separation logic formula $\ls{x}{y}$ for singly-linked list segments.
The aforementioned probability is then given by
$\wlp{\cpop}{\ff}$ for postexpectation
\[
  \ff \eeq \Sup y\colon ~ \Sup z\colon ~  \singleton{x}{y} \sepcon \singleton{y}{z} \sepcon \ils{z}{0}~.
\] 
We propose the loop invariant $I$ below to show that
$I \entails \wlp{\cpop}{\ff}$, \ie $I$ is a lower bound on the sought-after probability.
\begin{align*}
  I \eeq &
  \Sup y\colon ~ \singleton{x}{y} \sepcon 
   \big(\iverson{c = 0} \cdot \Sup z\colon ~ \singleton{y}{z} \sepcon \ils{z}{0} \\
    & \qquad\qquad\qquad + \iverson{c \neq 0} \cdot \sfrac{1}{2} \cdot (\Sup z\colon~ \singleton{y}{z} \sepcon \ils{z}{0} + \sfrac{1}{2} \cdot \ils{z}{0})\big)~.
\end{align*}
To verify that $I$ is indeed a loop invariant (hint: it is),
we need to prove that
\begin{align*}
  I ~\entails~
  \iverson{c = 0} \cdot \ff
  + \iverson{c \neq 0} \cdot \wlp{\PCHOICE{\ASSIGN{c}{0}}{0.5}{\ALLOC{x}{x}}}{I}~.
\end{align*}
As described in \Cref{sec:no-wands},
we can compute $\wlpsymbol$ in a way such that the resulting formula contains no magic wands.
Our reduction from $\QSL$ entailments to standard $\SL$ entailments
then allows us to discharge the above invariant check using existing separation logic
solvers with support for fixed list \mbox{predicates, \eg \cite{Piskac2013Automating}.}

\section{Quantitative Entailment Checking}
\label{sec:entailment_checking}
We  present our main contribution of reducing entailment checking in $\QSLA$ to entailment checking in $\SLA$. We consider the key observations leading to our reduction in \Cref{sec:qsla_to_sla_idea}. We then deal with the formalization and more technical considerations of our approach \mbox{in Sections \ref{sec:evaluationset} and \ref{sec:atleast}.}

\subsection{Idea and Key Observations}
\label{sec:qsla_to_sla_idea}
We reduce entailment checking in $\QSLA$ to entailment checking in $\SLA$, i.e.,
\begin{center}
	\emph{Given $\ff,\fg \in \QSLA$, we reduce checking $\ff \entails \fg$ to checking finitely many entailments of the form $\sla \entails \slb$ with $\sla,\slb \in\SLA$.}
\end{center}
We instantiate $\QSLA$ and $\SLA$, respectively, for the sake of concreteness. For that, we fix the set $\Predset$ of predicate symbols given by
\begin{align*}
	%
	%\Predset \quad \rightarrow \quad 
	\Predset ~=~ \{\: \true,~ \slemp,~\xx=\xy,~\xx\neq\xy,~\slsingleton{\xx}{\xy} ~\mid~ \xx,\xy \in \Vars \:\}~.
%\end{align}
%
%
%Moreover, we fix the set $\PurePredset$ of pure predicates by
%
%
%\begin{align*}
	%\qquad \qquad \PurePredset \quad \rightarrow \quad  \xx=\xy  ~\mid \xx\neq\xy ~.
\end{align*}
Now, consider the following entailment $\fh_1 \entails \fh_2$ as a running example:
\begin{align*}
	%\label{eqn:idea_running}
   \textcolor{gray}{\fh_1 \eeq} 
   0.4 \cdot (\underbrace{\singleton{\xx}{\xy} \sepcon \singleton{\xy}{\xz}}_{=\fg_1}) + 0.6 \cdot \underbrace{\singleton{\xx}{\xy} }_{=\fg_2}
   \eentails 0.6 \cdot (\singleton{\xx}{\xy} \sepcon \iverson{\true}) \textcolor{gray}{\eeq \fh_2}~.
\end{align*}
Intuitively speaking, $\fh_1$ expresses that with probability $0.4$ the heap consists of two cells where $\xx$ points to $\xy$ and separately $\xy$ points to $\xz$, and that with probability $0.6$ the heap consists of a single cell where $\xx$ points to $\xy$. Formula $\fh_2$ expresses that with probability $0.6$ the heap contains a cell where $\xx$ points to $\xy$. How can we reduce the problem of checking whether $\fh_1 \entails \fh_2$ holds to checking finitely many entailments in $\SLA$? 
We rely on two key observations: \\ \\
\noindent
\emph{Observation 1.} For every $\ff \in \QSLA$, the set 
\[
    \eval{\ff} \eeq \left\{ \sem{\ff}{(\sk,\hh)} ~\mid~ (\sk,\hh) \in \States \right\} ~{}\subset{}~\Probs
\]
is \emph{finite}. Moreover, there is an effectively constructible \emph{finite} and \emph{sound overapproximation} $\evaluationSet{\ff}$ of $\eval{\ff}$, i.e., $\eval{\ff} \subseteq \evaluationSet{\ff}$. 

\begin{example}
Consider the expectation $\fh_1$ from our running example: We have $\eval{\fh_1} = \{0,0.4,0.6\}$. We construct a finite overapproximation of $\eval{\fh_1}$ as follows: First, we observe that both subformulae $\fg_1$ and $\fg_2$ evaluate to a value in $\{0,1\}$, i.e, $\evaluationSet{\fg_1} = \evaluationSet{\fg_2} = \{0,1\}$. From $\evaluationSet{\fg_1}$ and $\evaluationSet{\fg_2}$, we obtain a finite overapproximation $\evaluationSet{\fh_1}$ of $\eval{\fh_1}$ given by
\[
\evaluationSet{\fh_1} \eeq \left\{0.4\cdot \alpha + 0.6\cdot \beta ~\mid \ra\in \evaluationSet{\fg_1},~\rb\in \evaluationSet{\fg_2} \right\} \eeq \{0,0.4,0.6,1\}~.
\]
Notice that $\evaluationSet{\fh_1}$ is a \emph{proper} superset of $\eval{\fh_1}$ since $1\not\in\eval{\fh_1}$. \hfill $\triangle$
\end{example}
We consider the construction of $\evaluationSet{\ff}$ for arbitrary \mbox{$\ff \in \QSLA$ in \Cref{sec:evaluationset}.} \\ \\
\noindent
\emph{Observation 2.} Given $\ff \in \QSLA$ and a probability $\ra \in \Probs$, there is an effectively constructible $\SLA$ formula, which we denote by $\atleast{\ra}{\ff}$, such that $(\sk,\hh)$ is a model of $\atleast{\ra}{\ff}$ if and only if $\ff$ evaluates at least to $\ra$ on state $(\sk,\hh)$, i.e.,
\[
     \underbrace{(\sk,\hh) \mmodels \atleast{\ra}{\ff}}_{\text{in}~ \SLA}
     \qquad \text{iff} \qquad  
     \underbrace{\ra \lleq \semapp{\ff}{(\sk,\hh)}}_{\text{in}~ \QSLA}~.
\]
We can thus lower bound $\QSLA$ formulae in terms of $\SLA$ formulae.
\begin{example}
Continuing our running example, we construct $\atleast{0.5}{\fh_1}$, i.e., an $\SLA$ formula evaluating to $\true$ on state $(\sk,\hh)$ if and only if $\fh_1$ evaluates at least to $0.5$. We start by considering the subformulae of $\fh_1$. Since both $\fg_1$ and $\fg_2$ embed $\SLA$ predicates, we have for every $\ra  \in \Probs$
\begin{align*}
    &\atleast{\ra}{\fg_1} \eeq
    \true~\text{if $\ra=0$ else}~\slsingleton{\xx}{\xy} \sepcon \slsingleton{\xy}{\xz} \\
    %\lambda \ra\mydot }eeq 
    %\begin{cases}
    %	\true, &\text{if}~\ra=0 \\
    	%\slsingleton{\xx}{\xy} \sepcon \slsingleton{\xy}{\xz},&\text{otherwise}
    %\end{cases} \\
  %
  %
  \text{and} \quad&
  \atleast{\ra}{\fg_2} \eeq
  \true~\text{if $\ra=0$ else}~\slsingleton{\xx}{\xy} ~.
%  \lambda \ra\mydot
%  \begin{cases}
%  	\true, &\text{if}~\rc=0 \\
%  	\slsingleton{\xx}{\xy},&\text{otherwise}~.
%  \end{cases} 
%
\end{align*}
The intuition is as follows: $\ra=0$ lower bounds every probability. Conversely, if $\ra >0$ then $\ra$ lower bounds $\fg_1$ (resp.\ $\fg_2$) on state $(\sk,\hh)$ if and only if $(\sk,\hh)$ satisfies the predicate $\fg_1$ (resp.\ $\fg_2$).
Now, when does $\fh_1$ evaluate at least to $0.5$? Given $\evaluationSet{\fg_1}$ and $\evaluationSet{\fg_2}$ and the fact that the valuation of $\fh_1$ is a convex combination of the valuations of $\fg_1$ and $\fg_2$, there are (at most) two cases: Either \emph{both} $\fg_1$ \emph{and} $\fg_2$ evaluate to (at least) $1$, or $\fg_2$ (but not necessarily $\fg_1$) evaluates to (at least) $1$. Given $\atleast{1}{\fg_1}$ and $\atleast{1}{\fg_2}$, the aforementioned informal disjunction translates to a formal \mbox{disjunction in $\SLA$:}
\begin{align*}
  \atleast{0.5}{\fh_1} \eeq&  
  \big(\atleast{1}{\fg_1} \wedge  \atleast{1}{\fg_2} \big) \vee \atleast{1}{\fg_2}
  \\
  \eeq &\big((\slsingleton{\xx}{\xy} \sepcon \slsingleton{\xy}{\xz}) \wedge  \slsingleton{\xx}{\xy}\big) \vee \slsingleton{\xx}{\xy}~.
\end{align*}
Notice that---as it is the case for $\evaluationSet{\fh_1}$---we construct $\atleast{0.5}{\fh_1}$ \emph{syntactically}. In particular, we disregard that the disjunct $(\slsingleton{\xx}{\xy} \sepcon \slsingleton{\xy}{\xz}) \wedge  \slsingleton{\xx}{\xy}$ is unsatisfiable and therefore equivalent to $\false$. \hfill $\triangle$
\end{example}
We provide the construction of $\atleast{\ra}{\ff}$ for \emph{arbitrary} $\QSLA$ formulae $\ff$---including quantitative quantifiers and the magic wand---in \Cref{sec:atleast}. \\

%Observations $1$ and $2$ imply that every formula $\ff \in \QSLA$ is expressible in terms of finitely many $\SLA$ formulae---a fact we will formalize by \Cref{def:sup_formulae} and \Cref{thm:translation_sound}---since we have
%%To see this, recall that $\evaluationSet{\ff}$ is finite and overapproximates $\eval{\ff}$. Hence, we have for every $(\sk,\hh)\in\States$, 
%%
%\begin{align*}
%    \semapp{\ff}{(\sk,\hh)} 
%    %
%    \eeq & \sup \left\{ \ra \in \evaluationSet{\ff} ~\mid~ \ra \leq  \semapp{\ff}{(\sk,\hh)}  \right\}
%    \tag{by Observation $1$}  \notag \\
%    %
%    \eeq &\sup \left\{ \ra \in \evaluationSet{\ff} ~\mid~ (\sk,\hh) \models \formulaFunc{\ff}(\ra) \right\}~.
%    \tag{by Observation $2$}
%\end{align*}
%
%
Finally, Observations $1$ and $2$ together yield our reduction from $\ff \entails \fg$ to finitely many entailments in $\SLA$. Intuitively speaking, we formalize that 
\begin{center}
	\emph{whenever $\ff$ evaluates at least to $\ra$, then $\fg$ too evaluates at least to $\ra$}
\end{center}
equivalently in terms of finitely many $\SLA$ entailments. Put more formally, since $\evaluationSet{\ff}$ is finite, we have
\begin{align*}
	&\ff \entails \fg \\
	\text{iff}\quad &\text{for all $(\sk,\hh)$}\colon \semapp{\ff}{(\sk,\hh)} \leq \semapp{\fg}{(\sk,\hh)} 
	\tag{by definition}\\
	\text{iff}\quad& \text{for all $(\sk,\hh)$ and all $\alpha \in \evaluationSet{\ff}$}\colon \ra \leq \semapp{\ff}{(\sk,\hh)} 
	~\text{implies}~ \alpha \leq \semapp{\fg}{(\sk,\hh)} 
	\tag{by Observation $1$}\\
	\text{iff}\quad& \text{for all $(\sk,\hh)$ and all $\alpha \in \evaluationSet{\ff}$}\colon (\sk,\hh)\models\atleast{\ra}{\ff}
	~\text{implies}~(\sk,\hh)\models\atleast{\ra}{\fg}
	\tag{by Observation $2$}\\
	\text{iff}\quad& \text{for all $\alpha \in \evaluationSet{\ff}$}\colon \atleast{\ra}{\ff} \entails \atleast{\ra}{\fg}~.
	\tag{by definition}
\end{align*}
\begin{example}
	Reconsider our running example. Since $\sizeof{\evaluationSet{\fh_1}} = 4$, the $\QSLA$ entailment $\fh_1 \entails\fh_2$ is equivalent to the four entailments 
	\[
	   \atleast{\ra}{\fh_1}\entails\atleast{\ra}{\fh_2}\quad\text{for $\ra \in \{0,0.4,0.6,1\}$}
	 \]
	in $\SLA$, each of which actually holds. \hfill $\triangle$
\end{example}
%
%
%For our running example $\fh_1 \entails\fh_2$, we have
%%
%\begin{align*}
%	&\fh_1 \entails \fh_2  \\
%	%
%	\text{iff}\quad & \text{for all $\alpha \in \{0,0.4,0.6,1\}$}\colon \formulaFunc{\fh_1}(\alpha) \entails \formulaFunc{\fh_2}(\alpha) \\
%	%
%	\text{iff}\quad & \text{for all $\alpha \in \{0,0.4,0.6,1\}$}\colon \true \entails \true
%	%
%	%
%\end{align*}

\subsection{Constructing Finite Overapproximations of $\eval{\ff}$}
\label{sec:evaluationset}

\begin{table}[t]
	\centering
	\caption{Inductive definition of $\evaluationSet{\ff}$.}
	%\kbcommentinline{TODO: 3rd column for $\sizeof{\evaluationSet{\ff}}$? @Ira}}%
	\label{tab:evluationsets}%
	\renewcommand{\arraystretch}{1.4}%
		\begin{tabular}[t]{l@{\qquad\qquad}l@{\qquad}l@{\qquad}l}
		\hline\hline
		$\ff \in \QSLA$			&          $\evaluationSet{\ff} \subset \Probs $			\\
		\hline
		$\iverson{\slb}$		& $\{0,1\}$	              \\
		$\iverson{\bb} \cdot \fg + \iverson{\neg\bb} \cdot \fh $ & $\evaluationSet{\fg} \cup \evaluationSet{\fh}$ \\
		$\qq \cdot \fg + (1-\qq) \cdot \fh $ & $\pp\cdot \evaluationSet{\fg} + (1-\pp)\cdot \evaluationSet{\fh}$ \\
		$\fg \cdot \fh $  & $\evaluationSet{\fg} \cdot \evaluationSet{\fh}$ \\
		$1-\fg$ & $1-\evaluationSet{\fg}$  \\
		$\emax{\fg}{\fh}$ & $\emax{\evaluationSet{\fg}}{\evaluationSet{\fh}}$ \\
		$\emin{\fg}{\fh}$ & $\emin{\evaluationSet{\fg}}{\evaluationSet{\fh}}$ \\
		$\Sup \xx \colon \fg $ & $\evaluationSet{\fg}$\\
		$\Inf \xx \colon \fg $ & $\evaluationSet{\fg}$\\
		$\fg \sepcon \fh $ & $\evaluationSet{\fg} \cdot \evaluationSet{\fh}$ \\
		$\iverson{\slb} \sepimp \fg$ & $\evaluationSet{\fg}$ \\
		\hline\hline
	\end{tabular}%
\renewcommand{\arraystretch}{1}%
	
\end{table}
We consider the formal construction underlying Observation 1 from the previous section, i.e., given $\ff \in \QSLA$, we provide a syntactic construction of a finite overapproximation $\evaluationSet{\ff}$ of $\eval{\ff}$. This construction is by induction on the structure of $\ff$ as shown in \Cref{tab:evluationsets}. For that, we define some shorthands. Given $\ra \in \Probs$,  $V,W \subseteq \Probs$, and a binary operation $\circ \colon \Probs \times\Probs \to \Probs$, we define
\begin{align*}
	\ra \cdot V \eeq \left\{  \ra \cdot \rb~\mid~ \rb \in V\right\} 
	\quad\text{and}\quad
	V\circ W \eeq \left\{  \rb \circ\rc~\mid~ \rb\in V, ~\rc\in W\right\}~.
\end{align*}
Let us now go over the individual cases. \\ \\
 \noindent
 \emph{The case $\ff = \iverson{\slb}$.} We have $\iverson{\slb}(\sk,\hh) \in\{0,1\}$ by definition.  \\ \\
 %
 %
% \noindent
% For the composite cases now assume that for some arbitrary, but fixed, $\fg,\fh \in \QSLA$ we have constructed finite overapproximations $\evaluationSet{\fg}$ and $\evaluationSet{\fh}$ of $\eval{\fg}$ and $\eval{\fh}$, respectively. \\ \\
 %
 %
 \noindent
\emph{The case $\ff = \iverson{\bb} \cdot \fg + \iverson{\neg\bb} \cdot \fh$.} For every $(\sk,\hh)$, the formula $\ff$ \emph{either} evaluates to $\semapp{\fg}{(\sk,\hh)}$ \emph{or} to $\semapp{\fh}{(\sk,\hh)}$, depending on whether $(\sk,\hh) \models \bb$ holds. \\ \\
% Hence, we set
%
%\[
%     \evaluationSet{\ff} \eeq \evaluationSet{\fg} \cup \evaluationSet{\fh}~.
%\]
%
%
\noindent
\emph{The case $\ff = \pp \cdot \fg + (1-p) \cdot \fh$.} 
The formula $\ff$ evaluates to $\pp \cdot \ra + (1-\pp)\cdot \rb$ for some $\ra \in \evaluationSet{\fg}$ and $\rb\in\evaluationSet{\fh}$.  \\ \\
%Hence, we set
%%
%\[
%    \evaluationSet{\ff} \eeq \pp\cdot\evaluationSet{\fg} + (1-\pp)\cdot\evaluationSet{\fh}~.
%\]
%
\noindent
\emph{The case $\ff =  \fg \cdot\fh$ or $\ff = \fg \sepcon\fh$.} The formula $\ff$ evaluates to $\ra\cdot \rb$ for some $\ra \in \evaluationSet{\fg}$ and $\rb\in\evaluationSet{\fh}$.  \\ \\
%Hence, we set
%%
%\[
%\evaluationSet{\ff} \eeq \evaluationSet{\fg}\cdot\evaluationSet{\fh}~.
%\]
%
\noindent
\emph{The case $\ff =  1-\fg$.} The formula $\ff$ evaluates to $1- \ra$ for some $\ra \in \evaluationSet{\fg}$. \\ \\
% Hence, we set
%
%\[
%\evaluationSet{\ff} \eeq \{1-\ra ~\mid~ \ra\in \evaluationSet{\fg} \}~.
%\]
%
\noindent
\emph{The case $\ff =  \fg \circ \fh$ for $\circ \in \{\max,\min \}$.} Since $\max$ and $\min$ are defined point-wise, the formula $\ff$ evaluates to some value $\ra \circ\rb$ for $\ra \in \evaluationSet{\fg},\rb\in\evaluationSet{\fh}$. \\ \\
%
%\[
%   \evaluationSet{\ff} \eeq \evaluationSet{\fg} \circ \evaluationSet{\fh}~.
%\]
%
\noindent
\emph{The case $\ff =  \Sup \xx \colon \fg$ or $\ff = \Inf \xx \colon \fg$.} Since $\evaluationSet{\fg}$ overapproximates the set of \emph{all} valuations of $\fg$, quantitative quantifiers do not add any valuation. \\ \\
%Hence, we set
%%
%\[
%\evaluationSet{\ff} \eeq \evaluationSet{\fg} ~.
%\]
%
\noindent
\emph{The case $\ff =  \iverson{\slb}\sepimp \fg$.} Recall that
\[
   \semapp{\ff}{(\sk,\hh)} \eeq \inf \left\{ \semapp{\fg}{(\sk,\hh\sepcon\hh')}  ~\mid~ \hh' \disjoint \hh ~\text{and}~ \iverson{\slb}(\sk,\hh')=1 \right\}~.
\]
If the above set is \emph{non}-empty, the infimum is actually a minimum and therefore $\ff$ evaluates to some value in $\evaluationSet{\fg}$. If the above set \emph{is} empty, then $\semapp{\ff}{(\sk,\hh)} =1$. It is easy to verify that $1$ is necessarily an element of $\evaluationSet{\fg}$ (\cf\Cref{lem:zero_one_preservation}). \\ \\
\noindent
Summarizing our considerations on $\evaluationSet{\ff}$, we get:
\begin{theorem}
	\label{thm:evalset}
	For every $\ff \in \QSLA$, the effectively constructible set $\evaluationSet{\ff} \subset \Probs$ given in \Cref{tab:evluationsets} satisfies
	\[
	\sizeof{\evaluationSet{\ff}}< \infty \quad \text{and}\quad \eval{\ff}\subseteq \evaluationSet{\ff}~.
	\]
\end{theorem}
\begin{proof}
	Straightforward by induction on $\ff$.
\end{proof}
\subsection{Lower Bounding $\QSLA$ by $\SLA$ Formulae}
\label{sec:atleast}

\begin{table}[t]
	\centering
	\caption{Inductive definition of $\atleast{\ra}{\ff}$ for a \emph{given} $\ra \in \Probs$.}
	\label{tab:at_least}%
	\renewcommand{\arraystretch}{1.4}%
		\begin{tabular}[t]{l@{\qquad\qquad}l@{\qquad}l@{\qquad}l}
		\hline\hline
		$\ff \in \QSLA$			&          $\atleast{\ra}{\ff} \in \SLA$			\\
		\hline
		$\iverson{\slb}$		& $\true~\text{if}~\ra=0~\text{else}~\slb $         \\
		$\iverson{\bb} \cdot \fg + \iverson{\neg\bb} \cdot \fh $ & $(\bb \wedge \atleast{\ra}{\fg}) \vee (\neg\bb  \wedge \atleast{\ra}{\fh})$\\
		$\qq \cdot \fg + (1-\qq) \cdot \fh $ & $\bigvee_{\rb \in \evaluationSet{\fg}, \rc \in  \evaluationSet{\fh}, \pp\cdot\rb +(1-\pp)\cdot\rc \geq \ra } \quad \atleast{\rb}{\fg} \wedge \atleast{\rc}{\fh}$ \\
		$\fg \cdot \fh $  &$\bigvee_{\rb \in \evaluationSet{\fg}, \rc \in  \evaluationSet{\fh}, \rb\cdot\rc \geq \ra }\quad  \atleast{\rb}{\fg} \wedge \atleast{\rc}{\fh}$\\
		$1-\fg$ & $\true ~\text{if}~\ra=0~\text{else}~\neg~\atleast{\rd}{\ff}$  \\
				& \quad for $\rd = \min \setcomp{\rb \in \evaluationSet{\fg}}{\rb > 1-\ra}$ \\
		$\emax{\fg}{\fh}$ & $\atleast{\ra}{\fg} \vee\atleast{\ra}{\fh}$ \\
		$\emin{\fg}{\fh}$ & $\atleast{\ra}{\fg} \wedge \atleast{\ra}{\fh}$ \\
		$\Sup \xx \colon \fg $ & $\exists \xx \colon \atleast{\ra}{\fg}$\\
		$\Inf \xx \colon \fg $ & $\forall \xx \colon \atleast{\ra}{\fg}$\\
		$\fg \sepcon \fh $ & $\bigvee_{\rb \in \evaluationSet{\fg}, \rc \in  \evaluationSet{\fh}, \rb \cdot\rc \geq \ra } \quad \atleast{\rb}{\fg} \sepcon\atleast{\rc}{\fh}$ \\
		$\iverson{\slb} \sepimp \fg$ & $\slb \sepimp \atleast{\ra}{\fg}$ \\
		\hline\hline
	\end{tabular}%
\renewcommand{\arraystretch}{1}%
	
\end{table}

We now consider the formal construction underlying Observation $2$ from \Cref{sec:qsla_to_sla_idea}. That is, given $\ff \in \QSLA$ and $\ra \in \Probs$, we provide the syntactic construction of an $\SLA$ formula $\atleast{\ra}{\ff}$ evaluating to $\true$ on state $(\sk,\hh)$ if and only if $\ff$ evaluates at least to $\ra$ on $(\sk,\hh)$. This construction relies on $\evaluationSet{\ff}$ from the previous section and is given by induction on the structure of $\ff$ as shown in \Cref{tab:at_least}. We consider the individual constructs.  For that, we fix some state $(\sk,\hh)$. \\ \\
\noindent
\emph{The case $\ff = \iverson{\slb}$.} There are two cases. If $\ra=0$, then $\ra$ trivially lower bounds the value of $\iverson{\slb}$. Conversely, if $\ra >0$, then $\ra$ lower bounds $\iverson{\slb}$ on state $(\sk,\hh)$ if and only if $(\sk,\hh)$ satisfies $\slb$. \\ \\
%
%\begin{align*}
%	\ra \leq \semapp{\iverson{\preda}}{(\sk,\hh)}
%	\qquad\text{iff}\qquad
%	\semapp{\iverson{\preda}}{(\sk,\hh)}=1
%	\qquad\text{iff}\qquad
%	(\sk,\hh)\models \preda~.
%\end{align*}
%
%
\noindent
%For the composite cases, now assume that for some arbitrary, but fixed, $\fg,\fh \in \QSLA$ and all $\rb,\rc \in \Probs$ there are effectively constructible $\SLA$ formulae $\atleast{\rb}{\fg}$ and $\atleast{\rc}{\fh}$ such that for all $(\sk,\hh) \in \States$,
%%
%\begin{align*}
%	\rb \leq \semapp{\fg}{(\sk,\hh)} ~\text{iff}~(\sk,\hh) \models \atleast{\rb}{\fg}
%	 \qquad \text{and} \qquad
%	\rc \leq \semapp{\fh}{(\sk,\hh)} ~\text{iff}~ (\sk,\hh) \models \atleast{\rc}{\fh}~.
%\end{align*}
%%
%Recall furthermore that by \Cref{thm:evalset} there are effectively constructible finite sets $\evaluationSet{\fg},\evaluationSet{\fh}$ covering all values $\fg$ and $\fh$ evaluate to. \\ \\
For the composite cases, recall that by \Cref{thm:evalset} there are effectively constructible finite sets $\evaluationSet{\fg},\evaluationSet{\fh}$ covering all values $\fg$ and $\fh$ evaluate to. \\ \\
\noindent
\emph{The case $\ff = \iverson{\bb} \cdot \fg + \iverson{\neg\bb} \cdot \fh$.} The formula $\ff$ represents a Boolean choice between the formulae $\fg$ and $\fh$, depending on the truth value of $\bb$. Hence, there are two cases: If $(\sk,\hh)$ \emph{does} satisfy $\bb$, then $\ra$ lower bounds $\ff$ iff $\ra$ lower bounds $\fg$. Conversely, if $(\sk,\hh)$ does \emph{not} satisfy $\bb$, then $\ra$ lower bounds $\ff$ iff $\ra$ \mbox{lower bounds $\fh$.} \\ \\
%
%\begin{align*}
%	 &\ra \leq \semapp{\iverson{\BB} \cdot \fg + \iverson{\neg\BB}\cdot\fh }{(\sk,\hh)} \\
%	 %
%	 \text{iff}\qquad& (\sk,\hh)\models \BB~\text{and}~ \ra\leq \semapp{\fg}{(\sk,\hh)}
%	                             \quad\text{or}\quad 
%	                              (\sk,\hh)\models \neg\BB~\text{and}~ \ra\leq \semapp{\fh}{(\sk,\hh)} \\
%	 %
%	 %
%	 \text{iff}\qquad& (\sk,\hh)\models \BB~\text{and}~ (\sk,\hh) \models \atleast{\ra}{\fg}
%	 \quad\text{or}\quad 
%	 (\sk,\hh)\models \neg\BB~\text{and}~ (\sk,\hh) \models \atleast{\ra}{\fh} \\
%	 %
%	 %
%	 \text{iff}\qquad& (\sk,\hh)\models (\BB \wedge \atleast{\ra}{\fg}) \vee (\neg\BB  \wedge \atleast{\ra}{\fh})
%\end{align*}
%
\noindent
\emph{The case $\ff = \pp \cdot \fg + (1-p) \cdot \fh$.} Since $\evaluationSet{\fg}$ and $\evaluationSet{\fh}$ cover every possible valuation of $\fg$ and $\fh$, respectively, it follows that $\ra$ lower bounds the valuation of $\ff$  if and only if there are $\rb \in \evaluationSet{\fg}$ and $\rc\in\evaluationSet{\fh}$ such that (1) $\rb$ lower bounds $\fg$, (2) $\rc$ lower bounds $\fh$, and (3) $\ra$ lower bounds the convex sum $\pp\cdot\rb+(1-\pp)\cdot\rc$. \\ \\
%
%This gives us
%%
%\begin{align*}
%	&\ra \leq \semapp{\pp \cdot \fg + (1-p) \cdot \fh}{(\sk,\hh)} \\
%	%
%	\text{iff}\qquad&\ra \leq \pp\cdot \semapp{\fg}{(\sk,\hh)} + (1-p) \cdot\semapp{\fh}{(\sk,\hh)}  \\
%	%
%	\text{iff}\qquad & \text{there are $\rb \in \evaluationSet{\fg}, \rc\in\evaluationSet{\fh}$ with $p\cdot \rb + (1-p)\cdot \rc\geq \ra$ with } \\
%	                         & \rb\leq \semapp{\fg}{(\sk,\hh)} ~\text{and}~\rc\leq \semapp{\fh}{(\sk,\hh)} \\
%	%
%	\text{iff}\qquad & \text{there are $\rb \in \evaluationSet{\fg}, \rc\in\evaluationSet{\fh}$ with $p\cdot \rb + (1-p)\cdot \rc\geq \ra$ with } \\
%	& (\sk,\hh) \models \atleast{\rb}{\fg} ~\text{and}~(\sk,\hh) \models \atleast{\rc}{\fh} \\
%	%
%	%
%	\text{iff}\qquad & \text{there are $\rb \in \evaluationSet{\fg}, \rc\in\evaluationSet{\fh}$ with $p\cdot \rb + (1-p)\cdot \rc\geq \ra$ with } \\
%	& (\sk,\hh) \models \atleast{\rb}{\fg} ~\text{and}~(\sk,\hh) \models \atleast{\rc}{\fh} \\
%	%
%	\text{iff}\qquad & (\sk,\hh)\models \bigvee_{\rb \in \evaluationSet{\fg}, \rc \in  \evaluationSet{\fh}, \pp\cdot\rb +(1-\pp)\cdot\rc \geq \ra } \quad \atleast{\rb}{\fg} \wedge \atleast{\rc}{\fh}~.
%\end{align*}
%\\ \\
%
%
%
\noindent
\emph{The case $\ff =  \fg \cdot\fh$.} The reasoning is analogous to the previous case. \\ \\
%Hence, we set
%%
%\[
%\evaluationSet{\ff} \eeq \evaluationSet{\fg}\cdot\evaluationSet{\fh}~.
%\]
%
\noindent
\emph{The case $\ff =  1-\fg$.} We write $\ra \leq \semapp{1-\fg}{(\sk, \hh)}$ equivalently as $\neg (1-\ra < \semapp{\fg}{(\sk,\hh)})$. In order to turn the strict inequality into a non-strict one, we consider the successor $\rd$ of $1-\ra$ in $\evaluationSet{\fg}$, \ie the smallest $\rd$ in $\evaluationSet{\fg}$ greater than $1-\ra$. Since $\evaluationSet{\fg}$ is finite, such a $\rd$ always exists if $1-\ra \neq 1$. We illustrate the idea in the following picture, where all elements in $\evaluationSet{\fg}$ are marked by \textbullet{}.
\begin{center}
	\begin{tikzpicture}
		\draw[-latex] (-3.5,0) -- (4,0) ; %number line with direction to left
		\foreach \x in {-3.5, -3, -2.5, -1.7, 1, 1.6, 2.4, 3.1, 3.5} %evaluation set
			\node at (\x,0) {\textbullet};
		
		%number zero
		\draw[shift={(-3.5,0)},color=black] (0pt,7pt) -- (0pt,-7pt);	
		\node at (-3.5,-0.5) {$0$};

		%number one
		\draw[shift={(3.5,0)},color=black] (0pt,7pt) -- (0pt,-7pt);	
		\node at (3.5,-0.5) {$1$};

		%number ra
		\draw[shift={(0,0)},color=black] (0pt,7pt) -- (0pt,-7pt);	
		\node at (0,-0.5) {$1-\ra$};

		%number rd
		\draw[shift={(1,0)},color=black] (0pt,7pt) -- (0pt,-7pt);	
		\node at (1,-0.5) {$\rd$};

		%number [[g]](s,h)
		\draw[shift={(2.4,0)},color=black] (0pt,7pt) -- (0pt,-7pt);	
		\node at (2.3,-0.5) {$\semapp{\fg}{(\sk,\hh)}$};
	\end{tikzpicture}
\end{center}
%
%\tncommentinline{perhaps better: \enquote{we ran out of possible evaluations that are strictly lower bounded by $1-\ra$?}} done :)
%
For the successor $\rd$, checking if $\rd$ is a lower bound of $\semapp{\fg}{(\sk, \hh)}$ is equivalent to checking if $1-\ra$ is a strict lower bound - if  $\rd$ is not a lower bound, then we ran out of possible valuations that are strictly lower bounded by $1-\ra$.
\\ \\
\noindent
\emph{The case $\ff =  \fg \circ \fh$ for $\circ \in \{\max,\min \}$.} The probability $\ra$ lower bounds the maximum of $\fg$ and $\fh$ on state $(\sk,\hh)$ if and only if $\ra$ lower bounds $\fg$ \emph{or} $\ra$ lower bounds $\fh$. For $\circ=\min$, the reasoning is dual.
\\ \\
\noindent
\emph{The case $\ff =  \Sup \xx \colon \fg$.}  Recall that
\[
 \semapp{\ff}{(\sk,\hh)} \eeq \max \big\{\semapp{\fg}{(\sk\statesubst{\xx}{v},\hh)} ~\mid~ v \in \Vals \big\}~.
\]
Now observe that $\ra$ lower bounds the above maximum if and only if $\ra$ lower bounds \emph{some} element of the above set, i.e., if and only if there is some $v$ with
\[
\ra \lleq \semapp{\fg}{(\sk\statesubst{\xx}{v},\hh)}
\qquad\text{which is equivalent to}\qquad 
(\sk,\hh) \mmodels \exists \xx \colon \atleast{\ra}{\fg}~.
\] 
%
% Furthermore $\evaluationSet{\fg}$ is finite, so is the above set. Hence
%
%
\noindent
\emph{The case $\ff = \Inf \xx \colon \fg$.}  Recall that
\[
\semapp{\ff}{(\sk,\hh)} \eeq \min \big\{\semapp{\fg}{(\sk\statesubst{\xx}{v},\hh)} ~\mid~ v \in \Vals \big\}~.
\]
Since $\ra$ lower bounds the above minimum if and only if $\ra$ lower bounds \emph{all} elements of the above set, the reasoning is dual to the previous case. \\ \\
\noindent 
\emph{The case $\ff= \fg\sepcon\fh$.} Recall that
\[
\semapp{\ff}{(\sk,\hh)} \eeq \max \left\{\semapp{\fg}{(\sk,\hh_1)} \cdot \semapp{\fh}{(\sk,\hh_2)} ~\mid~ \hh = \hh_1 \sepcon \hh_2 \right\}~.
\]
Since $\evaluationSet{\fg}$ and $\evaluationSet{\fh}$ cover every possible valuation of $\fg$ and $\fh$, respectively, $\ra$ lower bounds the evaluation of $\ff$ on $(\sk,\hh)$ iff there are $\rb\in \evaluationSet{\fg}, \rc\in\evaluationSet{\fh}$ and $\hh_1,\hh_2 $ with $\hh_1 \sepcon \hh_2 = \hh$ such that (1) $\rb$ lower bounds $\fg$ on $(\sk,\hh_1)$, (2) $\rc$ lower bounds $\fh$ on $(\sk,\hh_2)$, and (3) $\ra$ lower bounds $\rb\cdot\rc$. Given such $\rb$ and$\rc$, we can phrase this equivalently in $\SLA$ as 
\[
   (\sk,\hh) \mmodels \atleast{\rb}{\fg} \sepcon\atleast{\rc}{\fh}~.
\]
\noindent
\emph{The case $\ff =  \iverson{\slb}\sepimp \fg$.} Recall that
\[
\semapp{\ff}{(\sk,\hh)} \eeq \inf \left\{ \semapp{\fg}{(\sk,\hh\sepcon\hh')}  ~\mid~ \hh' \disjoint \hh ~\text{and}~ \iverson{\slb}(\sk,\hh')=1 \right\}~.
\]
Probability $\ra$ lower bounds the above infimum if and only if for \emph{every} extension $\hh'$ of the heap $\hh$ such that the stack $\sk$ together with $\hh'$ satisfy $\slb$,
probability  $\ra$ is a lower bound on $\semapp{\fg}{(\sk,\hh\sepcon\hh')}$. Put more formally, the latter statement reads
\[
   \text{for all $\hh' \disjoint \hh$ with $(\sk,\hh')\models\slb$}\colon~ (\sk,\hh\sepcon\hh')\models \atleast{\ra}{\fg}~,
\]
which is equivalent to $(\sk,\hh)\models\slb \sepimp\atleast{\ra}{ \fg}$. \\

\noindent
Our construction thus applies to arbitrary $\QSLA$ formulae and we get:
\begin{theorem}
	\label{thm:atleast}
	For every $\ff \in \QSLA$ and all $\ra  \in \Probs$ there is an effectively constructible $\SLA$ formula $\atleast{\ra}{\ff}$ such that for all $(\sk,\hh) \in \States$, we have
	\[
	  (\sk,\hh) \mmodels \atleast{\ra}{\ff}
		\qquad\text{iff}\qquad
	  \ra\lleq \semapp{\ff}{(\sk,\hh)}~.
	\]
\end{theorem}
\begin{proof}
	By induction on $\ff$. 
	See \Cref{app:entailment_checking} for details. \cmcomment{TR / App.}
\end{proof}
Finally, we obtain our main theorem.
\begin{theorem}
	\label{thm:qsla_by_sla}
	Entailment checking in $\QSLA$ reduces to entailment checking in $\SLA$, i.e,
	for all $\ff,\fg \in \QSLA$, we have
	\[
	   \ff\eentails\fg 
	   \qquad\text{iff}\qquad 
	   \text{for all $\alpha \in \evaluationSet{\ff}$}\colon ~ \atleast{\ra}{\ff} \eentails \atleast{\ra}{\fg}~.
	\]
\end{theorem}
\begin{proof}
	Follows from Theorems \ref{thm:evalset} and \ref{thm:atleast} and the reasoning \mbox{at the end of \Cref{sec:qsla_to_sla_idea}.}
\end{proof}
\begin{remark}[Avoiding $\true$ in $\SLA$ entailments] \label{remark:true_in_sla}
	 Formulae of the form $\atleast{\ra}{\ff} \in\SLA$ may introduce the atom $\true$, which is not admitted by some decidable separation logic fragments, 
     such as~\cite{Iosif2013Tree}. %can be problematic for some separation logic solvers. 
Fortunately, we can avoid $\true$ in $\atleast{\ra}{\ff}$ formulae. $\true$ is only required in formulae of the form $\atleast{0}{\ff}$, which arise in two situations when applying \Cref{thm:qsla_by_sla}: (1) in entailment checks of the form $\atleast{0}{\ff} \entails \atleast{0}{\fg}$, which always hold and can thus be omitted, and (2) if $\ff=\pp \cdot \fg + (1-\pp) \cdot \fh$.  In the latter case, if we have  $\ra \neq 0$ in 
	\[\atleast{\ra}{\ff} \eeq \bigvee_{\rb \in \evaluationSet{\fg}, \rc \in  \evaluationSet{\fh}, \pp\cdot\rb +(1-\pp)\cdot\rc \geq \ra } \quad \atleast{\rb}{\fg} \wedge \atleast{\rc}{\fh}~, \]
	then either $\rb \neq 0$ or $\rc \neq 0$ holds for every disjunct. Hence, subformulae of the form $\atleast{0}{\fg}$ or $\atleast{0}{\fh}$ can be omitted, as well. \hfill $\triangle$
	%For other cases, the statement is even stronger, take for example $\ff= \fg \sepcon \fh$ then for $\ra \neq 0$ the formula $\atleast{\ra}{\ff}$ will not have any occurrences of $\atleast{0}{\fg}$ or $\atleast{0}{\fh}$. By this reasoning, we can eliminate all occurrences of $\true$.
\end{remark}
%
%We are now able to verify the open quantitative entailment from \Cref{sec:motivation:swap} using automatable techniques. First, we transform the quantitative entailment into three qualitative entailments using \Cref{thm:qsla_by_sla} and \Cref{remark:true_in_sla}. Secondly, we compile these entailments into the SMT-LIB SL format and, lastly, check them using the separation logic extension of \mbox{CVC4 \cite{Reynolds2016Decision}}. The quantitative entailments and the SMT-LIB SL file can be found in \Cref{app:swap_example}.

\section{Complexity}
\label{sec:complexity}
We now analyze the complexity of our approach. Recall that \Cref{thm:qsla_by_sla} reduces checking $\ff \entails \fg$ in $\QSLA$ to checking 
	\begin{equation*}
	\text{for all}~\ra \in \evaluationSet{\ff}\colon ~ \atleast{\ra}{\ff} \eentails \atleast{\ra}{\fg}
 \end{equation*}
in $\SLA$. We consider two aspects: (1) the \emph{number of $\SLA$ entailments} and (2) the \emph{size of the resulting $\SLA$ formulae} occurring in each entailment. 
%, and (3) the time complexity of computing $\QSLtoSNF{\ff}$.
We express these quantities in terms of the size of a $\QSLA$ formula $\ff$ and a $\SLA$ formula $\sla$ and denote them as $\sizeof{\ff}$ and $\sizeof{\sla}$ respectively. In these sizes, we count every construct in the formula and require that the size of atoms are defined at instantiation. Moreover, we assume that every atom in $\Predset$ is at least of size $1$ and especially the atom $\true$ is of size $1$. Additionally we count in an $\QSLA$ formula $\ff$ the constructs that increase the number of possible evaluation results of $\ff$, namely $\qq \cdot \fg + (1-\qq) \cdot \fh$, $\fg \cdot \fh$ and $\fg \sepcon \fh$, and denote it as $\probsizeof{\ff}$.\footnote{For a formal definition see \Cref{tab:sizes_sl_qsl} on \cpageref{tab:sizes_sl_qsl}.}

We will see that for an entailment $\ff \entails \fg$ in $\QSLA$, (1) the number of $\SLA$ entailments is in $2^{\bigo(\probsizeof{\ff})}$ in the worst case (see \Cref{thm:evalset_size}) and (2) the size of the resulting $\SLA$ formulae are in $\bigo(\sizeof{\ff}) \cdot 2^{\bigo(\probsizeof{\ff}^2)}$ and $\bigo(\sizeof{\fg}) \cdot 2^{\bigo(\probsizeof{\fg}^2)}$ respectively in the worst case (see \Cref{thm:atleast_size}). Now let us assume we have an entailment checker for $\SLA$ formulae that can solve entailments of the form $\atleast{\ra}{\ff} \entails \atleast{\ra}{\fg}$ and which has a runtime complexity of $\SLRuntime(n,m)$ where $n$ and $m$ are the size of $\SLA$ formulae on the left and right side of an entailment respectively. Putting the above together, checking the entailment $\ff \entails \fg$ in $\QSLA$ then has a runtime complexity of
\begin{gather*}  
	2^{\bigo(\probsizeof{\ff})} \cdot \SLRuntime\left(\bigo (\sizeof{\ff}) \cdot 2^{\bigo(\probsizeof{\ff}^2)}, \bigo(\sizeof{\fg}) \cdot 2^{\bigo(\probsizeof{\fg}^2)}\right)\\
	+ \bigo (\sizeof{\ff}) \cdot 2^{\bigo(\probsizeof{\ff}^2)} + \bigo(\sizeof{\fg}) \cdot 2^{\bigo(\probsizeof{\fg}^2)}~.
\end{gather*}
If we furthermore reasonably assume that $\SLRuntime(n,m)$ is at least linear in both arguments (otherwise the entailment checker can only check trivial entailments anyway), the runtime complexity simplifies to
\[ 2^{\bigo(\probsizeof{\ff})} \cdot \SLRuntime\left(\bigo (\sizeof{\ff}) \cdot 2^{\bigo(\probsizeof{\ff}^2)}, \bigo(\sizeof{\fg}) \cdot 2^{\bigo(\probsizeof{\fg}^2)}\right)~. \]
As for aspect (1), we first observe that checking $\ff \entails \fg$  by means of \Cref{thm:qsla_by_sla} requires checking $\sizeof{\evaluationSet{\ff}}$ entailments in $\SLA$. However, only the constructs we count with $\probsizeof{\ff}$ increase the number of possible evaluations, which in turn will also increase the size of the overapproximation $\evaluationSet{\ff}$. Every time any of these constructs occur, the number of possible evaluations $\eval{\ff}$ may double. Consequently, also the overapproximation $\evaluationSet{\ff}$ doubles in size when any of these constructs occur. Other constructs do not increase the number of evaluations, but instead inherit the evaluations from their subformulae. 
\begin{theorem}
	\label{thm:evalset_size}
	 We have $\sizeof{\evaluationSet{\ff}} \leq 2^{\probsizeof{\ff}+1}$. Hence, checking $\ff \entails \fg$ by means of \Cref{thm:qsla_by_sla} requires checking $2^{\bigo(\probsizeof{\ff})}$ entailments in $\SLA$.
\end{theorem}
\begin{proof}
	By induction on $\ff$. 
	For details see \Cref{app:complexity}.\cmcomment{TR / App.}
\end{proof}
For the size of the resulting $\SLA$ formulae, i.e., aspect (2), recall that we construct entailments of the form
\[
   \atleast{\ra}{\ff} \eentails \atleast{\ra}{\fg}~.
\]
We thus determine an upper bound on the size of any $\SLA$ formula $\atleast{\ra}{\ff}$. Here we make a similar observation as in aspect (1): whenever one of the constructs we count with $\probsizeof{\ff}$ appears, the size of the formula \emph{increases by the exponential factor $\sizeof{\evaluationSet{\ff}}$}. Such a multiplication of increasing exponential expressions then results asymptotically in a squared exponent. The other constructs increase the size by only a constant per construct. By combining both observations we can finally conclude an upper bound on the size of the formula $\atleast{\ra}{\ff}$.
\begin{theorem}\label{thm:atleast_size}
	For any formula $\ff \in \QSLA$ and all probabilities $\ra \in \Probs$, the $\SLA$ formula $\atleast{\ra}{\ff}$ has at most size $ 3 \cdot \sizeof{\ff} \cdot 2^{(\probsizeof{\ff}+1)^2}$. Hence the size of the formula $\atleast{\ra}{\ff}$ is in $\bigo(\sizeof{\ff}) \cdot 2^{\bigo(\probsizeof{\ff}^2)}$.
\end{theorem}
\begin{proof}
	By induction on $\ff$. 
	For details see \Cref{app:complexity}.\cmcomment{TR / App.}
\end{proof}
\begin{remark}[Complexity of $\SLA$ Entailments in $\QSLA$]
By \Cref{thm:evalset_size} and \Cref{thm:atleast_size}, the number of entailments and the size of formulae $\atleast{\ra}{\ff}$ is only exponential if $\probsizeof{\ff}$ is not constant. However, we would assume that an entailment $\ff \entails \fg$ in $\QSLA$, where neither in $\ff$ nor in $\fg$ the probabilistic choice $\pp \cdot \fg + (1-\pp) \cdot \fh$ appears, should have a similar runtime complexity as $\SLA$ entailment. While it is easy to see that $\evaluationSet{\ff}=\{0, 1\}$ has constant size in this setting, the size of the formula is still exponential. In the case where no probabilistic choice is present, we generate multiple exponentially-sized tautologies of the form $\atleast{0}{\ff}$. However, due to \Cref{remark:true_in_sla} we can eliminate all occurrences of $\atleast{0}{\ff}$. That means, if $\ff$ does not contain $\pp \cdot \fg + (1-\pp) \cdot \fh$, then for $\ra \neq 0$, we can construct an equivalent formula to $\atleast{\ra}{\ff}$ in such a way that its size is in $\bigo(\sizeof{\ff})$ and $\sizeof{\evaluationSet{\ff}}=2$. \hfill $\triangle$
\end{remark}

\section{Application: Decidable $\hpgcl$ Verification}
\label{sec:applications}
\begin{table}[t]
	\centering
	\caption{$\SLA$ requirements for entailment checking in $\QSLA$.}%
	\label{tab:sla_requirements}%
	\renewcommand{\arraystretch}{1.25}%
	\begin{tabular}{l@{\qquad}l@{\qquad}l@{\qquad}l}
		\hline\hline
		$\QSLfrag$ fragment contains		&          $\SLfrag$ contains/is closed under		%& for $\Predset, \PurePredset$ we have
		\\
		\hline
        $\iverson{\slb}$                                      &  $\slb$, $\true$  %& $\preda \in \Predset$
        \\
        $\iverson{\bb} \cdot \ff + \iverson{\neg\bb} \cdot \fg$ &  $\bb$,$\neg\bb$, $\land$, $\lor$                                      % & $\BB,\neg \BB \in \PurePredset$
        \\
        $\pp \cdot \ff + (1-\pp) \cdot \fg$                     & $\land, \lor$                                       %& 
        \\
        $\ff \cdot \fg$                                         & $\land, \lor$                                       %&
        \\
        $1 - \ff$                                               & $\neg$, $\true$                          % &
        \\

        $\emax{\ff}{\fg}$                                       &  $\vee$                                %   &
        \\

        $\emin{\ff}{\fg}$                                       &  $\wedge$                                % &
        \\
        $\Sup \xx \colon \ff$                                   & $\exists$                            %  &
        \\
        $\Inf \xx \colon \ff$                                   & $\forall$                              %&
        \\
        $\ff \sepcon \fg$                                       &       $\sepcon,\lor$                             % &
        \\
        $\iverson{\slb} \sepimp \ff$                          &  $\slb \sepimp \cdot$                      %  & $\preda \in \Predset, \false \in \PurePredset$
        \\
        %
        %
        %$\iverson{\preda} \sepcon (\iverson{\predb} \sepimp \ff)$&   $\sepcon, \sepimp, \land, \lor$                     & $\preda, \predb \in \Predset$\\
        %
        %
        %$\Inf v \colon \iverson{\slsingleton{v}{\xy}} \sepimp \ff$&       $\forall, \sepimp, \land, \lor$                     & $\slsingleton{\xx}{\xy} \in \Predset$\\
        %
		%
		\hline\hline
	\end{tabular}%
	\renewcommand{\arraystretch}{1}%
	%
	%\vspace*{.5em}%
    %\ifcommentinline{Remove extra transformations if they wont be in this paper.}
\end{table}
Since entailment in full separation logic is undecidable, it is common to consider \emph{fragments} of separation logic with a (semi-)decidable entailment problem. 
Given a $\QSLA$ fragment $\QSLfrag$, we provide sufficient and easy-to-check characterizations on $\SLA$ fragments $\SLfrag$ ensuring that entailment checking in $\QSLfrag$ reduces to entailment checking in $\SLfrag$. This simplifies the search for decidable fragments of \emph{quantitative} separation logic. 

We then apply our results in \Cref{sec:qsh} to show the decidability of entailment checking for \emph{quantitative symbolic heaps}---a quantitative extension of the well-known symbolic heap fragment of separation logic---and demonstrate the applicability to the verification of probabilistic pointer programs. \\ \\
Our reduction from entailments in $\QSLA$ to entailments in $\SLA$ relies on the construction of the $\atleast{\ra}{\ff}$ formulae from \Cref{sec:atleast}. This \mbox{suggests to define:}
\begin{definition}
	Let $\QSLfrag$ be a $\QSLA$ fragment. We say that an $\SLA$ fragment $\SLfrag$ is $\QSLfrag$-admissible if $\atleast{\ra}{\ff} \in \SLfrag$ holds for all $\ff \in \QSLfrag$ and all $\ra \in \Probs$.  \hfill $\triangle$
	%If moreover $\atleast{\ra}{\ff}\in \SLfrag$ for all $\ra \in \Probs$ and all $\ff \in \QSLfrag$, then we say that $\SLA$ is $(\QSLfrag, \atleastsymbol)$-admissible.
\end{definition}
The syntactic nature of our construction of the $\SLfrag$ formulae $\atleast{\ra}{\ff}$ 
allows for a syntactic criterion on $\SLA$ fragments to be $\QSLfrag$-admissible.
\begin{lemma}
	\label{lem:sla_requirements}
	Let $\QSLfrag$ be a $\QSLA$ fragment. If an $\SLA$ fragment $\SLfrag$ satisfies the requirements provided in \Cref{tab:sla_requirements}, then $\SLfrag$ is $\QSLfrag$-admissible.
\end{lemma}
\begin{proof}
	By induction on $\ff$. For details see \Cref{app:applications}.
\end{proof}
Finally, we provide a sufficient criterion for the decidability of entailment in $\QSLA$ fragments given $\SLA$ fragments with a decidable entailment problem. 
Since entailment checks $\sla \entails \slb$ in $\SLA$ can often (but not always) be reduced to unsatisfiability checks $\sla \wedge \neg\slb$, we take a more fine-grained perspective and distinguish between fragments for the left- and the right-hand side of entailments, respectively. This distinction matters when, e.g., $\SLA$ fragments with a decidable satisfiability problem impose restrictions on quantifiers (\cf \cite{Echenim2020Bernays}).
\begin{theorem}
	\label{thm:decide_entailment_qsl_by_sl}
	Let $\QSLfrag_1, \QSLfrag_2$ be $\QSLA$ fragments, and let $\SLfrag_1,\SLfrag_2$ be $\SLA$ fragments. If $\SLfrag_1$ is $\QSLfrag_1$-admissible and $\SLfrag_2$ is $\QSLfrag_2$-admissible, then
	\begin{align*}
	   &\sla \entails \slb ~\text{for}~\sla \in \SLfrag_1, \slb\in \SLfrag_2 ~\text{is decidable}~ \\
	   \text{implies} \qquad & 
	   \fg \entails \ff ~\text{for}~ \fg \in \QSLfrag_1, \ff\in \QSLfrag_2 ~\text{is decidable}~.
	\end{align*}
\end{theorem}
\begin{proof}
	This is a consequence of \Cref{thm:qsla_by_sla}. \qed
\end{proof}

\subsection{Quantitative Symbolic Heaps}
\label{sec:qsh}
\noindent
We now demonstrate that our approach can facilitate the automated verification of probabilistic
pointer programs by providing a sample $\QSL$ fragment with a decidable entailment problem.

Recall that $\QSLA$ is parameterized by a set $\Predset$ of predicate symbols. We obtain the quantitative symbolic heap fragment of $\QSL$ by instantiating $\Predset$.
\begin{definition}
   Let $\Predset$ be the set of predicate symbols given by 
   %the grammar
   \begin{align*}
   \Predset ~=~ & \{\: \true, \slemp \:\} 
           ~\cup~ \{\: \slsingleton{\xx}{\tuplenotation{\xy_1,\ldots,\xy_\recnum}} ~\mid~ \xx,\xy_1,\ldots,\xy_\recnum \in \Vars \:\} 
   \\
                & ~\cup~ \{\: \xx = \xy,~ \xx \neq \xy,~ \xx=\xy \wedge \slemp,~ \xx\neq\xy \wedge \slemp ~\mid~ \xx,\xy \in \Vars 
                  \:\}~.
%   \Predset \quad \rightarrow \quad  &\xx=\xy  ~\mid~ \xx\neq\xy ~\mid~ \true  \\
%   %
%   &\mid~\slemp ~\mid~ \slsingleton{\xx}{\tuplenotation{\xy_1,\ldots,\xy_\recnum}} ~\mid \xx=\xy \wedge \slemp ~\mid \xx\neq\xy \wedge \slemp~.
   \end{align*}
	Then the set $\QSH$ of \emph{quantitative symbolic heaps} is given by the grammar
	\begin{align*}
		\ff \quad \rightarrow \quad \iverson{\slb} ~\mid~ \iverson{\bb} \cdot \ff + \iverson{\neg\bb}\cdot \ff ~\mid~ \qq \cdot \ff + (1-\qq)\cdot \ff ~\mid~ \Sup \xx \colon \ff 
		~\mid~ \ff\sepcon\ff~. \tag*{$\triangle$}
   \end{align*}
\end{definition}
Quantitative symbolic heaps naturally extend the symbolic heap fragment of separation logic.
Intuitively speaking, a quantitative symbolic heap $\ff$ specifies \emph{probability (sub-)distributions} over (symbolic) heaps. By applying \Cref{thm:qsla_by_sla}, we obtain the following decidability result.
\begin{theorem}
	\label{thm:decide_wp_qsh}
	For loop- and allocation-free $\hpgcl$ programs $\cc$
    (that only perform pointer operations, no arithmetic, and guards from the pure fragment of $\Predset$) and $\ff_1,\ff_2 \in \QSH$,
    it is decidable whether the entailment
	$\wlp{\cc}{\ff_1} \entails \ff_2$ holds.
\end{theorem}
Hence, for loop- and allocation-free programs $\cc$ as above, \emph{upper bounds} (in terms of quantitative symbolic heaps $\ff_2$) on the probability $\wlp{\cc}{\ff_1} $ of terminating in a given quantitative symbolic heap $\ff_1$ are decidable. We refer to \Cref{sec:motivation:swap} for an example entailment involving quantitative symbolic heaps. In the sequel, we show how to prove the above result.
%
%\begin{example}[Faulty Swap]
%	\newcommand{\swap}[2]{\textsf{swap}(#1,#2)}
%	Let $\recnum = 1$ and consider the program
%	%
%	\begin{align*}
%		\cswap\colon \qquad
%		& \ASSIGNH{\textsf{tmp1}}{x}\SEMI \\
%		& \ASSIGNH{\textsf{tmp2}}{y}\SEMI \\
%		& \PCHOICE{\HASSIGN{x}{\textsf{tmp2}}}{0.999}{\HASSIGN{x}{\textsf{err}}}\SEMI \\
%		& \HASSIGN{y}{\textsf{tmp1}}~.
%	\end{align*}
%    %
%    %
%    $\cswap$ swaps the contents of cells at addresses $x$ and $y$ respectively on unreliable hardware. That is, with probability $0.001$,  $\cswap$ writes $\textsf{err}$ into the cell pointed to by $x$ instead of $\textsf{tmp2}$. The entailment
%    %
%    \begin{align*}
%       &\wlp{\cswap }{\singleton{\xx}{\xz_2} \sepcon\singleton{\xy}{\xz_1} }    \\
%       %
%       \models{}~& \iverson{\textsf{err}=\xz_2}\cdot(\singleton{\xx}{\xz_1} \sepcon\singleton{\xy}{\xz_2} )
%                         +\iverson{\textsf{err}\neq\xz_2}\cdot(0.999 \cdot (\singleton{\xx}{\xz_1} \sepcon\singleton{\xy}{\xz_2} ))~.                      
%    \end{align*}
%    %
%    expresses that if $\textsf{err}$ is \emph{different} from the value $y$ points to, then $\cswap$ performs a correct swap with probability at most $0.999$. \hfill $\triangle$
%%
%%
%\end{example}
%%
%
%
%
\subsubsection{Proof of \Cref{thm:decide_wp_qsh}.}
The proof relies on \emph{extended} quantitative symbolic heaps $\eQSH$, which include magic wands with points-to formulae on their left-hand side.
\begin{definition}\label{def:eQSH}
	The set $\eQSH$ of \emph{extended quantitative symbolic heaps} is given by the grammar 
	\begin{align*}
		\fg \quad \rightarrow \quad 
		&\iverson{\slb}
		~\mid~ \iverson{\bb} \cdot \fg + \iverson{\neg\bb}\cdot \fg 
		 ~\mid~ \qq \cdot \fg + (1-\qq)\cdot \fg 		
		~\mid~ \fg\sepcon\fg  \\
		& ~\mid~ \Sup \xx \colon \fg 
		   ~\mid~ \singleton{\xx}{\tuplenotation{\xy_1,\ldots,\xy_\recnum}} \sepimp \fg~.
		   \tag*{$\triangle$}
	\end{align*}
\end{definition}
Notice that indeed $\QSH \subseteq \eQSH$.
\begin{lemma} \label{lem:qsh_wp_closed}
		For every loop- and allocation-free program $\cc \in \hpgcl$ without arithmetic and only with guards of the pure fragment of $\Predset$, extended quantitative symbolic heaps are closed under $\wlpC{\cc}$, i.e.,
		\[
		   \text{for all $\fg \in \eQSH$}\colon \quad \wlp{\cc}{\fg} \in \eQSH~.
		\]
		In particular, since $\QSH \subseteq \eQSH$, we have 
		\[
		    \text{for all $\ff \in \QSH$}\colon \quad \wlp{\cc}{\ff} \in \eQSH~. 
		\]
\end{lemma}
\begin{proof}
By induction on the structure of loop- and allocation-free program $\cc$. 
See \Cref{app:qsh} for details.\cmcomment{TR / App.}
\end{proof}
Hence, if $\fg \models \ff$ is decidable for $\fg \in \eQSH$ and $\ff \in \QSH$, \Cref{thm:decide_wp_qsh} follows.
\begin{lemma}
	\label{lem:decide_entailment_eqsh_qsh}
	For $\fg \in \eQSH$ and $\ff \in \QSH$, it is decidable whether 
	      $\fg \entails \ff$ holds.
\end{lemma}
\begin{proof}
	We employ \Cref{lem:sla_requirements} to determine two $\SLA$ fragments $\SLfrag_1, \SLfrag_2$ such that $\SLfrag_1$ is $\eQSH$-admissible and $\SLfrag_2$ is $\QSH$-admissible. 
	Then, by \Cref{thm:decide_entailment_qsl_by_sl}, decidability of $\fg \entails \ff$ follows from decidability of $\sla\entails\slb$ for $\sla \in \SLfrag_1$ and $\slb \in \SLfrag_2$. For that, we exploit the equivalence
	\[
	    \sla\entails\slb \qquad \text{iff} \qquad \sla \wedge \neg\slb ~\text{is unsatisfiable}~.
	\]
	The latter is decidable by \cite[Theorem 3.3]{Echenim2020Bernays} since  $\sla \wedge \neg\slb$ is equivalent to a formula of the form $\exists^*\forall^* \colon \slc$ with $\slc$ quantifier-free and no formula $\slc_1 \sepimp \slc_2$ occurring in $\slc$ contains a universally quantified variable. 
	See \Cref{app:qsh} for details.\cmcomment{TR / App.}
\end{proof}

\section{Related Work}
\label{sec:related}
\paragraph{Weakest preexpectations.}
Weakest precondition reasoning was established in a classical setting by Dijkstra \cite{dijkstra1976discipline} and has been extended to provide semantic foundations for probabilistic programs by Kozen \cite{Kozen1983Probabilistic,Kozen1997Semantics} and McIver \& Morgan \cite{McIver2005Abstraction}, who also coined the term \emph{weakest preexpectations}.
Their relation to operational models is studied in \cite{gretz2014semantics}.
Moreover, weakest preexpectation reasoning has been shown to be useful for obtaining bounds on the expected resource consumption \cite{ngo2018resource} and, in particular, the expected run-time \cite{kaminski2018weakest} of probabilistic programs.

\paragraph{Logics for probabilistic pointer programs.}
Although many algorithms rely on randomized dynamic data structures, formal reasoning about programs that are both probabilistic and heap manipulating has received scarce attention.
A notable exception is the work by Tassarotti and Harper \cite{Tassarotti2019Separation}, who introduce a concurrent separation logic with support for probabilistic reasoning, called Polaris.
Their focus is on program refinement, employing a semantic model that is based on the idea of coupling, which underlies recent work on probabilistic relational Hoare logics \cite{barthe2017coupling}.
However, no other decision procedures targeting entailments for $\QSL$ or other logics targeting probabilistic pointer programs exist.

\paragraph{Leveraging SL research.}
As shown in \Cref{tab:sla_requirements}, building $\QSL$ entailment checkers by employing our reduction technique requires the availability of $\SL$ fragments that support certain logical operations, and whose entailment problem is decidable.
Since the inception of separation logic~\cite{Ishtiaq2001BI}, the latter has been extensively studied.
In particular, the symbolic heap fragment of $\SL$ has received a lot of attention.
\Cref{tab:sl_entailment} gives an overview of related approaches. \footnote{$\sepcon$ is always covered. Supported (Boolean or separating) connectives are marked with \enquote{+}, unsupported ones with \enquote {--}. \enquote{$\ast$} means that the restrictions on the connective are more involved.  \enquote{Pure} means that the connective can only appear in pure formulae and \enquote{flat} means that the quantifier needs to be on the outermost level.}

\begin{table}%[t]
	\centering
	\caption{$\SL$ fragments with decidable entailment problem.}%
	\label{tab:sl_entailment}%
	\begin{tabular}{l@{\quad}cccccc@{\qquad}l@{\qquad}l}
	  \hline\hline
	  Paper &
	  $\neg$ &
	  $\wedge$ &
	  $\vee$ &
	  $\sepimp$ &
	  $\exists$ &
	  $\forall$ &
	  Ind.~predicates &
	  Complexity\\
	  \hline
	  % https://link.springer.com/chapter/10.1007/978-3-642-54830-7_27
	  \cite{Antonopoulos2014Foundations} &
	  pure &
	  pure &
	  \changedinline{pure} &
	  -- &
	  flat &
	  -- &
	  user defined &
	  \textsc{ExpTime}-hard\\
	  % https://link.springer.com/chapter/10.1007/978-3-540-30538-5_9
	  % https://link.springer.com/chapter/10.1007/978-3-642-23217-6_16
	  \cite{Berdine2004Decidable}
	  \cite{Cook2011Tractable} &
	  -- &
	  pure &
	  -- &
	  -- &
	  -- &
	  -- &
	  Lists &
	  Polynomial\\
	  % https://hal.archives-ouvertes.fr/hal-03040180/document
	  \cite{Echenim2021Decidable} &
	  -- &
	  -- &
	  -- &
	  -- &
	  + &
	  -- &
	  user defined &
	  2-\textsc{ExpTime}-complete\\
	  % https://arxiv.org/abs/2012.14361
	  \cite{Echenim2021Unifying} &
	  -- &
	  -- &
	  + &
	  -- &
	  + &
	  -- &
	  user defined &
	  2-\textsc{ExpTime}-complete\\
	  % https://link.springer.com/chapter/10.1007/978-3-642-38574-2_2
	  \cite{Iosif2013Tree} &
	  -- &
	  + &
	  -- &
	  -- &
	  flat &
	  -- &
	  user defined &
	  ?\\
	  % https://link.springer.com/chapter/10.1007/978-3-319-11936-6_15
	  \cite{Iosif2014Deciding} &
	  -- &
	  pure &
	  -- &
	  -- &
	  flat &
	  -- &
	  user defined &
	  \textsc{ExpTime}-complete\\
	  % https://publik.tuwien.ac.at/files/publik_287079.pdf
	  \cite{Katelaan2019Effective} &
	  -- &
	  -- &
	  -- &
	  -- &
	  + &
	  -- &
	  user defined &
	  2-\textsc{ExpTime}\\
	  % https://ui.adsabs.harvard.edu/abs/2020arXiv200201202P/abstract
	  \cite{Matheja2020Complete} &
	  $\ast$ &
	  + &
	  + &
	  $\ast$ &
	  -- &
	  -- &
	  user defined &
	  2-\textsc{ExpTime}\\
	  % https://link.springer.com/chapter/10.1007/978-3-319-46520-3_16
	  \cite{Reynolds2016Decision} &
	  + &
	  + &
	  + &
	  + &
	  -- &
	  -- &
	  -- &
	  ?\\
	  \cite{Demri2018} &
	  + &
	  + &
	  + &
	  -- &
	  -- &
	  -- &
	  Lists &
	  \textsc{PSpace}-complete\\
	  \cite{Echenim2020Bernays} &
	  + &
	  + &
	  + &
	  $\ast$ &
	  $\ast$ &
	  $\ast$ &
	  -- &
	  \textsc{PSpace}-complete\\
%   % URL
%	  \cite{CITATION} &
%	  FRAGMENT &
%	  NEG &
%	  AND &
%	  OR &
%	  WAND &
%	  FORALL &
%	  PREDICATES &
%	  COMPLEXITY\\
	  \hline\hline
	\end{tabular}%
\end{table}

%Undecidable and basics: 
%\cite{Antonopoulos2014Foundations}
%\cite{Brotherston2014Undecidability}
%\cite{Calcagno2005Separation}
%\cite{Calcagno2001Computability}

%\paragraph{Other Verification Techniques.}
%
%SL seminal: \cite{Ishtiaq2001BI,Reynolds2002Separation}
%
%Overview: \cite{OHearn2019Separation}
%
%SL symbolic execution (inspiring updated rules): 
%\cite{Berdine2005Symbolic}
%
%Bernays\cite{Echenim2019Bernays,Echenim2020Bernays}
%
%Expressiveness: \cite{Echenim2018Expressive}
%
%McIver and Morgan\cite{McIver2005Abstraction}
%
%Relative completeness\cite{Batz2021Relatively}
%
%Latticed k-induction\cite{Batz2021Latticed}
%
%QSL \cite{Batz2019Quantitative}
%\cite{Matheja2020Automated}
%Haslbeck PhD thesis
%
%SL-COMP 2019 \cite{Sighireanu2019SLCOMP}
%
%SL solvers
%\cite{Enea2017Compositional}
%\cite{Le2018Frame}
%\cite{Serban2018Entailment}
%\cite{Brotherston2011Automated}
%\cite{Piskac2013Automating}
%
%Incorrectness logic \cite{OHearn2020Incorrectness}
%
%
%SL fragments with variable restrictions
%\cite{Demri2015Two,Demri2017Separation}
%
%Decision procedure for satisfiability:
%\cite{Brotherston2014Decision}
%
%Automated reasoning about prob. programs:
%\cite{Moosbrugger2021Automated}
%
%Joe's work on SL for randomized programs \cite{Tassarotti2019Separation}
%
%The Kozen \cite{Kozen1983Probabilistic}
%
%Invariants for pgcl
%\cite{Gretz2013Prinsys}
%
%PPLs are harder than ordinary programs
%\cite{Kaminski2019Hardness}
%
%Proof rules using entailments
%\cite{McIver2018New,Hark2020Aiming,Kaminski2019Advanced}
%\cite{Takisaka2021Ranking}
%
%PGCL Isabelle
%\cite{Hurd2005Probabilistic}

\section{\cm{Discussion and} Conclusion}
\label{sec:conclusion}
We studied entailment checking in $\QSL$ by means of a reduction to entailment checking in $\SL$. We analyzed the complexity of our approach and demonstrated its applicability by means of several examples. In particular, our reduction yields the first decidability result for probabilistic pointer program verification. 

\cm{
Our primary goal was to investigate the entailment problem for QSL to pave the way for automated verification of probabilistic pointer programs. 
\Cref{thm:decide_entailment_qsl_by_sl} provides a generic result that enables building upon the large body of work dealing with classical SL entailments to obtain both theoretical and practical insights.
Theoretically, \Cref{thm:decide_entailment_qsl_by_sl} gives sufficient criteria to derive QSL fragments with a decidable entailment problem from a classical SL fragment.
We derived a QSL fragment such that reasoning about a simple probabilistic heap-manipulating language becomes decidable.
More practically, \Cref{thm:decide_entailment_qsl_by_sl} allows reusing existing (possibly incomplete) SL solvers to solve the entailments derived by our construction---an empirical evaluation of how well existing solvers can deal with these entailments is an interesting direction for future work.
}

\cm{
We believe that our fine-grained complexity analysis demonstrates that our approach can be practically feasible: the exponential blow-up in \Cref{thm:atleast_size} stems from the number of probabilistic constructs in the given QSL formulae. We expect the number of such constructs to be small for many randomized algorithms.
%, which typically contain only few probabilistic statements.
%While we cannot justify our approach's complexity via a lower bound, w
We remark that existing approaches on checking quantitative entailments between heap-independent expectations encounter similar exponential blow-ups (cf., \cite{DBLP:conf/sas/KatoenMMM10,Batz2021Latticed}). 
    There is thus some
%These approaches have been implemented and give 
evidence that such exponential blow-ups do not prohibit one from automatically verifying non-trivial properties. We are not aware of work on checking quantitative entailments between expectations that avoids such exponential blow-ups.
}

Future work includes considering richer classes of $\QSL$ and applications of entailment checking such as $k$-induction~\cite{Batz2021Latticed}. Another interesting direction is the applicability of our reduction to other approaches that aim for local reasoning about the resources employed by probabilistic programs, such as~\cite{Tassarotti2019Separation,Bao21BunchedIndependence,Barthe20ProbSL}.

 %targeting the verification of probabilistic pointer programs by separation logic such as 

 \bibliographystyle{splncs04}
\bibliography{bibfile}

%%%%% To display Open Access text and logo, Please add below text and copy the cc_by_4-0.eps in the manuscript package %%%

\vfill

{\small\medskip\noindent{\bf Open Access} This chapter is licensed under the terms of the Creative Commons\break Attribution 4.0 International License (\url{http://creativecommons.org/licenses/by/4.0/}), which permits use, sharing, adaptation, distribution and reproduction in any medium or format, as long as you give appropriate credit to the original author(s) and the source, provide a link to the Creative Commons license and indicate if changes were made.}

{\small \spaceskip .28em plus .1em minus .1em The images or other
third party material in this chapter are included in the\break
chapter's Creative Commons license, unless indicated otherwise in a
credit line to the\break material.~If material is not included in
the chapter's Creative Commons license and\break your intended use
is not permitted by statutory regulation or exceeds the
permitted\break use, you will need to obtain permission directly
from the copyright holder.}

\medskip\noindent\includegraphics{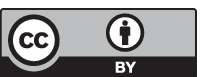}

%% Appendix
\appendix

\newpage
\section*{Appendix}
\section{Proof of \Cref{thm:qsl_well_defined}}\label{app:qsl_well_defined}
\begin{restateTheorem}{thm:qsl_well_defined}
    The semantics of $\QSLA$ formulae is well-defined, i.e., for all $\ff\in \QSLA$, we have
	%
	%\[
	$\sem{\ff} \in\Eone.$
	%\]
\end{restateTheorem}
\begin{proof}
By induction on $\ff$. For the base case $\ff=\iverson{\slb}$, we have $\semapp{\iverson{\slb}}{(\sk, \hh)} \in \{0,1\}$. For all other cases, it is straightforward to prove that $\semapp{\ff}{(\sk, \hh)}\in [0,1]$, since $[0,1]$ is closed under all operations used in \Cref{tab:semantics_qsl} if they are defined. It is only left to proof that $\sem{\Sup \xx \colon \fg}$, $\sem{\Inf \xx \colon \fg}$ and $\sem{\fg \sepcon \fh}$ are well-defined. For that the $\max$ and $\min$ need to be defined on the given sets, i.e. that the sets
\begin{gather*}
    \big\{\semapp{\fg}{(\sk\statesubst{\xx}{v},\hh)} ~\mid~ v \in \Vals \big\} \\
    \left\{\semapp{\fg}{(\sk,\hh_1)} \cdot \semapp{\fh}{(\sk,\hh_2)} ~\mid~ \hh = \hh_1 \sepcon \hh_2 \right\}
\end{gather*}
are non-empty and finite. Since $\Vals$ is non-empty and $\hh = \hh \sepcon \emptyheap$ where $\emptyheap$ is the heap with $\dom{\emptyheap}=\emptyset$, both sets are non-empty. Lastly, from \Cref{thm:evalset} it follows directly, that both sets are also finite.
\end{proof}
\section{Appendix to \Cref{sec:motivation}}
\subsection{Proof of \Cref{thm:wp_closed}}\label{app:wp_closed}
\begin{restateTheorem}{thm:wp_closed}
    Let $\cc \in \hpgcl$ be loop-free and $\Predset$ be a set of \mbox{predicate symbols. If}
	\begin{enumerate}
		\item  $\Predset$ contains the points-to predicate for all variables and all expressions occurring in allocation, disposal, lookup and mutation in $\cc$,
		\item $\Predset$ contains all guards and their negations occurring in $\cc$, and
		\item all predicates in $\Predset$ are closed under substitution of variables by other variables and arithmetic expressions occurring on right-hand sides of assignments in $\cc$, 
	\end{enumerate}
	then, for every $\QSLA$ formula $\ff$, $\wlp{\cc}{\ff} \in \QSLA$.
\end{restateTheorem}
\begin{proof}
First we remark that since $\Predset$ is closed under substitution of variables by arithmetic expressions occurring on right-hand sides of assignments in $\cc$, so is $\QSLA$. For a formula $\ff \in \QSLA$, substitutions are always only handed to the atoms in $\Predset$. Since $\Predset$ is closed under substitution of variables by arithmetic expressions occurring on right-hand sides of assignments, we have for every such expression $\ee$ on a right-hand side of an assignment and $\slb \in \Predset$ that also $\slb \subst{\xx}{\ee} \in \Predset$, where the semantics $\sem{\slb \subst{\xx}{\ee}}$ of $\slb \subst{\xx}{\ee}$ is defined as
\[ \sem{\slb \subst{\xx}{\ee}} = \setcomp{(\sk \statesubst{\xx}{\ee(\sk)}, \hh) \in \slb}{ (\sk, \hh) \in \States}~. \]
Now we prove the theorem by induction on $\cc$.

\medskip \noindent
\emph{For the base case $\cc=\SKIP$,} we have $\wlp{\SKIP}{\ff}=\ff\in \QSLA$ by assumption.

\medskip \noindent
\emph{For the base case $\cc=\ASSIGN{\xx}{\ee}$,} we have $\wlp{\ASSIGN{\xx}{\ee}}{\ff}=\ff \subst{\xx}{\ee} \in \QSLA$ since $\QSLA$ is closed under the substitution of $\xx$ by $\ee$.

\medskip \noindent
\emph{For the base case $\cc=\ALLOC{\xx}{\ee}$,} we have 
\[\wlp{\ALLOC{\xx}{\ee}}{\ff}=\Inf \xy \colon \singleton{\xy}{\ee} \sepimp \ff \subst{\xx}{\xy} \in \QSLA\] 
because $\QSLA$ is closed under substitution of $\xx$ by $\xy$ and $\Predset$ contains the points-to formula $\singleton{\xy}{\ee}$ by assumption.

\medskip \noindent
\emph{For the base case $\cc=\FREE{\ee}$,} we have 
\[\wlp{\FREE{\ee}}{\ff}=\validpointer{\ee} \sepcon \ff = \Sup \xx \colon \singleton{\ee}{\xx} \sepcon \ff \in \QSLA\] 
because $\Predset$ contains $\singleton{\ee}{\xx}$ by assumption.

\medskip \noindent
\emph{For the base case $\cc=\ASSIGNH{\xx}{\ee}$,} we have 
\[\wlp{\ASSIGNH{\xx}{\ee}}{\ff} = \Sup \xx \colon \singleton{\ee}{\xy} \sepcon (\singleton{\ee}{xy} \sepimp \ff \subst{\xx}{\xy}) \in \QSLA\] 
because $\QSLA$ is closed under substitution of $\xx$ by $\xy$ and $\Predset$ contains the points-to formulae $\singleton{\ee}{\xy}$ by assumption.

\medskip \noindent
\emph{For the base case $\cc=\HASSIGN{\ee}{\ee'}$,} we have 
\begin{align*}
    &\wlp{\HASSIGN{\ee}{\ee'}}{\ff} \\
    = \quad & \validpointer{\ee} \sepcon (\singleton{\ee}{\ee'} \sepimp \ff) \\
    = \quad & (\Sup \xx \colon \singleton{\ee}{\xx}) \sepcon (\singleton{\ee}{\ee'} \sepimp \ff) \in \QSLA
\end{align*} 
because $\Predset$ contains the points-to formulae $\singleton{\ee}{\xx}$ and $\singleton{\ee}{\ee'}$ by assumption.

\medskip \noindent
For all other composite cases we assume for some fixed, but arbitrary loop-free programs $\cc_1, \cc_2 \in \hpgcl$ that for all $\fg \in \QSLA$ we have $\wlp{\cc_1}{\fg} \in \QSLA$ and for all $\fh \in \QSLA$ we have $\wlp{\cc_2}{\fh} \in \QSLA$.

\medskip \noindent
\emph{For the case $\cc=\PCHOICE{\cc_1}{\pp}{\cc_2}$,} we have 
\[\wlp{\cc}{\ff}=\pp \cdot \wlp{\cc_1}{\ff} + (1-\pp) \cdot \wlp{\cc_2}{\ff} \in \QSLA\] 
by the induction hypothesis.

\medskip \noindent
\emph{For the case $\cc=\COMPOSE{\cc_1}{\cc_2}$,} we have $\wlp{\cc_2}{\ff} \in \QSLA$ 
by the induction hypothesis, thus we also have $\wlp{\cc_1}{\wlp{\cc_2}{\ff}} \in \QSLA$ by the induction hypothesis.

\medskip \noindent
\emph{For the case $\cc=\ITE{\guard}{\cc_1}{\cc_2}$,} we have 
\[\wlp{\cc}{\ff}= \iverson{\guard} \cdot \wlp{\cc_1}{\ff} + \iverson{\neg \guard} \cdot \wlp{\cc_2}{\ff} \in \QSLA\] 
by the induction hypothesis and since $\Predset$ contains $\guard$ and $\neg \guard$ by assumption.

\medskip \noindent
This concludes the proof.

\end{proof}

\subsection{$\cswap$ Example} \label{app:swap_example}
\begin{changed}
\subsubsection{Full Computation}
\begingroup
\allowdisplaybreaks
\begin{align*}
    & \wlp{\cswap}{\singleton{x}{\xz_1} \sepcon \singleton{y}{\xz_2}} \\
\mirrorentails \qquad & \wlp{\cswap}{\singleton{y}{\xz_2} \sepcon (\singleton{x}{\xz_1} \cdot (\iverson{\textsf{tmp1}=\xz_2} \cdot \iverson{\textsf{tmp2} = \xz_1}))} \tag{monotonicity} \\
\mirrorentails \qquad & \wlp{\cswap}{\singleton{y}{\textsf{tmp1}} \sepcon (\singleton{x}{\textsf{tmp2}} \cdot (\iverson{\textsf{tmp1}=\xz_2} \cdot \iverson{\textsf{tmp2} = \xz_1}))} \tag{variable substitution}\\
\mirrorentails \qquad & \wlpC{\cc_1 \SEMI \cc_2 \SEMI \cc_3}(\wlpC{\cc_4}(\singleton{y}{\textsf{tmp1}} \sepcon (\singleton{x}{\textsf{tmp2}}  \\
    & \qquad \qquad \qquad \qquad \cdot (\iverson{\textsf{tmp1} = \xz_2} \cdot \iverson{\textsf{tmp2}=\xz_1})))) \tag{$\wlpsymbol$ application}\\
\mirrorentails \qquad & \wlp{\cc_1 \SEMI \cc_2 \SEMI \cc_3}{\validpointer{y} \sepcon (\singleton{x}{\textsf{tmp2}} \cdot (\iverson{\textsf{tmp1} = \xz_2} \cdot \iverson{\textsf{tmp2} = \xz_1})) } \tag{Rule (ii)} \\
\mirrorentails \qquad & \wlpC{\cc_1 \SEMI \cc_2 }(\wlpC{\cc_3}(\validpointer{y} \sepcon (\singleton{x}{\textsf{tmp2}}\\
    & \qquad \qquad \qquad \qquad  \cdot (\iverson{\textsf{tmp1} = \xz_2} \cdot \iverson{\textsf{tmp2} = \xz_1})))) \tag{$\wlpsymbol$ application}\\
\mirrorentails \qquad & \wlpC{\cc_1 \SEMI \cc_2 }(0.999 \cdot \wlpC{\HASSIGN{x}{\textsf{tmp2}}}(\validpointer{y} \\
    & \qquad \qquad \qquad \qquad  \sepcon (\singleton{x}{\textsf{tmp2}} \cdot (\iverson{\textsf{tmp1} = \xz_2} \cdot \iverson{\textsf{tmp2} = \xz_1})))\\
    & \qquad \qquad ~~~ + 0.001 \cdot \wlpC{\HASSIGN{x}{\textsf{err}}}(\validpointer{y}\\
    & \qquad \qquad \qquad \qquad \sepcon (\singleton{x}{\textsf{tmp2}}\cdot (\iverson{\textsf{tmp1} = \xz_2} \cdot \iverson{\textsf{tmp2} = \xz_1})))) \tag{$\wlpsymbol$ application}\\
\mirrorentails \qquad & \wlpC{\cc_1 \SEMI \cc_2 }(0.999 \cdot \wlpC{\HASSIGN{x}{\textsf{tmp2}}}(\validpointer{y} \\
& \qquad \qquad \qquad \qquad  \sepcon (\singleton{x}{\textsf{tmp2}} \cdot (\iverson{\textsf{tmp1} = \xz_2} \cdot \iverson{\textsf{tmp2} = \xz_1})))\\
    & \qquad \qquad ~~~ + 0.001 \cdot \iverson{\false}) \tag{monotonicity}\\
\mirrorentails \qquad & \wlpC{\cc_1 \SEMI \cc_2 }(0.999 \cdot \wlpC{\HASSIGN{x}{\textsf{tmp2}}}(\singleton{x}{\textsf{tmp2}} \\
& \qquad \qquad \qquad \qquad  \sepcon (\validpointer{y} \cdot (\iverson{\textsf{tmp1} = \xz_2} \cdot \iverson{\textsf{tmp2} = \xz_1})))\\
    & \qquad \qquad ~~~ + 0.001 \cdot \iverson{\false}) \tag{commutativity and associativity}\\
\mirrorentails \qquad & \wlpC{\cc_1 \SEMI \cc_2 }(0.999 \cdot (\validpointer{x} \sepcon (\validpointer{y}\\
& \qquad \qquad \qquad \qquad ~~\cdot (\iverson{\textsf{tmp1} = \xz_2} \cdot \iverson{\textsf{tmp2} = \xz_1})))\\
& \qquad \qquad ~~~~ + 0.001 \cdot \iverson{\false}) \tag{Rule (ii)}\\
\mirrorentails \qquad & \wlpC{\cc_1}(\wlpC{\cc_2}(0.999 \cdot (\validpointer{x} \sepcon (\validpointer{y}\\
& \qquad \qquad \qquad \qquad \qquad\cdot (\iverson{\textsf{tmp1} = \xz_2} \cdot \iverson{\textsf{tmp2} = \xz_1})))\\
& \qquad \qquad \qquad ~~ + 0.001 \cdot \iverson{\false}) \tag{$\wlpsymbol$ application}\\
\mirrorentails \qquad & \wlpC{\cc_1}(0.999 \cdot \wlpC{\cc_2}(\validpointer{x} \sepcon (\validpointer{y}\\
& \qquad \qquad \qquad \qquad \qquad\cdot (\iverson{\textsf{tmp1} = \xz_2} \cdot \iverson{\textsf{tmp2} = \xz_1})))\\
& \qquad ~~ + 0.001 \cdot \wlpC{\cc_2}(\iverson{\false})) \tag{super-distributivity}\\
\mirrorentails \qquad & \wlpC{\cc_1}(0.999 \cdot(\wlpC{\cc_2}((\singleton{y}{\xz_1} \cdot \iverson{\textsf{tmp2} = \xz_1})  \\
    & \qquad \qquad \qquad \qquad \qquad \sepcon (\iverson{\textsf{tmp1} = \xz_2} \cdot \validpointer{x}))\\
    & \qquad ~~ + 0.001 \cdot \iverson{\false}) \tag{monotonicity}\\
\mirrorentails \qquad & \wlp{\cc_1}{0.999 \cdot (\singleton{y}{\xz_1} \sepcon (\iverson{\textsf{tmp1} = \xz_2} \cdot \validpointer{x})) + 0.001 \cdot \iverson{\false}} \tag{Rule (i)} \\
\mirrorentails \qquad & \quad 0.999 \cdot \wlp{\cc_1}{\singleton{y}{\xz_1} \sepcon (\iverson{\textsf{tmp1} = \xz_2} \cdot \validpointer{x})} \\
    & + 0.001 \cdot \wlp{\cc_1}{\iverson{\false}} \tag{super-distributivity} \\
\mirrorentails \qquad & \quad 0.999 \cdot \wlp{\cc_1}{\singleton{y}{\xz_1} \sepcon (\iverson{\textsf{tmp1} = \xz_2} \cdot \singleton{x}{\xz_2})} \\
    & + 0.001 \cdot \iverson{\false} \tag{monotonicity} \\
\mirrorentails \qquad & \quad 0.999 \cdot \wlp{\cc_1}{(\singleton{x}{\xz_2} \cdot \iverson{\textsf{tmp1} = \xz_2}) \sepcon \singleton{y}{\xz_1}} \\
    & + 0.001 \cdot \iverson{\false} \tag{commutativity} \\
\mirrorentails \qquad & \quad 0.999 \cdot (\singleton{x}{\xz_2} \sepcon \singleton{y}{\xz_1}) + 0.001 \cdot \iverson{\false} \tag{Rule (i)} \\
\end{align*}
\endgroup

\subsubsection{Necessary Separation Logic Entailments}

\begin{align*}
    & \quad (((\singleton{\xx}{\xz_2} \sepcon \singleton{\xy}{\xz_1}) \land \false) \lor (\singleton{\xx}{\xz_2} \sepcon \singleton{y}{\xz_1})) \lor \false\\
    \entails \quad & \quad ((\singleton{\xx}{\xz_2} \sepcon \singleton{\xy}{\xz_1}) \land \false) \lor (\singleton{\xx}{\xz_2} \sepcon \singleton{\xy}{\xz_1})\\
\\
    & \quad ((\singleton{\xx}{\xz_1} \sepcon \singleton{\xy}{\xz_2}) \land \false) \lor (\singleton{\xx}{\xz_2} \sepcon \singleton{\xy}{\xz_1})\\
    \entails \quad & \quad ((\singleton{\xx}{\xz_1} \sepcon \singleton{\xy}{\xz_1}) \land \false) \lor (\singleton{\xx}{\xz_2} \sepcon \singleton{\xy}{\xz_1})\\
\\
    & \quad (\singleton{\xx}{\xz_2} \sepcon \singleton{\xy}{\xz_1}) \land \false\\ 
    \entails \quad & \quad (\singleton{\xx}{\xz_2} \sepcon \singleton{\xy}{\xz_1}) \land \false
\end{align*}

\subsubsection{SMT-LIB 2 File}

\begin{verbatim}
    (set-logic QF_ALL_SUPPORTED)
    (declare-sort Loc 0)
    (declare-heap (Loc Int))
    (declare-const x Loc)
    (declare-const y Loc)
    (declare-const z1 Int)
    (declare-const z2 Int)
    (assert (or (and (or (or (and (sep (pto x z2) (pto y z1)) false) 
    (sep (pto x z2) (pto y z1))) false) 
    (not (or (and (sep (pto x z2) (pto y z1)) false) 
    (sep (pto x z2) (pto y z1))))) 
    (or (and (or (and (sep (pto x z2) (pto y z1)) false) 
    (sep (pto x z2) (pto y z1))) 
    (not (or (and (sep (pto x z2) (pto y z1)) false) 
    (sep (pto x z2) (pto y z1))))) 
    (and (and (sep (pto x z2) (pto y z1)) false) 
    (not (and (sep (pto x z2) (pto y z1)) false))))))
    (check-sat)
\end{verbatim}
\end{changed}

\section{Appendix to \Cref{sec:entailment_checking}}
\label{app:entailment_checking}
\begin{lemma} \label{lem:zero_one_preservation}
	For all formulae $\ff \in \QSLA$ we have $0, 1 \in \evaluationSet{\ff}$.
\end{lemma}
\begin{proof}
	By induction on $\ff$.
    
    \medskip \noindent
    \emph{For the base case $\ff=\iverson{\slb}$,} we have $0, 1 \in \{0, 1\} = \evaluationSet{\ff}$.

    \medskip \noindent
    For all other composite cases we assume for some fixed, but arbitrary $\fg, \fh \in \QSLA$ that $0, 1 \in \evaluationSet{\fg}$ and $0, 1 \in \evaluationSet{\fh}$.
    
    \medskip \noindent
    \emph{For the induction step $\ff=\iverson{\bb} \cdot \fg +\iverson{\neg\bb}\cdot \fh$,} we have by the induction hypothesis $0, 1 \in \evaluationSet{\fg} \subseteq \evaluationSet{\ff}$.

    \medskip \noindent
    \emph{For the induction step $\ff=\pp \cdot \fg + (1-\pp) \cdot \fh$,} we have by the induction hypothesis $0, 1 \in \evaluationSet{\fg}$ and $0, 1 \in \evaluationSet{\fh}$. Thus we also have $0, \pp \in \pp \cdot \evaluationSet{\fg}$ and $0, (1-\pp) \in (1-\pp) \cdot \evaluationSet{\fh}$, thus with $0=0+0$ and $1=\pp + (1-\pp)$ also $0, 1 \in \evaluationSet{\ff}$.

    \medskip \noindent
    \emph{For the induction step $\ff=\fg \cdot \fh$,} we have by the induction hypothesis $0, 1 \in \evaluationSet{\fg}$ and $0, 1 \in \evaluationSet{\fh}$. Thus we also have with $0 = 0 \cdot 0$ and $1 = 1 \cdot 1$ that $0, 1 \in \evaluationSet{\ff}$.

    \medskip \noindent
    \emph{For the induction step $\ff=1- \fg$,} we have by the induction hypothesis $0, 1 \in \evaluationSet{\fg}$, thus with $0 = 1-1$ and $1= 1-0$ we also have $0, 1 \in \evaluationSet{\ff}$.

    \medskip \noindent
    \emph{For the induction step $\ff=\emax{\fg}{\fh}$,} we have by the induction hypothesis $0, 1 \in \evaluationSet{\fg}$ and $0, 1 \in \evaluationSet{\fh}$ thus with $0 = \max(0, 0)$ and $1 = \max(1, 1)$ we also have $0, 1 \in \evaluationSet{\ff}$.

    \medskip \noindent
    \emph{The induction step $\ff=\emin{\fg}{\fh}$} is analogous to the previous case.

    \medskip \noindent
    \emph{For the induction steps $\ff=\Inf \xx \colon \fg$,} we have by the induction hypothesis $0, 1 \in \evaluationSet{\fg} = \evaluationSet{\ff}$.

    \medskip \noindent
    \emph{The induction step $\ff=\Sup \xx \colon \fg$} is analogous to the previous case.

    \medskip \noindent
    \emph{The induction step $\ff=\fg \sepcon \fh$} is analogous to the case $\fg \cdot \fh$.

    \medskip \noindent
    \emph{For the induction step $\ff=\iverson{\slb} \sepimp \fg$,} we have by the induction hypothesis $0, 1 \in \evaluationSet{\fg} = \evaluationSet{\ff}$.

    \medskip \noindent
    This concludes the proof.
\end{proof}
\bigskip

\begin{restateTheorem}{thm:atleast}
	For every $\ff \in \QSLA$ and all $\ra  \in \Probs$ there is an effectively constructible $\SLA$ formula $\atleast{\ra}{\ff}$ such that for all $(\sk,\hh) \in \States$, we have
	\[
	  (\sk,\hh) \mmodels \atleast{\ra}{\ff}
		\qquad\text{iff}\qquad
	  \ra\lleq \semapp{\ff}{(\sk,\hh)}~.
	\]
\end{restateTheorem}
\begin{proof}
By induction on $\ff$.

\medskip \noindent
\emph{The case $\iverson{\slb}$.} If $\ra = 0$ then $\ra \leq \semapp{\iverson{\slb}}{(\sk, \hh)}$ trivially. If $0 < \ra$, then
\begin{align*}
	\ra \leq \semapp{\iverson{\slb}}{(\sk,\hh)}
	\qquad\text{iff}\qquad
	\semapp{\iverson{\slb}}{(\sk,\hh)}=1
	\qquad\text{iff}\qquad
	(\sk,\hh)\models \slb~.
\end{align*}

\medskip \noindent
For the composite cases, now assume that for some arbitrary, but fixed formulae $\fg,\fh \in \QSLA$ and all probabilities $\rb,\rc \in \Probs$ there are effectively constructible $\SLA$ formulae $\atleast{\rb}{\fg}$ and $\atleast{\rc}{\fh}$ such that for all $(\sk,\hh) \in \States$,
\begin{align*}
	\rb \leq \semapp{\fg}{(\sk,\hh)} ~\text{iff}~(\sk,\hh) \models \atleast{\rb}{\fg}
	 \qquad \text{and} \qquad
	\rc \leq \semapp{\fh}{(\sk,\hh)} ~\text{iff}~ (\sk,\hh) \models \atleast{\rc}{\fh}~.
\end{align*}

\medskip \noindent
\emph{The case $\ff = \iverson{\bb} \cdot \fg + \iverson{\neg\bb} \cdot \fh$.}
\begin{align*}
	 &\ra \leq \semapp{\iverson{\bb} \cdot \fg + \iverson{\neg\bb}\cdot\fh }{(\sk,\hh)} \\
     \text{iff}\qquad&\ra \leq \iverson{\bb}(\sk,\hh) \cdot \semapp{\fg}{(\sk,\hh)} + \iverson{\neg\bb}(\sk,\hh) \cdot \semapp{\fh}{(\sk,\hh)} \\
	 \text{iff}\qquad& (\sk,\hh)\models \bb~\text{and}~ \ra\leq \semapp{\fg}{(\sk,\hh)}
	                             \quad\text{or}\quad 
	                              (\sk,\hh)\models \neg\bb~\text{and}~ \ra\leq \semapp{\fh}{(\sk,\hh)} \\
	 \text{iff}\qquad& (\sk,\hh)\models \bb~\text{and}~ (\sk,\hh) \models \atleast{\ra}{\fg}
	 \quad\text{or}\quad 
	 (\sk,\hh)\models \neg\bb~\text{and}~ (\sk,\hh) \models \atleast{\ra}{\fh} \tag{IH}\\
	 \text{iff}\qquad& (\sk,\hh)\models (\bb \wedge \atleast{\ra}{\fg}) \vee (\neg\bb  \wedge \atleast{\ra}{\fh})
\end{align*}

\medskip \noindent
\emph{The case $\ff = \pp \cdot \fg + (1-p) \cdot \fh$.}
\begin{align*}
	&\ra \leq \semapp{\pp \cdot \fg + (1-p) \cdot \fh}{(\sk,\hh)} \\
	\text{iff}\qquad&\ra \leq \pp\cdot \semapp{\fg}{(\sk,\hh)} + (1-p) \cdot\semapp{\fh}{(\sk,\hh)}  \\
	\text{iff}\qquad & \text{there are $\rb \in \evaluationSet{\fg}, \rc\in\evaluationSet{\fh}$ with $p\cdot \rb + (1-p)\cdot \rc\geq \ra$ with } \\
	                         & \rb\leq \semapp{\fg}{(\sk,\hh)} ~\text{and}~\rc\leq \semapp{\fh}{(\sk,\hh)} \tag{monotonicity} \\
	\text{iff}\qquad & \text{there are $\rb \in \evaluationSet{\fg}, \rc\in\evaluationSet{\fh}$ with $p\cdot \rb + (1-p)\cdot \rc\geq \ra$ with } \\
	& (\sk,\hh) \models \atleast{\rb}{\fg} ~\text{and}~(\sk,\hh) \models \atleast{\rc}{\fh} \tag{IH}\\
	\text{iff}\qquad & (\sk,\hh)\models \bigvee_{\rb \in \evaluationSet{\fg}, \rc \in  \evaluationSet{\fh}, \pp\cdot\rb +(1-\pp)\cdot\rc \geq \ra } \quad \atleast{\rb}{\fg} \wedge \atleast{\rc}{\fh}~.
    \tag{$\evaluationSet{\fg}$ and $\evaluationSet{\fh}$ are finite}
\end{align*}

\medskip \noindent
\emph{The case $\ff =  \fg \cdot\fh$} is analogous to the previous case.

\medskip \noindent
\emph{The case $\ff =  1-\fg$.} 
\begin{align*}
    &\ra \leq \semapp{1-\fg}{(\sk, \hh)}\\
    \text{iff}\qquad& \ra \leq 1-\semapp{\fg}{(\sk, \hh)}\\
    \text{iff}\qquad& 0 = \ra ~\text{or}\\
                    & 0<\ra ~\text{and}~ \ra \leq 1-\semapp{\fg}{(\sk, \hh)}\\
    \text{iff}\qquad& 0 = \ra ~\text{or}\\
                    & 1>1-\ra ~\text{and}~ 1-\ra \geq \semapp{\fg}{(\sk, \hh)}\\
    \text{iff}\qquad& 0 = \ra ~\text{or}\\
                    & 1 > 1-\ra ~\text{and not}~ 1-\ra < \semapp{\fg}{(\sk, \hh)}\\
    \text{iff}\qquad& 0 = \ra ~\text{or}\\
                    & 1 > 1-\ra ~\text{and not: there exists $\rd \in \evaluationSet{\fg}$ such that }~ \\
                    & \qquad \qquad 1-\ra < \rd \leq \semapp{\fg}{(\sk, \hh)} \tag{$\semapp{\fg}{(\sk, \hh)}\in \evaluationSet{\fg}$}\\
    \text{iff}\qquad& 0 = \ra ~\text{or}\\
                    & 1 > 1-\ra ~\text{and not}~ \min \setcomp{\rb \in \evaluationSet{\fg}}{\rb > 1-\ra} \leq \semapp{\fg}{(\sk, \hh)} \tag{$\dagger$, see below}\\
    \text{iff}\qquad& 0 = \ra ~\text{or}\\
                    & 0 < \ra ~\text{and not} ~ \rd \leq \semapp{\fg}{(\sk, \hh)} ~\text{for}~ \rd = \min \setcomp{\rb \in \evaluationSet{\fg}}{\rb > 1-\ra}\\
    \text{iff}\qquad& 0 = \ra ~\text{or}\\
                    & 0 < \ra ~\text{and not} ~ (\sk, \hh) \models \atleast{\rd}{\fg} ~\text{for}~ \rd = \min \setcomp{\rb \in \evaluationSet{\fg}}{\rb > 1-\ra}\\
    \text{iff}\qquad& 0 = \ra ~\text{or}\\
                    & 0 < \ra ~\text{and} ~ (\sk, \hh) \models \neg \atleast{\rd}{\fg} ~\text{for}~ \rd = \min \setcomp{\rb \in \evaluationSet{\fg}}{\rb > 1-\ra}\\
\end{align*}
Regarding $\dagger$: If there exists a $\rd \in \evaluationSet{\fg}$ with $\rd \leq \semapp{\fg}{(\sk, \hh)}$, then we also have for all $\rb \leq \rd$ that $\rb \leq \semapp{\fg}{(\sk, \hh)}$ by transitivity. Since furthermore $\evaluationSet{\fg}$ is finite, $1-\ra < 1$ and $1 \in \evaluationSet{\fg}$ we can also pick the smallest $\rb\in \evaluationSet{\fg}$ satisfying $1-\ra < \rb$. For the other direction we have that since $\evaluationSet{\fg}$ is finite, $1 \in \evaluationSet{\fg}$ and $1>1-\ra$, the set $\setcomp{\rb \in \evaluationSet{\fg}}{\rb > 1-\ra}$ is finite and non-empty. Thus there also exists an $\rd$ such that $1-\ra < \rd \leq \semapp{\fg}{(\sk, \hh)}$.

\medskip \noindent
\emph{The case $\ff =  \emax{\fg}{\fh}$.}
\begin{align*}
    &\ra \leq \semapp{\emax{\fg}{\fh}}{(\sk, \hh)}\\
    \text{iff}\qquad&\ra \leq \max( \semapp{\fg}{(\sk, \hh)}, \semapp{\fh}{(\sk, \hh)}) \\
    \text{iff}\qquad& \ra \leq \semapp{\fg}{(\sk, \hh)} ~\text{or}~ \ra \leq \semapp{\fh}{(\sk, \hh)} \\
    \text{iff}\qquad& (\sk, \hh) \models \atleast{\ra}{\fg} ~\text{or}~ (\sk, \hh) \models \atleast{\ra}{\fh} \tag{IH}\\
    \text{iff}\qquad& (\sk, \hh) \models \atleast{\ra}{\fg} \lor \atleast{\ra}{\fh}
\end{align*}

\medskip \noindent
\emph{The case $\ff =  \emin{\fg}{\fh}$} is analogous to the previous case.

\medskip \noindent
\emph{The case $\ff =  \Sup \xx \colon \fg$.} 
\begin{align*}
    &\ra \leq \semapp{\Sup \xx \colon \fg}{(\sk, \hh)} \\
    \text{iff}\qquad& \ra \leq \max \big\{\semapp{\fg}{(\sk\statesubst{\xx}{v},\hh)} ~\mid~ v \in \Vals \big\} \\
    \text{iff}\qquad& \text{there is a $v \in \Vals$ with}~ \ra \leq \semapp{\fg}{(\sk\statesubst{\xx}{v}, \hh)} \tag{$\eval{\fg}$ is finite}\\
    \text{iff}\qquad& \text{there is a $v \in \Vals$ with}~ (\sk\statesubst{\xx}{v}, \hh) \models \atleast{\ra}{\fg} \tag{IH} \\
    \text{iff}\qquad& (\sk, \hh) \models \exists \xx \colon \atleast{\ra}{\fg}
\end{align*}

\medskip \noindent
\emph{The case $\ff =  \Inf \xx \colon \fg$} is analogous to the previous case.

\medskip \noindent
\emph{The case $\ff =  \fg \sepcon \fh$.}
\begin{align*}
    &\ra \leq \semapp{\fg \sepcon \fh}{(\sk, \hh)} \\
    \text{iff}\qquad& \ra \leq \max \left\{\semapp{\fg}{(\sk,\hh_1)} \cdot \semapp{\fh}{(\sk,\hh_2)} ~\mid~ \hh = \hh_1 \sepcon \hh_2 \right\} \\
    \text{iff}\qquad& \text{there exists $\hh_1, \hh_2 \in \Heaps$ with $\hh = \hh_1 \sepcon \hh_2$ and}\\
    & \qquad \ra \leq \semapp{\fg}{(\sk, \hh_1)} \cdot \semapp{\fh}{(\sk, \hh_2)} \tag{$\eval{\fg}$ and $\eval{\fh}$ are finite}\\
    \text{iff}\qquad& \text{there exists $\hh_1, \hh_2 \in \Heaps$ with $\hh = \hh_1 \sepcon \hh_2$ and}\\
    & \text{there exists $\rb \in \evaluationSet{\fg}$, $\rc \in \evaluationSet{\fh}$ with $\rb \cdot \rc \geq \ra$ and}\\
    & \qquad \rb \leq \semapp{\fg}{(\sk, \hh_1)} ~\text{and}~ \rc \leq \semapp{\fh}{(\sk, \hh_2)} \tag{monotonicity}\\
    \text{iff}\qquad& \text{there exists $\rb \in \evaluationSet{\fg}$, $\rc \in \evaluationSet{\fh}$ with $\rb \cdot \rc \geq \ra$ and}\\
    & \text{there exists $\hh_1, \hh_2 \in \Heaps$ with $\hh = \hh_1 \sepcon \hh_2$ and}\\
    & \qquad \rb \leq \semapp{\fg}{(\sk, \hh_1)} ~\text{and}~ \rc \leq \semapp{\fh}{(\sk, \hh_2)}\\
    \text{iff}\qquad& \text{there exists $\rb \in \evaluationSet{\fg}$, $\rc \in \evaluationSet{\fh}$ with $\rb \cdot \rc \geq ra$ and}\\
    & \text{there exists $\hh_1, \hh_2 \in \Heaps$ with $\hh = \hh_1 \sepcon \hh_2$ and}\\
    & \qquad (\sk, \hh_1) \models \atleast{\rb}{\fg} ~\text{and}~ (\sk, \hh_2) \models \atleast{\rc}{\fh} \tag{IH}\\
    \text{iff}\qquad& \text{there exists $\rb \in \evaluationSet{\fg}$, $\rc \in \evaluationSet{\fh}$ with $\rb \cdot \rc \geq \ra$ and}\\
    & \qquad (\sk, \hh) \models \atleast{\rb}{\fg} \sepcon \atleast{\rc}{\fh}\\
    \text{iff}\qquad& (\sk, \hh) \models \bigvee_{\rb \in \evaluationSet{\fg}, \rc \in  \evaluationSet{\fh}, \rb \cdot \rc \geq \ra } \atleast{\rb}{\fg} \sepcon \atleast{\rc}{\fh} \tag{$\evaluationSet{\fg}$ and $\evaluationSet{\fh}$ are finite}\\
\end{align*}

\medskip \noindent
\emph{The case $\ff =  \iverson{\slb}\sepimp \fg$.}
\begin{align*}
    &\ra \leq \semapp{\iverson{\slb}\sepimp \fg}{(\sk, \hh)} \\
    \text{iff}\qquad& \ra \leq \inf \left\{ \semapp{\fg}{(\sk,\hh\sepcon\hh')}  ~\mid~ \hh' \disjoint \hh ~\text{and}~ \iverson{\slb}(\sk,\hh')=1 \right\} \\
    \text{iff}\qquad& \text{for all $\hh' \in \Heaps$ with $\hh' \disjoint \hh$ and $(\sk, \hh') \models \slb$ we have} ~ \ra \leq \semapp{\fg}{(\sk, \hh \sepcon \hh')} \tag{$\dagger$ see below}\\
    \text{iff}\qquad& \text{for all $\hh' \in \Heaps$ with $\hh' \disjoint \hh$ and $(\sk, \hh') \models \slb$ we have} ~ (\sk, \hh \sepcon \hh') \models \atleast{\ra}{\fg} \tag{IH} \\
    \text{iff}\qquad& (\sk, \hh) \models \slb \sepimp \atleast{\ra}{\fg}
\end{align*}
Regarding $\dagger$ we have the following cases: Either there is no $\hh'$ with $\hh' \disjoint \hh$ and $(\sk, \hh') \models \slb$, then both statements hold for all $\ra\in \Probs$. If there is at least one $\hh'$ with $\hh' \disjoint \hh$ and $(\sk, \hh') \models \slb$, then since $\eval{\fg}$ is finite, there is a $\hh'$ with $\hh \disjoint \hh'$ and $(\sk, \hh') \models \slb$ that minimizes $\semapp{\fg}{(\sk, \hh \sepcon \hh')}$ and we let $\rb = \semapp{\fg}{(\sk, \hh \sepcon \hh')}$. Thus $\ra \leq \rb$ if and only if for all $\hh'$ with $\hh' \disjoint \hh$ and $(\sk, \hh') \models \slb$ we have $\ra \leq \semapp{\fg}{(\sk, \hh \sepcon \hh')}$.

\bigskip \noindent
This concludes the proof.
\end{proof}

\section{Appendix to \Cref{sec:complexity}}
\label{app:complexity}
\begin{table}[t]
	\centering
	\caption{Inductive definitions of the size of $\SLA$ and $\QSLA$ formulae.}%
	\label{tab:sizes_sl_qsl}%
	\renewcommand{\arraystretch}{1.2}%
	\begin{tabular}[t]{l@{\quad}l@{\quad}l@{\quad}l}
		\hline\hline
		$\sla$			&          $\sizeof{\sla}$			\\
		\hline
		$\slb\in \Predset$		& $\sizeof{\sla}$		              \\
		$\neg \slb$		& $1+ \sizeof{\slb}$		              \\
		$\slb \wedge \slc$		& $1+\sizeof{\slb}+\sizeof{\slc}$ 		              \\
		$\slb \vee \slc$		& $1+\sizeof{\slb}+\sizeof{\slc}$ 		              \\
		$\exists \xx \colon \slb$  & $1+\sizeof{\slb}$ \\
		$\forall \xx \colon \slb$  & $1+\sizeof{\slb}$\\
		$\slb \sepcon \slc$ &  $1+\sizeof{\slb}+\sizeof{\slc}$ \\
		$\slb \sepimp \slc$ &  $1+\sizeof{\slb}+\sizeof{\slc}$ \\
		\hline\hline
	\end{tabular}%
	\qquad
	\begin{adjustbox}{scale=0.87}
	\begin{tabular}[t]{l@{\quad}l@{\quad}l@{\quad}l@{\quad}l}
		\hline\hline
		$\ff$			&          $\sizeof{\ff} $   &	$\probsizeof{\ff}$		\\
		\hline
		$\iverson{\slb}$		& $\sizeof{\slb}$ & $0$	              \\
		$\iverson{\bb} \cdot \fg + \iverson{\neg\bb} \cdot \fh $ & $1+\sizeof{\bb}+\sizeof{\fg}+\sizeof{\neg\bb}+\sizeof{\fh}$ & $\probsizeof{\fg}+\probsizeof{\fh}$ \\
		$\qq \cdot \fg + (1-\qq) \cdot \fh $ & $1+\sizeof{\fg}+\sizeof{\fh}$ & $1+\probsizeof{\fg}+\probsizeof{\fh}$ \\
		$\fg \cdot \fh $  & $1+\sizeof{\fg}+\sizeof{\fh}$ & $1+\probsizeof{\fg}+\probsizeof{\fh}$\\
		$1-\fg$ & $1+\sizeof{\fg}$  & $\probsizeof{\fg}$\\
		$\emax{\fg}{\fh}$ & $1+\sizeof{\fg}+\sizeof{\fh}$ & $\probsizeof{\fg}+\probsizeof{\fh}$\\
		$\emin{\fg}{\fh}$ & $1+\sizeof{\fg}+\sizeof{\fh}$ & $\probsizeof{\fg}+\probsizeof{\fh}$\\
		$\Sup \xx \colon \fg $ & $1+\sizeof{\fg}$ & $\probsizeof{\fg}$\\
		$\Inf \xx \colon \fg $ & $1+\sizeof{\fg}$ & $\probsizeof{\fg}$\\
		$\fg \sepcon \fh $ & $1+\sizeof{\fg}+\sizeof{\fh}$ & $1+\probsizeof{\fg}+\probsizeof{\fh}$\\
		$\iverson{\slb} \sepimp \fg$ & $1+\sizeof{\fg}+\sizeof{\slb}$ & $\probsizeof{\fg}$\\
		\hline\hline
	\end{tabular}%
	\end{adjustbox}
	\renewcommand{\arraystretch}{1}%
	%
	%\vspace*{.5em}%
\end{table}

\begin{restateTheorem}{thm:evalset_size}
	We have $\sizeof{\evaluationSet{\ff}} \leq 2^{\probsizeof{\ff}+1}$. Hence, checking $\ff \entails \fg$ by means of \Cref{thm:qsla_by_sla} requires checking $2^{\bigo(\probsizeof{\ff})}$ entailments in $\SLA$.
\end{restateTheorem}
\begin{proof}
    We prove this by induction of $\ff$.

	\medskip \noindent
	\emph{For the base case for $\ff=\iverson{\slb}$} we have $\sizeof{\evaluationSet{\iverson{\slb}}}=2^{0+1}=2^{\probsizeof{\iverson{\slb}}+1}$.

	\medskip \noindent
	For the composite cases we assume that for some fixed, but arbitrary formulae $\fg, \fh \in \QSLA$, the inequalities $\sizeof{\evaluationSet{\fg}} \leq 2^{\probsizeof{\fg}+1}$ and $\sizeof{\evaluationSet{\fh}} \leq 2^{\probsizeof{\fh}+1}$ hold. 

	\medskip \noindent
	\emph{For the induction steps $\ff=\iverson{\bb} \cdot \fg + \iverson{\neg\bb} \cdot \fh$,} we distinguish three cases. 

	\begin{enumerate}
		\item $\probsizeof{\fg}=0$. Then by the induction hypothesis, we have $\sizeof{\evaluationSet{\fg}}\leq 2$. By \Cref{lem:zero_one_preservation}, we have $\{0, 1\} = \evaluationSet{\fg}$ and $0, 1 \in \evaluationSet{\fh}$. However, then the union will not increase the set, i.e. $\evaluationSet{\ff}=\evaluationSet{\fh}$. Finally we have $\sizeof{\evaluationSet{\ff}}=\sizeof{\evaluationSet{\fh}}\leq2^{\probsizeof{\fh}+1}= 2^{\probsizeof{\ff}+1}$ by the induction hypothesis.

		\item $\probsizeof{\fh}=0$ is analogous.
	
		\item $0<\probsizeof{\fg}, \probsizeof{\fh}$. Then the size of the set is at most the sum of each set $\sizeof{\evaluationSet{\ff}}\leq \sizeof{\evaluationSet{\fg}}+\sizeof{\evaluationSet{\fh}}$. By the induction hypothesis we then have 
		\[ \sizeof{\evaluationSet{\ff}} \leq \underbrace{2^{\probsizeof{\fg}+1}+2^{\probsizeof{\fh}+1} \leq 2 \cdot 2^{\probsizeof{\fg}} \cdot 2^{\probsizeof{\fh}}}_{\text{because}~ 0<\probsizeof{\fg}, \probsizeof{\fh}} = 2^{\probsizeof{\fg}+\probsizeof{\fh}+1} = 2^{\probsizeof{\ff}+1}~.\]
	\end{enumerate}
	
	\medskip \noindent
	\emph{For the induction steps $\ff=\qq \cdot \fg + (1-\qq) \cdot \fh$,} we have that $\sizeof{\evaluationSet{\ff}} \lleq \sizeof{\evaluationSet{\fg}} \cdot \sizeof{\evaluationSet{\fh}}$. By the induction hypothesis we deduce the upper bound 
	\[ \evaluationSet{\ff}\leq 2^{\probsizeof{\fg}+1} \cdot 2^{\probsizeof{\fh}+1} = 2^{\probsizeof{\fg}+\probsizeof{\fh}+1+1}=2^{\probsizeof{\ff}+1}~.\]

	\medskip \noindent
	\emph{The induction step $\ff=\fg \cdot \fh$} is analogous to the previous case.

	\medskip \noindent
	\emph{For the induction step $\ff=1- \fg$,} we have that the set does not change in size $\sizeof{\evaluationSet{\ff}}=\sizeof{\evaluationSet{\fg}}$ and $\probsizeof{\ff}=\probsizeof{\fg}$. Thus this follows directly from the induction hypothesis.

	\medskip\noindent
	\emph{The induction steps $\ff=\emin{\fg}{\fh}$ and $\ff=\emax{\fg}{\fh}$} are analogous to case $\ff=\iverson{\bb} \cdot \fg + \iverson{\neg\bb} \cdot \fh$.

	\medskip\noindent
	\emph{The induction steps $\ff=\Sup \xx \colon \fg$ and $\ff=\Inf \xx \colon \fg$} are analogous to the case $\ff=1-\fg$.

	\medskip \noindent
	\emph{The induction step $\ff=\fg \sepcon \fh$} is analogous to the case $\ff=\fg \cdot \fh$.

	\medskip \noindent
	\emph{The induction step $\ff=\iverson{\slb} \sepimp \fg$} is analogous to the case $\ff=1-\fg$.

\end{proof}
\bigskip

\begin{restateTheorem}{thm:atleast_size}
	For any formulae $\ff \in \QSLA$ and all probabilities $\ra \in \Probs$, the $\SLA$ formula $\atleast{\ra}{\ff}$ has at most size $3 \cdot \sizeof{\ff} \cdot 2^{(\probsizeof{\ff}+1)^2}$. Hence the size of the formula $\atleast{\ra}{\ff}$ is in $\bigo(\sizeof{\ff}) \cdot 2^{\bigo(\probsizeof{\ff}^2)}$.
\end{restateTheorem}
\begin{proof}
	We show by induction on $\ff$ that $\sizeof{\atleast{\ra}{\ff}} \leq 3 \cdot \sizeof{\ff} \cdot 2^{(\probsizeof{\ff}+1)^2}$ for all $\ra \in \Probs$. 

	\medskip \noindent
	\emph{For the base case $\ff=\iverson{\slb}$} we have $\sizeof{\atleast{\ra}{\ff}} \leq  \sizeof{\slb} = \sizeof{\slb} \cdot 2^{1} \leq 3 \cdot \sizeof{\ff} \cdot 2^{(\probsizeof{\ff}+1)^2}$.

	\medskip \noindent
	For the composite cases we assume that for some arbitrary, but fixed formulae $\fg, \fh \in \QSLA$ and all probabilities $\ra \in \Probs$ the inequalities for both formulae $\sizeof{\atleast{\ra}{\fg}} \leq 3 \cdot \sizeof{\fg} \cdot 2^{(\probsizeof{\fg}+1)^2}$ and $\sizeof{\atleast{\ra}{\fh}} \leq 3 \cdot \sizeof{\fh} \cdot 2^{(\probsizeof{\fh}+1)^2}$ hold.

	\medskip \noindent
	\emph{For the induction step $\ff=\iverson{\bb} \cdot \fg + \iverson{\neg\bb} \cdot \fh$} we have that 
	\begin{align*}
				& \sizeof{\atleast{\ra}{\ff}} \\
	= \quad		& \sizeof{(\iverson{\bb} \land \atleast{\ra}{\fh}) \lor (\iverson{\neg \bb} \land \atleast{\ra}{\fh})} \\
	= \quad		& \sizeof{\iverson{\bb} + 1 + \sizeof{\atleast{\ra}{\fh}} + 1 + \sizeof{\iverson{\neg \bb}} + 1 + \sizeof{\atleast{\ra}{\fh}}} \\
	= \quad		& \sizeof{\bb}+\sizeof{\atleast{\ra}{\fg}}+\sizeof{\neg\bb}+\sizeof{\atleast{\ra}{\fh}}+3  \\
	\leq \quad  &  \sizeof{\bb} + %
				   3 \cdot \sizeof{\fg} \cdot 2^{(\probsizeof{\fg}+1)^2} + %
				   \sizeof{\neg\bb} +%
				   3 \cdot \sizeof{\fh} \cdot 2^{(\probsizeof{\fh}+1)^2} + %
				   3 \tag{IH} \\
	\leq \quad  &  \sizeof{\bb} + %
				   3 \cdot \sizeof{\fg} \cdot 2^{(\probsizeof{\fg}+\probsizeof{\fh}+1)^2} + %
				   \sizeof{\neg\bb} +%
				   3 \cdot \sizeof{\fh} \cdot 2^{(\probsizeof{\fg}+\probsizeof{\fh}+1)^2} + %
				   3 \\
	\leq \quad  & 3 \cdot ( \sizeof{\bb} + \sizeof{\fg} + \sizeof{\neg\bb} + \sizeof{\fh} + 1) \cdot 2^{(\probsizeof{\fg}+\probsizeof{\fh}+1)^2} \\
	= \quad     & 3 \cdot ( \sizeof{\bb} + \sizeof{\fg} + \sizeof{\neg\bb} + \sizeof{\fh} + 1) \cdot 2^{(\probsizeof{\ff}+1)^2} \tag{$\probsizeof{\ff}=\probsizeof{\fg}+\probsizeof{\fh}$}\\
	= \quad     & 3 \cdot \sizeof{\ff} \cdot 2^{(\probsizeof{\ff}+1)^2} \tag{$\sizeof{\ff}=\sizeof{\bb}+\sizeof{\fg}+\sizeof{\neg\bb}+\sizeof{\fh}+1$}
	\end{align*} 

	\medskip \noindent
	\emph{For the induction step $\ff=\pp \cdot \fg + (1-\pp) \cdot \fh$,} we have for $\rb, \rc \in \Probs$ maximizing $\sizeof{\atleast{\rb}{\fg}}$ and $\sizeof{\atleast{\rc}{\fh}}$ that
	\[ \sizeof{\atleast{\ra}{\ff}} \lleq \left( \sizeof{\evaluationSet{\fg}} \cdot \sizeof{\evaluationSet{\fh}} \right) \cdot \left( \sizeof{\atleast{\rb}{\fg}} + \sizeof{\atleast{\rc}{\fh}} + 2 \right)~, \]
	thus using the induction hypothesis and \Cref{thm:evalset_size}
	\[ \sizeof{\atleast{\ra}{\ff}} \lleq \left( 2^{\probsizeof{\fg}+1}\cdot 2^{\probsizeof{\fh}+1} \right) \cdot \left( 3 \cdot \sizeof{\fg} \cdot 2^{(\probsizeof{\fg}+1)^2} + 3 \cdot \sizeof{\fh} \cdot 2^{(\probsizeof{\fh}+1)^2} + 2 \right)~. \tag{$\dagger$}\]
	Remark that $1\leq \probsizeof{\ff}$ by assumption. We now prove the upper bound on the right side by distinguishing two cases on $\probsizeof{\ff}$:
	
	\begin{enumerate}
		\item $\probsizeof{\ff}=1$. Then $\probsizeof{\fg}=\probsizeof{\fh}=0$ and we have:
		\begin{align*}
					& \sizeof{\atleast{\ra}{\ff}} \\
		\leq \quad 	& 2 \cdot 2 \cdot (3 \cdot \sizeof{\fg} \cdot 2 + 3 \cdot \sizeof{\fh} \cdot 2 + 2) \tag{($\dagger$) and $\probsizeof{\ff}=1$} \\
		= \quad     & 2^3 \cdot 3 \cdot ( \sizeof{\fg} + \sizeof{\fh} + \sfrac{1}{5})  \\
		\leq \quad  & 2^3 \cdot 3 \cdot \sizeof{\ff} \tag{$\sizeof{\ff}=\sizeof{\fg}+\sizeof{\fh}+1$}\\
		\leq \quad  & 3 \cdot \sizeof{\ff} \cdot 2^{(\probsizeof{\ff}+1)^2} \tag{$\probsizeof{\ff}=1$}
		\end{align*}

		\item $\probsizeof{\ff} > 1$. Here we require a bit more mathematical tools in form of \Cref{lem:difficult_complexity_inequality}. Then we have:
		\begin{align*}
						& \sizeof{\atleast{\ra}{\ff}} \\
			\leq \quad 	& \left( 2^{\probsizeof{\fg}+1}\cdot 2^{\probsizeof{\fh}+1} \right) \cdot \left( 3 \cdot \sizeof{\fg} \cdot 2^{(\probsizeof{\fg}+1)^2} + 3 \cdot \sizeof{\fh} \cdot 2^{(\probsizeof{\fh}+1)^2} + 2 \right) \tag{$\dagger$}\\
			= \quad     & 2^{\probsizeof{\ff}+1} \cdot \left( 3 \cdot \sizeof{\fg} \cdot 2^{(\probsizeof{\fg}+1)^2} + 3 \cdot \sizeof{\fh} \cdot 2^{(\probsizeof{\fh}+1)^2} + 2 \right) \hspace{5em} \tag{$\probsizeof{\ff}=\probsizeof{\fg}+\probsizeof{\fh}+1$}\\
			= \quad     & 3 \cdot \sizeof{\fg} \cdot 2^{\probsizeof{\ff}+1} \cdot 2^{(\probsizeof{\fg}+1)^2} + 3 \cdot \sizeof{\fh} \cdot 2^{\probsizeof{\ff}+1} \cdot 2^{(\probsizeof{\fh}+1)^2} + 2^{\probsizeof{\ff}+1} \cdot 2\\
			\leq \quad  & \quad 3 \cdot (\sizeof{\fg}+\sizeof{\fh}) \cdot 2^{\probsizeof{\ff}+1} \cdot 2^{(\probsizeof{\fg}+\probsizeof{\fh}+1)^2} \\
						& + 3 \cdot (\sizeof{\fg}+\sizeof{\fh}) \cdot 2^{\probsizeof{\ff}+1} \cdot 2^{(\probsizeof{\fg}+\probsizeof{\fh}+1)^2} \\
						& + 2 \cdot 2^{\probsizeof{\ff}+1}\\
			\leq \quad 	& 2 \cdot 3 \cdot (\sizeof{\fg}+\sizeof{\fh}) \cdot 2^{\probsizeof{\ff}+1} \cdot 2^{(\probsizeof{\fg}+\probsizeof{\fh}+1)^2}  + 2^{\probsizeof{\ff}+2}\\
			\leq \quad 	& 2 \cdot 3 \cdot \sizeof{\ff} \cdot 2^{\probsizeof{\ff}+1} \cdot 2^{(\probsizeof{\fg}+\probsizeof{\fh}+1)^2}  + 2^{\probsizeof{\ff}+2} \tag{$\sizeof{\ff}=\sizeof{\fg}+\sizeof{\fh}+1$}\\
			= \quad 	& 2 \cdot 3 \cdot \sizeof{\ff} \cdot 2^{\probsizeof{\ff}+1} \cdot 2^{\probsizeof{\ff}^2} + 2^{\probsizeof{\ff}+2} \tag{$\probsizeof{\ff}=\probsizeof{\fg}+\probsizeof{\fh}+1$}\\
			= \quad 	& 3 \cdot \sizeof{\ff} \cdot 2^{\probsizeof{\ff}^2+\probsizeof{\ff}+2} + 2^{\probsizeof{\ff}+2} \\
			\leq \quad 	& 3 \cdot \sizeof{\ff} \cdot \left( 2^{\probsizeof{\ff}^2+\probsizeof{\ff}+2} + 2^{\probsizeof{\ff}+2} \right)\\
			\leq \quad  & 3 \cdot \sizeof{\ff} \cdot 2^{(\probsizeof{\ff}+1)^2} \tag{by $1<\probsizeof{\ff}$ and \Cref{lem:difficult_complexity_inequality}}
		\end{align*}

	\end{enumerate}

	\medskip \noindent
	\emph{The induction steps for $\ff=\fg \cdot \fh$} is analogous to the previous case.

	\medskip \noindent
	\emph{For the induction step $\ff=1-\fg$} we have that
	\begin{align*}
				& \sizeof{\atleast{\ra}{\ff}} \\
	\leq \quad  & \sizeof{\atleast{\ra}{\fg}}+1 \\
	\leq \quad  & 3 \cdot \sizeof{\fg} \cdot 2^{(\probsizeof{\fg}+1)^2} + 1 \tag{IH}\\
	\leq \quad 	& 3 \cdot \sizeof{\fg} \cdot 2^{(\probsizeof{\ff}+1)^2} +  2^{(\probsizeof{\ff}+1)^2} \tag{$0\leq \probsizeof{\fg} = \probsizeof{\ff} $} \\
	= \quad     & 3 \cdot (\sizeof{\fg} +\sfrac{1}{3}) \cdot 2^{(\probsizeof{\ff}+1)^2}\\
	\leq \quad  & 3 \cdot \sizeof{\ff} \cdot 2^{(\probsizeof{\ff}+1)^2} \tag{$\sizeof{\ff}=\sizeof{\fg}+1$}
	\end{align*}

	\medskip \noindent
	\emph{The induction steps for $\emax{\fg}{\fh}$, $\emin{\fg}{\fh}$, $\Sup \xx \colon \fg$ and $\Inf \xx \colon \fg$} are analogous to the previous case.

	\medskip \noindent
	\emph{The induction steps for $\fg \sepcon \fh$} is analogous to the case $\ff=\pp \cdot \fg + (1-\pp) \cdot \fh$.

	\medskip \noindent
	\emph{For the induction step $\iverson{\slb} \sepimp \fg$} we have that
	\begin{align*}
				& \sizeof{\atleast{\ra}{\ff}} \\
	= \quad  	& \sizeof{\atleast{\ra}{\fg}}+\sizeof{\slb}+1 \\
	\leq \quad  & 3 \cdot \sizeof{\fg} \cdot 2^{(\probsizeof{\fg}+1)^2} + \sizeof{\slb}+1 \tag{IH}\\
	\leq \quad 	& 3 \cdot \sizeof{\fg} \cdot 2^{(\probsizeof{\ff}+1)^2} + (\sizeof{\slb}+1) \cdot 2^{(\probsizeof{\ff}+1)^2} \tag{$0\leq \probsizeof{\fg} = \probsizeof{\ff} $} \\
	= \quad  & 3 \cdot (\sizeof{\fg} + \sfrac{\sizeof{\slb}}{3}+\sfrac{1}{3}) \cdot 2^{(\probsizeof{\ff}+1)^2}\\
	\leq \quad  & 3 \cdot \sizeof{\ff} \cdot 2^{(\probsizeof{\ff}+1)^2} \tag{$\sizeof{\ff}=\sizeof{\fg}+\sizeof{\slb}+1$}
	\end{align*}

	\medskip \noindent
	This concludes the proof.
\end{proof}
\medskip
\begin{lemma}\label{lem:difficult_complexity_inequality}
	For all natural numbers $n > 1$, we have $2^{n^2+n+2}+2^{n+2} \leq 2^{(n+1)^2}$.
\end{lemma}
\begin{proof}
By induction over $n$.

\smallskip \noindent
\emph{For the base case $n=2$,} we have $2^{2^2+2+2}+2^{2+2} = 2^{8}+2^4 < 2^8+2^8 = 2^{(2+1)^2}$.

\smallskip \noindent
Now we assume that for some fixed, but arbitrary natural number $n>1$ the inequality $2^{n^2+n+2}+2^{n+2} \leq 2^{(n+1)^2}$ holds.

\smallskip \noindent
\emph{For the induction step $n \rightarrow n+1$,} we have 
\begin{align*}
				& 2^{(n+1)^2+n+1+2}+2^{n+1+2} \\
	= \quad     & 2^{n^2+2n+1+n+1+2} + 2^{n+1+2} \\
	= \quad     & 2^{n^2+n+2} \cdot 2^{2n+2} + 2^{n+2} \cdot 2 \\
	\leq \quad  & 2^{2n+2} \cdot \left( 2^{n^2+n+2} + 2^{n+2} \right) \\
	\leq \quad  & 2^{2n+2} \cdot 2^{(n+1)^2} \tag{IH} \\
	= \quad     & 2^{n^2+4n+3} \\
	\leq \quad  & 2^{n^2+4n+4} \\
	= \quad     & 2^{(n+1+1)^2}~.
\end{align*}
This concludes the proof.
\end{proof}

\section{Appendix to \Cref{sec:applications}}
\label{app:applications}
\begin{restateLemma}{lem:sla_requirements}
	Let $\QSLfrag$ be a $\QSLA$ fragment. If an $\SLA$ fragment $\SLfrag$ satisfies the requirements provided in \Cref{tab:sla_requirements}, then $\SLfrag$ is $\QSLfrag$-admissible.
\end{restateLemma}
\begin{proof}
    By induction on $\ff$.

    \medskip \noindent
    \emph{For the base case $\ff=\iverson{\slb}$,} we have $\atleast{\ra}{\iverson{\slb}}= \true ~\text{if}~ \ra=0 ~\text{otherwise}~ \iverson{\slb}$, thus $\true$ and $\slb$ are required. 

    \medskip \noindent
    For all other composite cases we assume for some fixed, but arbitrary $\fg, \fh \in \QSLA$ that for all $\rb, \rc \in \Probs$ the formulae $\atleast{\rb}{\fg}$ and $\atleast{\rc}{\fh}$ satisfy the requirements.

    \medskip \noindent
    \emph{For the case $\ff=\iverson{\bb} \cdot \fg + \iverson{\neg \bb} \cdot \fh$,} we have $\atleast{\ra}{\ff}=(\bb \land \atleast{\ra}{\fg}) \lor (\neg \bb \land \atleast{\ra}{\fh})$, since by the induction hypothesis $\atleast{\ra}{\fg}$ and $\atleast{\ra}{\fh}$ already satisfy all requirements, we only require additionally $\bb, \neg \bb, \land$ and $\lor$.

    \medskip \noindent
    \emph{For the case $\ff=\pp \cdot \fg + (1-\pp) \cdot \fh$,} we have 
    \[ \atleast{\ra}{\ff} = \bigvee_{\rb \in \evaluationSet{\fg}, \rc \in  \evaluationSet{\fh}, \pp\cdot\rb +(1-\pp)\cdot\rc \geq \ra } \quad \atleast{\rb}{\fg} \wedge \atleast{\rc}{\fh}~. \]
    Here we have two observations: 
    \begin{enumerate}
        \item There are only finitely many disjunctions since $\evaluationSet{\fg}$ and $\evaluationSet{\fh}$ is finite by \Cref{thm:evalset}.
        \item The disjunctions is not empty since $1 \in \evaluationSet{\fg}$ and $1 \in \evaluationSet{\fh}$ by \Cref{lem:zero_one_preservation} and $\pp \cdot 1 + (1-\pp) \cdot 1 = 1 \geq \ra$ for all $\ra \in \Probs$.
    \end{enumerate}
    Thus, for any $\ra \in \Probs$, we can construct the big disjunction by only using $\land$, $\lor$, $\atleast{\rb}{\fg}$ and $\atleast{\rc}{\fh}$ for all $\rb, \rc \in \Probs$.
    Since $\atleast{\rb}{\fg}$ and $\atleast{\rc}{\fh}$ satisfy all requirements by the induction hypothesis, we only require additionally $\land$ and $\lor$.

    \medskip \noindent
    \emph{The case $\ff= \fg \cdot \fh$} is analogous to the previous case.

    \medskip \noindent
    \emph{For the case $\ff= 1- \fg$,} we have $\atleast{\ra}{\ff}=\true ~\text{if}~ \ra=0 ~\text{otherwise}~ \neg \atleast{\rd}{\fg}$ where $\rd = \min \setcomp{\rb \in \evaluationSet{\fg}}{\rb > 1-\ra}$. Remark that we compute $\rd$ during the construction of $\atleast{\ra}{\ff}$ and not during the checking of $(\sk, \hh) \models \atleast{\ra}{\ff}$. Since $\atleast{\rd}{\fg}$ already satisfy all requirements by the induction hypothesis, we thus only additionally require $\neg$ and $\true$.

    \medskip \noindent
    \emph{The cases $\ff=\emax{\fg}{\fh}$, $\ff=\emin{\fg}{\fh}$, $\ff=\Sup \xx \colon \fg$ and $\ff=\Inf \xx \colon \fg$} are analogous to the case $\ff=\iverson{\bb} \cdot \fg + \iverson{\neg \bb} \cdot \fh$.

    \medskip \noindent
    \emph{The case $\ff=\fg \sepcon \fh$} is analogous to the case $\ff=\pp \cdot \fg + (1-\pp) \cdot \fh$.

    \medskip \noindent
    \emph{The case $\ff=\iverson{\slb} \sepimp \fg$} is analogous to the case $\ff=\iverson{\bb} \cdot \fg + \iverson{\neg \bb} \cdot \fh$.

    \medskip \noindent
    This concludes the proof.
\end{proof}

\section{Appendix to \Cref{sec:qsh}}
\label{app:qsh}
\begin{restateTheorem}{thm:decide_wp_qsh}
	For loop- and allocation-free $\hpgcl$ programs $\cc$
    (that only perform pointer operations, no arithmetic and guards of the pure fragment of $\Predset$) and $\ff_1,\ff_2 \in \QSH$,
    it is decidable whether the entailment
	$\wlp{\cc}{\ff_1} \entails \ff_2$ holds.
\end{restateTheorem}
\subsubsection{Proof of \Cref{thm:decide_wp_qsh}}
The proof requires extended quantitative symbolic heaps:
\begin{restateDefinition}{def:eQSH}
	The set $\eQSH$ of \emph{extended quantitative symbolic heaps} is given by the grammar 
	\begin{align*}
		\fg \quad \rightarrow \quad 
		&\iverson{\preda}
		~\mid~ \iverson{\BB} \cdot \fg + \iverson{\neg\BB}\cdot \fg 
		~\mid~ \qq \cdot \fg + (1-\qq)\cdot \fg 		
		~\mid~ \fg\sepcon\fg  \\
		& ~\mid~ \Sup \xx \colon \fg 
		~\mid~ \singleton{\xx}{\tuplenotation{\xy_1,\ldots,\xy_k}} \sepimp \fg~.
		\tag*{$\triangle$}
	\end{align*}
\end{restateDefinition}
Notice that indeed $\QSH \subseteq \eQSH$.
\begin{restateLemma}{lem:qsh_wp_closed}
		For every loop- and allocation-free program $\cc \in \hpgcl$ without arithmetic and only with guards of the pure fragment of $\Predset$, extended quantitative symbolic heaps are closed under $\wlpC{\cc}$, i.e.,
	\[
	   \text{for all $\fg \in \eQSH$}\colon \quad \wlp{\cc}{\fg} \in \eQSH~.
	\]
	In particular, since $\QSH \subseteq \eQSH$, we have 
	\[
		\text{for all $\ff \in \QSH$}\colon \quad \wlp{\cc}{\ff} \in \eQSH~. 
	\]
\end{restateLemma}
\begin{proof}
	Since we not allow arithmetic in expressions, we only have expressions of the form $\ee = \xy$. Moreover, $\eQSH$ is trivially closed under the substitution $\fg \subst{\xx}{\xy}$ for any $\xx, \xy \in \Vars$. Now we prove the lemma by induction on the structure of loop- and allocation-free program $\cc$ with $\recnum = 1$, however the proof is easy to adapt for any $\recnum$.

	\medskip \noindent
	\emph{For the base case $\cc=\SKIP$} we have $\wlp{\SKIP}{\fg}=\fg \in \eQSH$.

	\medskip \noindent
	\emph{For the base case $\cc=\ASSIGN{x}{\xy}$} we have 
	\[\wlp{\ASSIGN{x}{\xy}}{\fg} \eeq \fg \subst{\xx}{\xy} \iin \eQSH~. \]

	\medskip \noindent
	\emph{For the base case $\cc=\FREE{\xy}$} we have 
	\[\wlp{\FREE{\xy}}{\fg} \eeq \Sup \xz \colon \singleton{\xy}{\xz} \sepcon \fg \iin \eQSH\] 
	where $\xz$ is fresh.

	\medskip \noindent
	\emph{For the base case $\cc=\ASSIGNH{\xx}{\xy}$} we have 
	\[\wlp{\ASSIGNH{\xx}{\xy}}{\fg} \eeq \Sup \xz \colon \singleton{\xx}{\xz} \sepcon (\singleton{\xx}{\xy} \sepimp \ff \subst{\xx}{\xy}) \iin \eQSH \]
	where $\xz$ is fresh.

	\medskip \noindent
	\emph{For the base case $\cc=\HASSIGN{\xx}{\xy}$} we have 
	\[ \wlp{\HASSIGN{\xx}{\xy}}{\fg} \eeq \Sup \xz \colon \singleton{\xx}{\xz} \sepcon (\singleton{\xx}{\xy} \sepimp \ff ) \iin \eQSH \] 
	where $\xz$ is fresh.

	\medskip \noindent
	For all other composite cases we assume for some fixed, but arbitrary loop- and allocation-free programs $\cc_1, \cc_2 \in \hpgcl$ without arithmetic and only with guards of the pure fragment of $\Predset$ such that for all $\fg \in \eQSH \colon \quad \wlp{\cc_1}{\fg} \in \eQSH$ and for all $\fh \in \eQSH \colon \wlp{\cc_2}{\fh} \in \eQSH$.

	\medskip \noindent
	\emph{For the case $\cc=\PCHOICE{\cc_1}{\pp}{\cc_2}$} we have 
	\[\wlp{\PCHOICE{\cc_1}{\pp}{\cc_2}}{\fg} \eeq \pp \cdot \wlp{\cc_1}{\fg} + (1- \pp) \cdot \wlp{\cc_2}{\fg} \iin \eQSH\] 
	by the induction hypothesis.
	
	\medskip \noindent
	\emph{For the case $\cc=\COMPOSE{\cc_1}{\cc_2}$} we have 
	\[\wlp{\COMPOSE{\cc_1}{\cc_2}}{\fg} \eeq \wlp{\cc_1}{\wlp{\cc_2}{\fg}} \iin \eQSH\] 
	by the induction hypothesis.

	\medskip \noindent
	\emph{For the case $\cc=\ITE{\guard}{\cc_1}{\cc_2}$} we have either $\guard= (\xx = \xy)$, then 
	\begin{gather*}
		\wlp{\ITE{\xx = \xy}{\cc_1}{\cc_2}}{\fg} \qquad \qquad \qquad \qquad \qquad \\
		\qquad \qquad \qquad \eeq \iverson{\xx = \xy} \cdot \wlp{\cc_1}{\fg} + \iverson{\xx \neq \xy} \cdot \wlp{\cc_2}{\fg} \iin \eQSH
	\end{gather*}
	by the induction hypothesis; or $\guard=(\xx \neq \xy)$ which is analogous.

	\medskip \noindent
	This concludes the proof.
\end{proof}
Hence, if $\fg \models \ff$ is decidable for $\fg \in \eQSH$ and $\ff \in \QSH$, \Cref{thm:decide_wp_qsh} follows.
\begin{restateLemma}{lem:decide_entailment_eqsh_qsh}
	For $\fg \in \eQSH$ and $\ff \in \QSH$, it is decidable whether 
	      $\fg \entails \ff$ holds.
\end{restateLemma}
\begin{proof}
	We employ \Cref{lem:sla_requirements} to determine two $\SLA$ fragments $\SLfrag_1, \SLfrag_2$ such that $\SLfrag_1$ is $\eQSH$-admissible and $\SLfrag_2$ is $\QSH$-admissible. 
	Then, by \Cref{thm:decide_entailment_qsl_by_sl}, decidability of $\fg \entails \ff$ follows from decidability of $\sla\entails\slb$ for $\sla \in \SLfrag_1$ and $\slb \in \SLfrag_2$. For that, we exploit the equivalence
	\[
	\sla\entails\slb \qquad \text{iff} \qquad \sla \wedge \neg\slb ~\text{is unsatisfiable}~.
	\]
	The latter is decidable by \cite[Theorem~3.3]{Echenim2020Bernays} since  $\sla \wedge \neg\slb$ is equivalent to a formula of the form $\exists^*\forall^* \colon \slc$ with $\slc$ quantifier-free and such that no formula $\slc_1 \sepimp \slc_2$ occurring in $\slc$ contains a universally quantified variable. 
The $\eQSH$-admissible $\SLA$ fragment $\SLfrag_1$ is given by
\begin{align*}
	\sla \quad \rightarrow \quad 
	&\preda \\
	&\mid~ \sla \wedge \sla  \\
	&\mid~ \sla \vee \sla  \\
	&\mid~ \exists\xx\colon \sla \\
	&\mid~ \sla \sepcon \sla \\
	&\mid~ \slsingleton{\xx}{\tuplenotation{\xy_1,\ldots,\xy_k}}  \sepimp \sla ~. 
\end{align*} 
The $\QSH$-admissible $\SLA$ fragment $\SLfrag_2$ is given by
\begin{align*}
	\slb \quad \rightarrow \quad 
	&\preda \\
	&\mid~ \slb \wedge \slb  \\
	&\mid~ \slb \vee \slb  \\
	&\mid~ \exists\xx\colon \slb \\
	&\mid~ \slb \sepcon \slb ~.
\end{align*} 
\begin{lemma}
	\label{lem:prenex_eqsh}
	Every $\sla \in \SLfrag_1$ is equivalent to a formula $\slb = \exists \xx_1,\ldots,\xx_n\colon \slc$ for some $n \in \Nats$ with $\slc \in \SLfrag_1$ quantifier-free.
\end{lemma}
\begin{proof}
	By induction on $\sla$. For the base case $\sla =  \preda$ we choose $n=0$ and have nothing to show. The cases $\wedge,\vee,\exists$ are standard. For the remaining cases, we reason as follows: As the induction hypothesis assume that for some arbitrary, but fixed, $\sla_1,\sla_2 \in \SLfrag_1$ there are $\slb_1 = \exists \xx_{1},\ldots,\xx_{n}\colon \slc_1$ and  $\slb_2 = \exists \xy_{1},\ldots,\xy_{m}\colon \slc_2$ with $\slc_1,\slc_2 \in \SLfrag_1$ quantifier-free and $\slb_1 \equiv\sla_1$ and $\slb_2 \equiv\sla_2$.  Furthermore, assume without loss of generality that $\xx_{1},\ldots,\xx_{n}$ do not occur in $\slc_2$ and that $\xy_{1},\ldots,\xy_{m}$ do not occur in $\slc_1$. \\ \\
	\noindent
	\emph{The case $\sla = \sla_1\sepcon\sla_2$.} For every $(\sk,\hh) \in \States$, we have
	\begin{align*}
		&(\sk, \hh) \mmodels \sla_1 \sepcon \sla_2 \\%
		\text{iff}\quad & \text{there are $\hh_1,\hh_2$ with $\hh_1 \sepcon\hh_2 =\hh$ such that} \\
		& (\sk, \hh_1) \mmodels \sla_1~\text{and} ~(\sk, \hh_2) \mmodels \sla_2 \\
		\text{iff}\quad & \text{there are $\hh_1,\hh_2$ with $\hh_1 \sepcon\hh_2 =\hh$ such that} \\
		&\text{there are $v_1,\ldots,v_n,w_1,\ldots,w_m$ such that} \\
		& (\sk\statesubst{\xx_1}{v_1}\ldots\statesubst{\xx_n}{v_n}, \hh_1) \mmodels \slc_1
		~\text{and} ~(\sk\statesubst{\xy_1}{w_1}\ldots\statesubst{\xy_m}{w_m}, \hh_2) \mmodels \slc_2 
		\tag{by I.H.}\\
		\text{iff}\quad &\text{there are $v_1,\ldots,v_n,w_1,\ldots,w_m$ such that} \\ 
		&\text{there are $\hh_1,\hh_2$ with $\hh_1 \sepcon\hh_2 =\hh$ such that} \\
		& (\sk\statesubst{\xx_1}{v_1}\ldots\statesubst{\xx_n}{v_n}, \hh_1) \mmodels \slc_1
		~\text{and} ~(\sk\statesubst{\xy_1}{w_1}\ldots\statesubst{\xy_m}{w_m}, \hh_2) \mmodels \slc_2 
		\\
		\text{iff}\quad &\text{there are $v_1,\ldots,v_n,w_1,\ldots,w_m$ such that} \\ 
		& (\sk\statesubst{\xx_1}{v_1}\ldots\statesubst{\xx_n}{v_n}\statesubst{\xy_1}{w_1}\ldots\statesubst{\xy_m}{w_m}, \hh) \mmodels \slc_1 \sepcon \slc_2 
		\tag{variables do not overlap}\\
		\text{iff}\quad &(\sk,\hh)\models \exists v_1,\ldots,w_n,w_1,\ldots,w_m\colon  
		\slc_1 \sepcon \slc_2~.
	\end{align*}
	\emph{The case $\sla =\slsingleton{\xx}{\tuplenotation{\xz_1,\ldots,\xz_k}}  \sepimp \sla_1$.} We assume without loss of generality that $\xx \not \in \{\xx_1,\ldots,\xx_n\}$ and $\{\xx_1,\ldots,\xx_n\} \cap \{\xz_1 \ldots \xz_k\} = \emptyset$.
	 Let $(\sk,\hh) \in \States$. We distinguish the cases $\sk(\xx) \in \dom{\hh}$ and $\sk(\xx) \not\in \dom{\hh}$. If $\sk(\xx) \in \dom{\hh}$, we have 
	 \begin{align*}
	    &(\sk, \hh) \mmodels \slsingleton{\xx}{\tuplenotation{\xz_1,\ldots,\xz_k}}  \sepimp \sla_1\\%
	    \text{iff}\quad & \true \\
	    \text{iff}\quad &(\sk, \hh) \mmodels \exists \xx_{1},\ldots,\xx_{n} \colon \slsingleton{\xx}{\tuplenotation{\xz_1,\ldots,\xz_k}}  \sepimp \slc_1~.
	    \tag{variables do not overlap}
	 \end{align*}
    If $\sk(\xx) \not\in \dom{\hh}$, we have
    \begin{align*}
    	&(\sk, \hh) \mmodels \slsingleton{\xx}{\tuplenotation{\xz_1,\ldots,\xz_k}}  \sepimp \sla_1\\%
    	\text{iff}\quad & (\sk, \hh\sepcon \slsingleton{\sk(\xx)}{\tuplenotation{\sk(\xz_1),\ldots,\sk(\xz_k)}} ) \mmodels \sla_1\\
    	\text{iff}\quad & (\sk, \hh\sepcon \slsingleton{\sk(\xx)}{\tuplenotation{\sk(\xz_1),\ldots,\sk(\xz_k)}} ) \mmodels \exists \xx_{1},\ldots,\xx_{n}\colon \slc_1
    	\tag{by I.H.} \\
    	\text{iff}\quad &\text{there are $v_1,\ldots,v_n$ such that}  \\
    	& (\sk\statesubst{\xx_1}{v_1}\ldots\statesubst{\xx_n}{v_n}, \hh\sepcon \slsingleton{\sk(\xx)}{\tuplenotation{\sk(\xz_1),\ldots,\sk(\xz_k)}} ) \mmodels  \slc_1 
    	\tag{variables do not overlap}\\
    	\text{iff}\quad &\text{there are $v_1,\ldots,v_n$ such that}  \\
    	& (\sk\statesubst{\xx_1}{v_1}\ldots\statesubst{\xx_n}{v_n}, \hh ) \mmodels \slsingleton{\xx}{\tuplenotation{\xz_1,\ldots,\xz_k}}  \sepimp  \slc_1
    	\tag{variables do not overlap} \\
    	\text{iff}\quad & (\sk, \hh ) \mmodels \exists \xx_1,\ldots,\xx_n \colon \slsingleton{\xx}{\tuplenotation{\xz_1,\ldots,\xz_k}}  \sepimp  \slc_1~.
    \end{align*}
    This completes the proof.
\end{proof}
\begin{lemma}
	\label{lem:prenex_qsh}
	Every $\sla \in \SLfrag_2$ is equivalent to a formula $\slb = \exists \xx_1,\ldots,\xx_n\colon \slc$ for some $n \in \Nats$ with $\slc \in \SLfrag_2$ quantifier-free.
\end{lemma}
\begin{proof}
	Analogous to the proof of \Cref{lem:prenex_eqsh}.
\end{proof}
Now let $\sla \in \eQSH$ and $\slb \in \QSH$. By \Cref{lem:prenex_eqsh} and \Cref{lem:prenex_qsh} there are $\sla' =\exists \xx_1,\ldots,\xx_n\colon \slc_1$ and $\slb' = \exists \xy_1,\ldots,\xy_m\colon \slc_2$ with $\slc_1 \in \SLfrag_1$,$\slc_2 \in \SLfrag_2$ quantifier-free and $\sla \equiv \sla'$ and $\slb \equiv \slb'$. Notice that $\slc_2$ does not contain $\sepimp$. We may without loss of generality assume that $\xx_1,\ldots,\xx_n$ do not occur in $\slc_2$ and that $\xy_1,\ldots,\xy_m$ do not occur in $\slc_1$. Hence, we get
\begin{align*}
	& \sla \eentails \slb \\
	\text{iff} \quad &\sla \wedge \neg\slb~\text{is unsatisfiable} \\
	\text{iff} \quad & \exists \xx_1,\ldots,\xx_n\colon \slc_1 \wedge \neg (\exists \xy_1,\ldots,\xy_m\colon \slc_2)~\text{is unsatisfiable}
	\tag{by above reasoning} \\
	\text{iff} \quad & (\exists \xx_1,\ldots,\xx_n\colon \slc_1) \wedge (\forall  \xy_1,\ldots,\xy_m\colon \neg \slc_2)~\text{is unsatisfiable} \\
	\text{iff} \quad & \exists \xx_1,\ldots,\xx_n \colon \forall  \xy_1,\ldots,\xy_m\colon \slc_1  \wedge  \neg \slc_2~\text{is unsatisfiable}~,
	\tag{standard prenexing}
\end{align*}
since $\slc_1 \wedge \neg \slc_2$ is quantifier free and since none of $\xy_1,\ldots,\xy_m$ occur in an instance of $\sepimp$ occurring in $\slc_1 \wedge \neg \slc_2$, the claim follows.

\end{proof}
\end{document}